\documentclass[conference,compsoc]{IEEEtran}

\usepackage{amsmath,amsfonts, amsthm, amssymb, bm}
\usepackage[noend]{algorithmic}
\usepackage{algorithm}
\usepackage{array}
\usepackage[caption=false,font=normalsize,labelfont=sf,textfont=sf]{subfig}
\usepackage{textcomp}
\usepackage{stfloats}
\usepackage{url}
\usepackage{verbatim}
\usepackage{graphicx}
\usepackage{cite}
\usepackage{mathtools}
\usepackage{multirow, multicol}
\usepackage[font=small,skip=2pt]{caption}
\usepackage{enumitem}
\usepackage[table]{xcolor}
\usepackage{xcolor}

\newcommand\Tstrut{\rule{0pt}{2.7ex}}         
\newcommand\Bstrut{\rule[-1.4ex]{0pt}{0pt}}   

\newtheorem{lemma}{Lemma}
\newtheorem{theorem}{Theorem}
\newtheorem{corollary}{Corollary}
\theoremstyle{definition}
\newtheorem{example}{Example}
\newtheorem{definition}{Definition}

\newtheorem{remark}{Remark}

\newcommand{\cA}{\mathfrak{A}}

\newcommand{\E}{\mathcal{E}}
\newcommand{\Q}{\mathcal{Q}}
\newcommand{\R}{\mathcal{R}}
\newcommand{\T}{\mathcal{T}}

\renewcommand{\O}{\mathcal{O}}

\newcommand{\kl}{k_\ell}
\newcommand{\il}{i_\ell}
\newcommand{\ilp}{i'_\ell}
\newcommand{\jkl}{j_{k_\ell}}
\newcommand{\sbl}{\texttt{sb}_\ell}

\newcommand{\rand}{\overset{R}{\gets}}

\renewcommand{\H}{\mathsf{H}}
\renewcommand{\S}{\mathsf{S}}

\newcommand{\Cq}{{\mathsf{C_Q}}}
\newcommand{\Ce}{{\mathsf{C_E}}}
\newcommand{\B}{{\mathsf{B}}}

\newcommand{\bq}{\bm{q}}

\newcommand{\br}{\bm{r}}
\newcommand{\bx}{\bm{x}}

\newcommand{\negl}{{\sf{negl}}}
\newcommand{\poly}{{\sf{poly}}}

\newcommand{\et}{\textit{et al.}}

\begin{document}

\title{TreePIR: Efficient Private Retrieval of Merkle Proofs via Tree Colorings\\ with Fast Indexing and Zero Storage Overhead\vspace{-20pt}}

\author{Son Hoang Dau, Quang Cao, Rinaldo Gagiano, Duy Huynh, Xun Yi,\\ 
Phuc Lu Le, 
Quang-Hung Luu,
Emanuele Viterbo,
Yu-Chih Huang,\\
Jingge Zhu,
Mohammad M. Jalalzai, and Chen Feng}
\vspace{-10pt}

\maketitle

\begin{abstract}
A Batch Private Information Retrieval (batch-PIR) scheme allows a client to retrieve multiple data items from a database without revealing them to the storage server(s).  
Most existing approaches for batch-PIR are based on batch codes, in particular, probabilistic batch codes (PBC) (Angel \et~S\&P'18), which incur large storage overheads.
In this work, we show that \textit{zero} storage overhead is achievable for tree-shaped databases. In particular, we develop \textit{TreePIR}, a novel approach tailored made for private retrieval of the set of nodes along an arbitrary \textit{root-to-leaf path} in a Merkle tree with no storage redundancy.
This type of trees has been widely implemented in many real-world systems such as Amazon DynamoDB, Google's Certificate Transparency, and blockchains. Tree nodes along a root-to-leaf path forms the well-known \textit{Merkle proof}. 
TreePIR, which employs a novel tree coloring, outperforms PBC, a fundamental component in state-of-the-art batch-PIR schemes (Angel \et~S\&P'18, Mughees-Ren~S\&P'23, Liu \et~S\&P'24), in all metrics, achieving $3\times$ lower total storage and $1.5$-$2\times$ lower computation and communication costs.
Most notably, TreePIR has $8$-$160\times$ lower setup time and its \textit{polylog}-complexity indexing algorithm is $19$-$160\times$ faster than PBC for trees of $2^{10}$-$2^{24}$ leaves.
\end{abstract}


\section{Introduction}
\label{sec:intro}
A \textit{Merkle tree} is a binary tree 
in which each node is the (cryptographic) hash of the concatenation of the contents of its child nodes~\cite{Merkle1988}. 
A Merkle tree can be used to represent a large number of data items in a way that not only guarantees data integrity but also allows a very efficient membership verification, which can be performed with complexity $\Theta(\log(n))$ where $n$ is the number of items represented by the tree. 
More specifically, the membership verification of an item uses its \textit{Merkle proof} defined as follows: the Merkle proof 
for the item 
corresponding to a leaf node consists of $\Theta(\log(n))$ hashes stored at the siblings of the nodes in the path from that leaf node to the root. 
Due to their simple construction and powerful features, Merkle trees have been widely used in practice, e.g., for data synchronization in Amazon DynamoDB~\cite{AmazonDynamo2007}, for certificates storage in Google's Certificate Transparency~\cite{GoogleCT, CTwork21, rfc9162}, and states/transactions storage in blockchains~\cite{Dryja_UTREEXO_2019,BaileySankagiri_FC21,HederaWP1,HederaStateTree,wood2014ethereum,Hopwood2023}.

\subsection{The Problem of Interest} 
\vspace{-5pt}

In this work, we investigate 
the problem of \textit{private retrieval of Merkle proofs} from a Merkle tree described as follows. Suppose that 
$n$ items $(T_i)_{i=1}^n$ are represented 
by a Merkle tree of height $h$ and that the Merkle root 
is made public. We also assume that the Merkle tree is collectively stored at one or more servers. The goal is to design an \textit{efficient} retrieval scheme that allows a client who owns an item $T_i$ to retrieve its corresponding Merkle proof (in order to verify if $T_i$ indeed belongs to the tree) \textit{without} revealing $i$ to the servers that store the Merkle tree. This problem was originally introduced in the work of Lueks and Goldberg~\cite{LueksGoldberg2015} and Kale \et~\cite{kales2019} in the context of Certificate Transparency~\cite{GoogleCT, CTwork21}.

\subsection{Applications} 
\vspace{-5pt}

We demonstrate below 
potential applications that motivate the problem of private retrieval of Merkle proofs. 

\textbf{Certificate Transparency} (CT) was initiated by Google in 2012 to provide transparency and verifiability for website certificate issuance and has become an Internet security standard~\cite{GoogleCT, CTwork21,rfc6962,rfc8446}, adopted in major internet browsers including Chrome and Safari. A \textit{log} in a CT system is an append-only Merkle tree of certificates and precertificates (for more details, see~\cite[p.~17]{rfc9162}). At the time of writing, large CT logs such as Cloudfare's Nimbus2024 and DigiCert's Yeti2024 contain hundreds of million or even more than a \textit{billion} certificates (valid 1/1/2024-1/1/2025) for the case of Google's Xenon2024~\cite{MerkleTown}. 
However, in CT, an HTTPS client must fetch a Merkle proof to verify the validity of a certificate either via an auditor or by itself, which raises an immediate privacy concern: the log or the auditor can track the websites that the client is visiting or has previously visited (see~\cite[Secs. 10.5, 11.2]{ietf-trans-gossip-05}, \cite{EskandarianMesseriBonneuBoneh2017,kales2019,Kwon_etal_2023}).
Our solution would allow the client to efficiently verify a website certificate without revealing which website it is visiting. 



\textbf{A blockchain's stateless (or state-compact) client}, as opposed to a full node, does not store the entire state of the chain but can still verify the correctness of its operation by downloading state proofs. The chain's state is usually large. For example, Bitcoin's state is represented by a fast-growing list (currently at 170 \textit{million}\footnote{\url{https://www.blockchain.com/explorer/charts/utxo-count}}) of unspent transactions outputs (UTXOs)\footnote{Most blockchains adopt the UTXO model (similar to cash transactions) or the account model (resembles how bank accounts work).}. There have been several proposals to organize the Bitcoin UTXOs list into a single Merkle tree~\cite{Dryja_UTREEXO_2019} or a forest of perfect Merkle trees as in UTREEXO~\cite{BaileySankagiri_FC21} to support stateless clients.  A \textit{bridge} server (see~\cite{Dryja_UTREEXO_2019, BaileySankagiri_FC21}), which stores all the state trees at all times, is responsible for producing the Merkle proof for each UTXO. A stateless client, who has the root of such a tree, can verify the validity of a new transaction by retrieving a Merkle proof from bridge servers for each of its UTXOs, ensuring that the transaction spends from a set of valid UTXOs. Merkle trees are actually used as state trees storing all UTXOs in other chains such as Hedera~\cite{HederaStateTree,HederaWP1} and Neptune~\cite{NeptuneWP_2021}. Our TreePIR 
will allow the stateless clients in such blockchains to efficiently download the Merkle proof of a UTXO without revealing them to the bridge servers. This ensures that even if the bridge server is compromised or under surveillance, the attacker cannot pinpoint or target specific users based on their transaction verification requests.

\begin{figure}
    \centering    
    \includegraphics[scale=0.95]{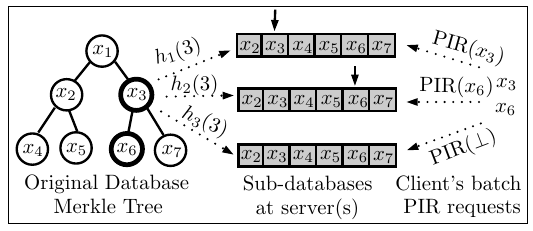}
    \caption{An illustration of the \textit{Probabilistic Batch Code} (PBC) approach~\cite{angel2018} for private retrieval of a Merkle proof $(x_3,x_6)$ in a (swapped) Merkle tree of height $h=2$. There are $1.5h = 3$ sub-databases and \textit{each tree node $x_i$ is replicated three times} and stored at three sub-databases indexed by the hash functions $h_1, h_2, h_3$. To privately retrieve $(x_3,x_6)$ (the root $x_1$ is publicly known), the client must send $1.5h=3$ PIR queries to three sub-databases, one of which is useless (PIR($\perp$)) but required to guarantee the privacy.}
    \label{fig:pbc_toy}
\end{figure}

\textbf{Merkle Tree Ladder} mode of operation was recently proposed by researchers from VeriSign to allow a signer to prove to a verifier that a particular message in a growing set of messages has been signed, without signing every message~\cite{FreglyHarveyKaliskiSheth_CTRSA_2023}. The trick is to let the signer create a forest of perfect Merkle trees with leaves storing the hashes of the messages and then sign the root hashes only. When a verifier need to verify a particular message, the signer simply sends the corresponding Merkle proof (including the signed root hash) to the verifier, who then can verify that the message is indeed included in that Merkle tree with a properly signed root. TreePIR will 
provide an efficient mechanism for the verifier to verify a message without revealing it to the signer.  

\begin{figure}
    \centering 
    \includegraphics{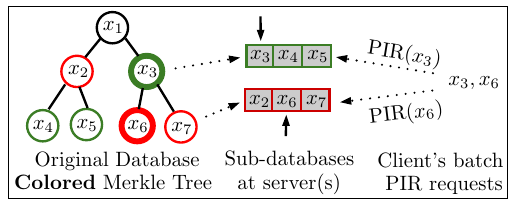}
    \caption{An illustration of \textit{our coloring-based approach} for batch private retrieval of a Merkle proof $(x_3,x_6)$ in a (swapped) Merkle tree of height $h=2$. \textit{Each tree node $x_i$ is stored in exactly one sub-database} according to its assigned color, hence incurring no storage redundancy. As the result, our scheme uses only $h=2$ sub-databases, each of which is of size \textit{half} of that in PBC~\cite{angel2018}, and the client only computes and sends $h=2$ PIR queries.}
    \label{fig:coloring_toy}
\end{figure}

\subsection{Batch-PIR Generic Solution} 
\label{subsec:batchPIR}
\vspace{-5pt}

The problem of private retrieval of a Merkle proof from a Merkle tree can be solved by any \textit{batch Private Information Retrieval} (batch-PIR) scheme, in which a client retrieves a \textit{batch} of $h$ data items $\{x_{k_1},x_{k_2},\ldots,x_{k_h}\}$ from the database $(x_1,x_2,\ldots,x_N)$ stored at one or more servers \textit{privately}, i.e., without revealing $\{k_1,k_2,\ldots,k_h\}$ to the server(s). Ordinary PIR schemes, which correspond to the case of batch size one ($h=1$), were first introduced in the seminal work of Chor-Goldreich-Kushilevitz-Sudan \cite{Chor1995}, followed by Kushilevitz-Ostrovsky~\cite{kushilevitz1997}, and have been extensively studied in the literature. Batch codes (BC)~\cite{ishai2004} and combinatorial batch codes (CBC)~\cite{stinson2009} are powerful primitives that allow one to conveniently construct a batch-PIR scheme from an ordinary PIR scheme, leveraging the plethora of PIR schemes existing in the literature. 
\textit{In a nutshell}, a BC/CBC constructs $m$ \textit{sub-databases} from the original one so that any batch of $h$ items can be privately retrieved by sending \textit{one} (ordinary) PIR query to \textit{every} sub-database. 
A BC/CBC with \textit{fewer sub-databases} of \textit{smaller sizes} lead to a batch-PIR with faster server and client computation times and less communication.

Traditional (deterministic) batch codes have small sub-database sizes (hence, lower server computation time) but use a large number of them, resulting in high communication costs. \textit{Probabilistic batch codes} (PBC), introduced by Angel-Chen-Laine-Setty~\cite{angel2018}, 
is a probabilistic relaxation of CBC that offers a better trade-off between the number of sub-databases and their sizes. Similar to a BC/CBC, a PBC can be combined with a PIR scheme to create a batch-PIR.
More specifically, a PBC uses $w$ independent hash functions $h_1,\ldots,h_w$ to distribute $x_i$'s in the original database to $w$ among $m$ sub-databases 
indexed by $h_1(i),\ldots,h_w(i)$, $i=1,\ldots,n$. Hence, the total storage overhead is $wN$.
In its typical setting, $w = 3$ and $m = 1.5h$. 
The client then uses Cuckoo hashing~\cite{pagh2004} to find a one-to-one mapping between the $h$ items to be retrieved and some $h$ sub-databases, and sends one PIR query to each of the $h$ sub-databases to privately retrieve these items. It also sends a dummy (useless) PIR query to each of the remaining $m-h$ sub-databases (for privacy). PBC 
forms an essential component in state-of-the-art batch-PIR schemes~\cite{angel2018,mughees2022vectorized,Pirana2024}. 

\begin{table}[htb!]
\centering
\setlength{\tabcolsep}{0.5pt}
    \caption{A comparison of TreePIR and related batch-PIR approaches 
    when applied to a perfect Merkle tree of height $h$ and $N\hspace{-2pt}=\hspace{-2pt}2^{h+1}\hspace{-3pt}-\hspace{-2pt}2$ nodes. 
    Note that the number of sub-databases and their sizes determine the storage, computation, and communication costs of the scheme (see Section~\ref{subsec:ProblemDescription}).}
    \label{table:BatchPIRcomparison}
\begin{tabular}{p{1.6cm}|p{1.6cm}|p{1.7cm}|p{1.65cm}|p{1.7cm}}
    \hline
    \Tstrut
    \textbf{Approaches} & \textbf{Total storage} & \textbf{Number of sub-databases} 
    & \textbf{Sub-database size} & \textbf{Indexing complexity}\\
    \hline \cline{1-5}
    \Tstrut
    Subcube code \cite{ishai2004}  ($\ell\geq 2$) & $N h^{\log_2\frac{\ell + 1}{\ell}}$ & $h^{\log_2{(\ell + 1)}}$ & $\dfrac{N}{h^{\log_2\ell}}$ & $\Omega\big(h^{\log_2(\ell + 1)}$ $\times N\big)$\Bstrut\\ 
    \hline
    \Tstrut
    Balbuena\hspace{-1pt} graph \cite{rawat2016} & $2N$ & $2(h^3 - h)$ & $\dfrac{N}{h^3 - h}$ & $\Omega(N)$\Bstrut\\ 
    \hline
    \Tstrut
    CBC\qquad\qquad \cite[Thm. \hspace{-2pt}2.7]{stinson2009} & $hN\hspace{-2pt} -\hspace{-2pt} (h\hspace{-2pt} - \hspace{-2pt}1) \times {m  \choose {h - 1}}$ & $m$ \qquad \qquad$\hspace{-2pt}\big(\hspace{-2pt}{m \choose {h - 1}}\hspace{-2pt} \leq\hspace{-2pt} \frac{N}{h - 1}\hspace{-2pt}\big)\hspace{-2pt}$ & $\frac{hN}{m}\hspace{-2pt} -\hspace{-2pt} \frac{{{m - 1} \choose {h\hspace{-1pt} -\hspace{-1pt} 2}}}{m \hspace{-1pt}-\hspace{-1pt} h + 1}$ & 
    $\Omega(Nm)$ $\Omega(h)\hspace{-2pt} \text{ if } m\hspace{-2pt}=\hspace{-2pt}h$\Bstrut\\
    \hline
    \Tstrut
    PBC~\cite{angel2018} & $3N$ & $1.5h$ & $2N/h$ & $\Omega(N)$\\
    \hline
    \Tstrut
    \textbf{TreePIR} & $N$ & $h$ & $N/h$ & $\O(h^3)$\\
    \hline \cline{1-5}
\end{tabular}
\vspace{5pt}
    \vspace{-10pt}
\end{table}

\subsection{Our Proposal}
\vspace{-5pt}

In this work, we first show that Probabilistic Batch Codes, although provide an elegant \textit{generic} solution for the batch-PIR problem, turn out to have significant drawbacks when applied to the problem of private retrieval of Merkle proofs on Merkle trees. Especially, PBC's indexing strategies require the client to download an \textit{index} that is \textit{almost as large as the entire Merkle tree} for large trees, rendering it impractical. 
We then propose \textit{TreePIR}, a \textit{storage-optimal} solution based on a novel concept of \textit{balanced ancestral coloring} of a binary tree, in which tree nodes are assigned one of the $h$ colors so that nodes with ancestor-descendant relationship (i.e. nodes belonging to the same Merkle proof) have distinct colors \textit{and} each color is assigned to almost the same number of nodes. See Fig.~\ref{fig:toy} for examples of balanced ancestral coloring for perfect trees of height $h=1,2,3$. 

A balanced ancestral coloring of a height-$h$ Merkle tree partitions its nodes into $h$ parts of almost equal sizes so that every root-to-leaf path intersects with each part in exactly one node. The nodes in the same part have the same color and form a sub-database. Hence, by sending a single PIR query to each sub-database, all nodes in a root-to-leaf path can be privately retrieved. Having a unique color, each node appears in exactly one sub-database, leading to the minimum total storage. As the total storage is the product of the number of sub-databases (proportional to the communication cost) and their averaged size (proportional to the server computation time), TreePIR achieves the best trade-off between communication and computation costs 
among \textit{all} batch-code-based approaches (see Table~\ref{table:BatchPIRcomparison}, Column 2). Especially, TreePIR 
possesses a \textit{fast indexing} with \textit{polylog} complexity, allowing the client to find all required PIR indices in \textit{milliseconds} for trees with \textit{billions} of nodes.

A toy example of our approach for a Merkle tree of height $h=2$ is given in Fig.~\ref{fig:coloring_toy}. For convenience, we consider instead a \textit{swapped} Merkle tree (by swapping sibling nodes), in which a Merkle proof corresponds to a root-to-leaf path. Note that the tree root is publicly known. 
In the typical setting of PBC~\cite{angel2018, mughees2022vectorized}, there are $1.5h$ sub-databases 
and each tree node is replicated and stored in three different sub-databases. 
By contrast, TreePIR 
requires only $h$ sub-databases, each of which has size \textit{half} of that of the PBC. TreePIR client computes and sends only $h$ PIR queries as opposed to $1.5h$ queries in PBC. 

\begin{figure}[t]
\centering
\includegraphics[scale=0.95]{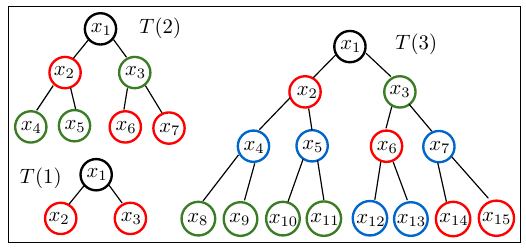}
\caption{\textit{Balanced ancestral colorings} of perfect binary trees $T(h)$ with $h = 1,2,3$. Nodes that are ancestor and descendant have different colors, and the color classes have sizes almost equal.}
\label{fig:toy}
\end{figure}

Our approach is orthogonal to the optimization techniques proposed by Lueks-Goldberg~\cite{LueksGoldberg2015}, Kale \et~\cite{kales2019}, Mughees-Ren~\cite{mughees2022vectorized}, and Liu \et~\cite{Pirana2024}. Employing TreePIR 
will further improve their schemes (see Section~\ref{subsec:comparisons}).

\subsection{Contributions}
\vspace{-5pt}

Our main contributions are summarized below.
\begin{itemize}
\item We propose \textit{TreePIR}, an efficient approach to privately retrieve a Merkle proof based on the novel concept of \textit{balanced ancestral coloring} of binary trees. 
TreePIR 
outperforms the state-of-the-art approach based on PBC~\cite{angel2018}, with $3\times$ less storage, $1.5\times$ less communication, $1.5$-$2\times$ 
faster server computation/client query generation times, $8$-$60\times$ faster setup, and $19$-$160\times$ faster indexing (even when ignoring PBC's index download). 
\item TreePIR 
requires no replication of the tree nodes, hence achieving \textit{optimal} total storage, i.e. zero storage redundancy. Existing works~\cite{ishai2004, Balbuena2009, stinson2009, angel2018} all require large redundancy, e.g. $200$\% in~\cite{angel2018} (see Table~\ref{table:BatchPIRcomparison}).
\item We develop a \textit{fast indexing} algorithm for TreePIR with $\O(h^3)$ space/time complexity, which is a huge improvement from the complexity $\Omega(2^h)$ of PBC~\cite{angel2018}. It finds required sub-indices in a tree of 64 \textit{billion} leaves in under a millisecond. 
\item TreePIR 
hinges on our discovery of a \textit{necessary and sufficient condition} for the existence of ancestral colorings and an efficient \textit{divide-and-conquer} algorithm that generates a coloring in $\Theta(n\log\log(n))$ operations on the perfect binary tree of $n$ leaves. 
Trees with \textit{billions} of nodes can be colored in \textit{minutes}.
\end{itemize}


The paper is organized as follows. Basic concepts 
are discussed in Section~\ref{sec:preliminaries}. Section~\ref{sec:OurProposalRelatedWorks} presents our approach to the problem of private retrieval of Merkle proofs based on ancestral tree colorings and comparisons to related works. We develop in Section~\ref{sec:Algorithm} a necessary and sufficient condition for the existence of an ancestral coloring and an efficient tree coloring algorithm. 
Experiments and evaluations are discussed in Section~\ref{sec:ExpandEval}. We conclude the paper in Section~\ref{sec:Conclusion}.

\section{Preliminaries}
\label{sec:preliminaries}
\vspace{-5pt}

For basic concepts on graphs, trees, coloring, Merkle tree and Merkle proof, please refer to Appendix~\ref{app:graph_tree}. We use $[n]$ to denote the set $\{1,2,\ldots,n\}$.

\begin{figure}[t]
\centering
\includegraphics[scale=0.93]{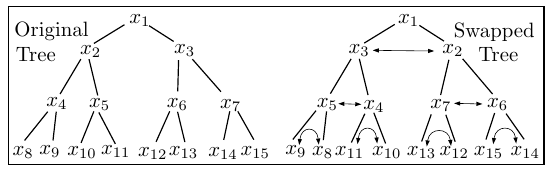}
\caption{Illustration of an original Merkle tree and a \textit{swapped} tree, in which sibling nodes, e.g., $x_2$ and $x_3$, are swapped. A Merkle proof in the original Merkle tree (consisting of the siblings of nodes along a root-to-leaf path) forms a root-to-leaf path in the swapped tree. \textit{We henceforth only work with swapped trees}.}
\label{fig:merkle_swapped}
\end{figure}

\subsection{Swapped Merkle Tree}
\label{subsec:MTMP}
\vspace{-5pt}

To facilitate the discussion using tree coloring, we consider the so-called \textit{swapped} Merkle tree, which is obtained from a Merkle tree by swapping the positions of every node and its sibling (the root node is the sibling of itself and stays unchanged). 
The nodes in a Merkle proof in the original tree now form a \textit
{root-to-leaf path} in the swapped Merkle tree, which is more convenient to handle. 
For example, in Fig.~\ref{fig:merkle_swapped}, the nodes in the Merkle proof $(x_{11},x_4,x_3)$ form a root-to-leaf path (excluding the root) in the swapped tree.

\subsection{Batch Private Information Retrieval}
\label{subsec:PIR}
\vspace{-5pt}

Private Information Retrieval (PIR) was first introduced by Chor-Goldreich-Kushilevitz-Sudan~\cite{Chor1995} and has since become an important area in data privacy. In a PIR scheme, one or more servers store a database of $n$ data items $\{x_1,x_2,\ldots,x_N\}$ while a client wishes to retrieve an item $x_j$ without revealing the index $j$ to the server(s). Depending on the way the privacy is defined, we have \textit{information-theoretic} PIR (IT-PIR)~\cite{Chor1995} or \textit{computational} PIR (cPIR)~\cite{kushilevitz1997}. 
Some recent noticeable practical developments 
include SealPIR~\cite{angel2018, SealPIR2022}, which is based on Microsoft SEAL (Simple Encrypted Arithmetic Library)~\cite{SealPIR2022}, and MulPIR from a Google research team~\cite{AliMultPIR2019}.

The notion of batch-PIR extends that of PIR to the setting where the client wants to retrieve a \textit{batch} of items. We define below a computational batch-PIR scheme. The information-theoretic version can be defined similarly. The problem of privately retrieving a Merkle proof from a Merkle tree 
is a special case of the batch-PIR problem: the tree nodes form the database and the proof forms a batch.

\begin{definition}[batch-PIR]\label{def::bPIR}
A (single-server) $(n,h)$ {\em batch-PIR scheme} consists of three algorithms described as follows.
\begin{itemize}
    \item $\bq \leftarrow \Q(1^\lambda, n, B)$ is a randomized {\em query-generation} algorithm for the client. It takes as input the database size $n$, the security parameter $\lambda$, and a set $B\subseteq [n]$, $|B|=h$, and outputs the query $\bq$.
    \item $\br\leftarrow\R(1^\lambda, \bx, \bq)$ is a deterministic {\em response-generation} algorithm used by the server, which takes as input the database~$\bx$, the security parameter $\lambda$, the query $\bq$, and outputs the response $\br$. 
    \item $\{x_i\}_{i \in B}\leftarrow\E(1^\lambda, B, \br)$ is a deterministic {\em answer-extraction} algorithm run by the client, which takes as input the security parameter $\lambda$, the index set $B$, the response $\br$ from the server, and reconstructs $\{x_i\}_{i\in B}$.
\end{itemize}
\end{definition}
\vspace{-5pt}
\begin{definition}[Correctness]\label{def::correctness}
A batch-PIR scheme is {\em correct} if the client can retrieve $\{x_i\}_{i\in B}$ from the response given that all parties correctly follow the protocol. More formally, for every $\bx$ and $B\subset [n]$, $|B|=h$, it holds that
$\Pr[\bq \leftarrow \Q(1^\lambda, n, B), \br\leftarrow\R(1^\lambda, \bx, \bq) \colon \{x_i\}_{i\in B}\hspace{-2pt}=\hspace{-2pt}\E(1^\lambda, B, \br)]\hspace{-2pt} =\hspace{-2pt} 1$.
\end{definition}

We denote by $\lambda\hspace{-2pt}\in\hspace{-2pt} \mathbb{N}$ the security parameter, e.g., $\lambda\hspace{-2pt} =\hspace{-2pt} 128$, and $\negl(\lambda)$ the set of \textit{negligible functions} in $\lambda$.
A positive-valued function $\varepsilon(\lambda)$ belongs to $\negl(\lambda)$ if for every $c\hspace{-2pt}>\hspace{-2pt}0$, there exists a $\lambda_0\in \mathbb{N}$ such that $\varepsilon(\lambda)<1/\lambda^c$ for all $\lambda > \lambda_0$.

\begin{definition}[Computational Privacy]\label{def::privacy}
The privacy of a PIR scheme requires that a computationally bounded server cannot distinguish, with a non-negligible probability, between the distributions of any two queries $\bq(B)$ and $\bq(B')$. That is, for every database $\bx$ of size $n$, for every $B, B'\subset [n]$, $|B|=|B'|=h$, and for every probabilistic polynomial time algorithm $\cA$, it is required that (following Cachin \textit{et al.}~\cite{cachin1999})
\begin{multline*}
\big|\Pr[\bq(B) \rand \Q(1^\lambda,n,B) \colon \cA(1^\lambda,n,h,\bq(B))=1] -\\ \Pr[\bq(\hspace{-1pt}B'\hspace{-1pt})\hspace{-2pt} \rand \hspace{-2pt} \Q(1^\lambda,n,B') \hspace{-1pt}\colon \hspace{-1pt} \cA(1^\lambda,n,h,\bq(B'))\hspace{-2pt}=\hspace{-2pt}1]\big|\hspace{-2pt} \in\hspace{-2pt} \negl(\lambda).
\end{multline*}
\end{definition}

Note that in our approach to batch-PIR, which is similar to most related works~\cite{ishai2004, angel2018, mughees2022vectorized}, one splits the database into independent sub-databases and applies an existing PIR on each. Hence, the privacy of our batch-PIR reduces to that of the underlying PIR in a straightforward manner. 

\section{Private Retrieval of Merkle Proofs}
\label{sec:OurProposalRelatedWorks}
\vspace{-5pt}

We first describe the problem of private retrieval of Merkle proofs and then propose TreePIR, an efficient solution 
based on the novel concept of tree coloring. 
We also introduce a fast algorithm with no communication overhead to resolve the \textit{sub-index problem}, described as follows: after partitioning a height-$h$ tree into $h$ sub-databases $C_1,C_2,\ldots,C_h$ of tree nodes, one must be able to efficiently determine the sub-index of an arbitrary node $j$ in $C_i$, for every $i=1,2,\ldots,h$. We note that PBC~\cite{angel2018} indexing has complexity $\O(N)$ in time or space, 
where $N=2^{h+1}-2$ is the number of nodes in the tree (excluding the root). TreePIR 
only requires the client to perform $\O(h^3)$ operations. 
Finally, we explain how to use TreePIR to improve several existing schemes in the literature~\cite{LueksGoldberg2015,angel2018,mughees2022vectorized,Pirana2024}.

\subsection{Problem Description}
\label{subsec:ProblemDescription}
\vspace{-5pt}

Given a (perfect) Merkle tree stored at one or more servers, we are interested in designing an \textit{efficient} private retrieval scheme in which a client can send queries to the server(s) and retrieve an arbitrary Merkle proof without letting the server(s) know which proof is being retrieved.
As discussed in Sections~\ref{sec:intro} and~\ref{subsec:MTMP}, this problem is equivalent to the problem of designing an efficient private retrieval scheme for every root-to-leaf path in a perfect binary tree (with height $h$ and $n=2^h$ leaves).
An efficient retrieval scheme should have \textit{low} \textit{storage}, \textit{computation}, and \textit{communication overheads}.
The server(s) should \textit{not} be able to learn which path is being retrieved based on the received queries. 
This is an instance of the batch-PIR problem (Definition~\ref{def::bPIR}) in which the $N$ tree nodes form the database and the $h$ nodes in a root-to-leaf path form a (special) batch.

\textbf{Performance metrics for batch-PIR.} Suppose that there are $m$ sub-databases stored at one or more servers. Moreover, the client generates and sends one PIR query to each sub-database to retrieve $h$ items. Clearly, the number $m$ of sub-databases and their sizes (assuming they are all approximately $s$) are the main factors that decide the cost of running the batch-PIR. The total storage $ms$ represents the \textit{trade-off} between computation and communication.

More specifically, let $\S(s), \Cq(s), \Ce(s)$, and $\B(s)$ denote the server computation time, client query generation time, client extraction time, and communication cost per pair of query-response, respectively, of the \textit{underlying PIR scheme} on a database of size $s$. Then a batch-PIR with $m$ sub-databases of size $s$ has max/average server computation time $\S(s)$, total server computation time $m\S(s)$, client query generation time $m\Cq(s)$, client extraction time $h\Ce(s)$, and communication cost $m\B(s)$. The total storage is $ms$. Note that the \textit{max} server computation time is relevant  in the parallel setting (one server/thread per sub-database), whereas the \textit{total} server computation time is relevant in the sequential setting (one server handling all sub-databases sequentially). Here we assume all servers/threads have the same computation/storage/networking capacities for simplicity. TreePIR 
does extend to the more general heterogeneous setting.



\subsection{Our Solution}
\label{subsec:parallelApproach}
\vspace{-5pt}

To motivate 
TreePIR, let us start with the trivial 
scheme \textit{$h$-Repetition}, which uses $m=h$ sub-databases, each contains the \textit{entire} Merkle tree of size $s=N$. The client sends $h$ PIR queries to $h$ sub-databases to retrieve privately all nodes of a Merkle proof. 
This scheme is wasteful because each tree node is replicated $h$ times, leading to large sub-databases.
Existing approaches all require some levels of replication: $2\times$, $3\times$, $h\times$, and $h^{\log_2 \frac{\ell+1}{\ell}}\times$ in~\cite{Balbuena2009, angel2018, stinson2009}, and \cite{ishai2004}, respectively (Table~\ref{table:BatchPIRcomparison}). 
It turns out that to privately retrieve Merkle proofs, \textit{node replication is not needed}. 



\begin{figure}
    \centering \includegraphics[scale=1]{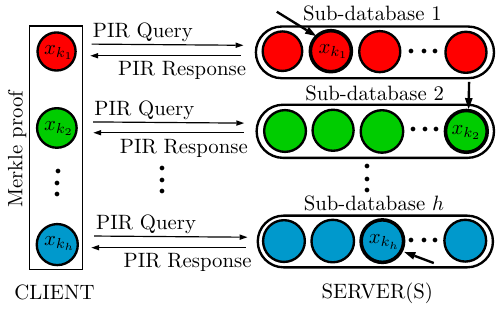}
    \caption{An illustration of our approach. First, the nodes of the (swapped) Merkle tree of height $h$ are partitioned into $h$ sub-databases, each corresponds to one color. The client runs $h$ PIR schemes independently on $h$ sub-databases to privately retrieve $h$ nodes of a Merkle proof. This is possible because our coloring ensures that each Merkle proof contains $h$ nodes of different colors.}
    \label{fig:coloring-based-PIR}
\end{figure}

Our proposal is to \textit{partition} the (swapped) Merkle tree into $h$ \textit{balanced} (disjoint) parts/sub-databases of sizes $\lfloor N/h \rfloor$ or $\lceil N/h \rceil$ (for brevity, we can set the sub-database size as $s=N/h$, ignoring the rounding), and let the client run a PIR scheme on $h$ independent sub-databases (see Fig.~\ref{fig:coloring-based-PIR}).
Note that partitioning the tree into $h$ sub-databases is equivalent to coloring its nodes with $h$ colors: two nodes have the same colors if any only if they belong to the same sub-database. To ensure that sending one PIR query to each sub-database is sufficient to retrieve all $h$ nodes, it is critical that \textit{every root-to-leaf path contains exactly one node of each color}. Equivalently, nodes in the same root-to-leaf path (having ancestor-descendant relationship) must have different colors. This motivates the new concept of \textit{ancestral coloring} defined below.
Compared to $h$-Repetition, the coloring-based solution has the same number of sub-databases but with $h\times$ smaller sizes, hence reducing the total storage and the server computation by a factor of $h$. 
\begin{definition}[Ancestral Coloring]
A (node) coloring of a rooted tree of height $h$ using $h$ colors $1,2,\ldots,h$ is referred to as an \textit{ancestral coloring} if it satisfies the Ancestral Property defined as follows. 
\begin{itemize}
    \item (\textbf{Ancestral Property}) Each color class $C_i$, which consists of tree nodes with color $i$, doesn't contain any pair of nodes $u$ and $v$ so that $u$ is an \textit{ancestor} of $v$. In other words, nodes that have ancestor-descendant relationship must have different colors. 
\end{itemize}    
An ancestral coloring is called \textit{balanced} if it also satisfies the Balanced Property defined below.
\begin{itemize}
    \item (\textbf{Balanced Property}) $\left||C_i|-|C_j|\right| \leq 1$, $\forall i,j\leq h$. 
\end{itemize}
\end{definition}
For brevity, we henceforth use $i$ to represent a tree node instead of $x_i$. Note that a trivial ancestral coloring is the coloring by layers, i.e., for a perfect binary tree, $C_1 = \{2,3\}$, $C_2 = \{4,5,6,7\}$, etc. The layer-based coloring, however, is not balanced. 
In Section~\ref{sec:Algorithm}, we will present a near-linear-time divide-and-conquer algorithm that finds a balanced ancestral coloring for every perfect binary tree. In fact, our algorithm is capable of finding an ancestral coloring for every \textit{feasible} color sequences $(c_1,c_2,\ldots,c_h)$, where $c_i$ denotes the number of tree nodes having color $i$, not just the balanced ones.
Thus, TreePIR 
can accommodate servers with \textit{heterogeneous} storage/computational capacities. Previous works do not have this flexibility.

Our coloring-based parallel private retrieval scheme of Merkle proofs (TreePIR) is presented in Algorithm~\ref{algo:coloring-based-PIR}.
Note that the pre-processing step is performed \textit{once} per Merkle tree and can be done offline. The time complexity of this step is dominated by the construction of a balanced ancestral coloring, which is $\mathcal{O}(N\log\log(N)) = \mathcal{O}\big(2^h\log(h)\big)$ (see Section~\ref{sec:Algorithm}). 
In fact, the ancestral coloring can be generated once and used for all (perfect) Merkle trees of the same height, regardless of the contents stored it their nodes. 
Nodes in the tree are indexed by $1,2,\ldots,N+1$ in a top-down left-right manner, where $1$ is the index of the root. 

\newpage
\floatname{algorithm}{Algorithm}
\begin{algorithm}[htb!]
    \begin{algorithmic}[1] 
    \vspace{2pt}
    \STATE{\textbf{Input}: A Merkle tree, tree height $h$, and a leaf index $j$;}
    \hrule
    \vspace{2pt}
    \textbf{//Pre-processing:}
    \textcolor{gray}{\scriptsize {// Performed offline \underline{once} per Merkle tree/database}}
    \STATE{Generate the swapped Merkle tree $T(h)$;}
    \STATE{Find a balanced ancestral coloring $\{C_i\}_{i=1}^h$ for $T(h)$;}
    \STATE{Sub-database $\T_i$ stores tree nodes indexed by $C_i$, $i\hspace{-2pt}\in\hspace{-2pt} [h]$;}\\
    \vspace{2pt}
    \hrule
    \vspace{2pt}
    \textbf{//Batch-PIR:} \textcolor{gray}{\scriptsize {// Run the same PIR scheme on $h$ sub-databases}}
    \STATE{Client finds the indices $k_1,\ldots,k_h$ of the nodes in the root-to-leaf-$j$ path by setting $k_h:= j$ and $k_\ell := \lfloor k_{\ell+1}/2\rfloor$ for $\ell=h-1,h-2,\ldots,1$;} 
    \STATE{Client finds the index $\jkl$ of the node $k_\ell$ in its corresponding color class $C_{\il}$, for $\ell\in[h]$;}  \textcolor{gray}{\scriptsize {// Indexing $\mathcal{O}(h^3)$}}
    \FOR{$\ell = 1$ to $h$}
      \STATE Client and server run a PIR scheme to retrieve node indexed $\jkl$ from each sub-database $\T_{\il}$;
    \ENDFOR
    \STATE{Client forms the Merkle proof from the retrieved nodes;}
    \end{algorithmic}
\caption{\textbf{TreePIR}: Coloring-Based Private Retrieval of Merkle Proofs on Merkle Trees}
\label{algo:coloring-based-PIR}
\end{algorithm}
\vspace{-10pt}

\begin{example}
Taking the perfect tree of height three in Fig.~\ref{fig:toy}, TreePIR first finds a balanced ancestral coloring, e.g. $C_1\hspace{-2pt}=\hspace{-2pt}\{2,6,14,15\}$, $C_2\hspace{-2pt}=\hspace{-2pt}\{8,9,10,11,3\}$,  $C_3\hspace{-2pt}=\hspace{-2pt}\{4,5,12,13,7\}$ (line~3). Note how the elements in each $C_i$ are listed using a \textit{left-to-right order}, that is, the nodes on the left in the tree are listed first. 
This order guarantees a \textit{fast indexing} in a later step.
Next, suppose that the client wants to retrieve all the nodes in the root-to-leaf-$11$ path, i.e., $j=11$. The client calculates the indices of these nodes as $k_3\hspace{-2pt}=\hspace{-2pt}j\hspace{-2pt}=\hspace{-2pt}11$, $k_2\hspace{-2pt}=\hspace{-2pt}\lfloor 11/2\rfloor \hspace{-2pt}=\hspace{-2pt} 5$, and $k_1\hspace{-2pt}=\hspace{-2pt}\lfloor 5/2\rfloor\hspace{-2pt}=\hspace{-2pt}2$ (line~5). It then uses an indexing algorithm (see below) to obtain the index of each node in the corresponding color class, namely, node $k_1\hspace{-2pt}=\hspace{-2pt}2$ is the first node in $C_1$ ($j_{k_1}\hspace{-2pt}=\hspace{-2pt}1$), node $k_2\hspace{-2pt}=\hspace{-2pt}5$ is the second node in $C_3$ ($j_{k_2}\hspace{-2pt}=\hspace{-2pt}2$), and node $k_3\hspace{-2pt}=\hspace{-2pt}11$ is the fourth node in $C_2$ ($j_{k_3}\hspace{-2pt}=\hspace{-2pt}4$) (line~6).
Finally, knowing all the indices of these nodes in their corresponding sub-databases, the client sends one PIR query to each sub-database to retrieve privately all nodes in the desired root-to-leaf path (lines~7-9).
\end{example}

\textbf{Indexing.} For the client to run a PIR scheme on a sub-database, it must be able to determine the sub-index $\jkl$ of each node $\kl$ in $C_{\il}$, $\ell \in [h]$ (line 6, Algorithm~\ref{algo:coloring-based-PIR}). We refer to this as the \textit{sub-index problem}. This problem was also discussed in~\cite[p. 970]{angel2018} and in~\cite[p.~439]{mughees2022vectorized}. Existing solutions~\cite{angel2018, AngelSetty2016, mughees2022vectorized} require space or time complexity in $\Omega(N)=\Omega(2^{h+1})$, which is unavoidable for solutions based on probabilistic batch codes. Other batch-code-based solutions~\cite{ishai2004, rawat2016, stinson2009} also require similar or higher indexing complexities, except for a special case with CBC (see Table~\ref{table:BatchPIRcomparison}). TreePIR 
allows an elegant indexing algorithm that incurs only $\O(h^3)$ space and time with no communication overhead. The trick is to let the client run a \textit{miniature} version of the coloring algorithm to retrieve the required sub-indices. More details can be found in Section~\ref{subsec:indexing}.

Theorem~\ref{thm:parallel_complexity} (See Appendix~\ref{app:proof_thm_main} for the proof) summarizes the discussion so far on the privacy and overhead of TreePIR 
(Algorithm~\ref{algo:coloring-based-PIR}). Note that apart from an \textit{exponentially} cheaper indexing algorithm, TreePIR 
achieves \textit{optimal} total storage, which was not the case for existing works, and obtains significant 
improvements in all metrics compared to PBC~\cite{angel2018}. 


        
\begin{theorem}
\label{thm:parallel_complexity}
    Assume that the underlying PIR scheme in TreePIR 
    is correct and computationally private, and with server computation time $\S(d)$, client query generation time $\Cq(d)$, client extraction time $\Ce(d)$, and communication cost $\B(d)$ per query-response pair for a database of size $d$.
    Also, the input Merkle tree is perfect with height $h\hspace{-2pt}\in\hspace{-2pt} \poly(\lambda)$ and $N\hspace{-2pt} = \hspace{-2pt} 2^{h+1}\hspace{-2pt}-\hspace{-2pt}2$ nodes. Then the following statements hold. 
    \begin{enumerate}
        \itemindent=-5pt
        \item TreePIR is correct and computationally private according to Definitions~\ref{def::correctness} and~\ref{def::privacy}, respectively. 
        \item TreePIR has $h$ sub-databases of size $\approx N/h$. It has max server computation time $\S(N/h)$, total server computation time $h\S(N/h)$, client query-generation time $h\Cq(N/h)$, client answer-extraction time $h\Ce(N/h)$, and communication cost $h\B(N/h)$. The total storage is $N$ (optimal) and the indexing time is $\O(h^3)$.
    \end{enumerate}
\end{theorem}

\begin{remark}
\label{rm:parallel_obvious}
A trivial approach to achieve optimal total storage is to use a layer-based ancestral coloring, i.e. $h$ sub-databases correspond to $h$ layers of the tree.
However, the maximum server computation time is $\S(N/2)$ nodes (the bottom layer), which is much 
higher compared to TreePIR. 
\end{remark}
\vspace{-5pt}

\subsection{Related Works and Performance Comparisons}
\label{subsec:comparisons}

As discussed before, the problem of private retrieval of Merkle proofs can be treated as a special case of the batch-PIR problem, in which an arbitrary \textit{subset of items} instead of a single one needs to be retrieved (see Definition~\ref{def::bPIR}). As a Merkle proof consists of $h$ nodes, the batch size is $h$. 
Thus, existing batch-PIR schemes can be used to solve this problem. 
However, a Merkle proof does \textit{not} consist of an \textit{arbitrary} set of $h$ nodes as in the batch-PIR's setting. Indeed, there are only $n = 2^h$ such proofs compared to $\binom{2n-2}{h}$ subsets of random $h$ nodes in an $n$-leaf Merkle tree. TreePIR 
exploits this fact to optimize the storage overhead 
compared to similar approaches using batch codes~\cite{ishai2004, stinson2009, rawat2016, angel2018, mughees2022vectorized}.
In this section, we review existing batch-PIR schemes and provide a comparison with TreePIR 
(Table~\ref{table:BatchPIRcomparison}). 
We also show that TreePIR 
can be combined with some schemes~\cite{LueksGoldberg2015,kales2019, angel2018,mughees2022vectorized} to improve their performance.


\textbf{Batch codes} (BC), introduced by Ishai-Kushilevitz-Ostrovsky-Sahai \cite{ishai2004}, encode a database of size $N$ into $m$ sub-databases (or \textit{buckets}) 
so that the client can retrieve \textit{every} batch of $h$ items by downloading at most one item from each sub-database.
A batch code can be used 
to construct a batch-PIR scheme.
For example, 
in the $\ell$-\textit{subcube code} with $\ell\hspace{-2pt}=\hspace{-2pt}2$ and $h\hspace{-2pt}=\hspace{-2pt}2$, a database $X\hspace{-2pt}=\hspace{-2pt}(x_i)_{i=1}^N$ is transformed into $\ell\hspace{-2pt}+\hspace{-2pt}1\hspace{-2pt}=\hspace{-2pt}3$ sub-databases: $X_1 \hspace{-2pt}\triangleq\hspace{-2pt} (x_i)_{i=1}^{N/2}$, $X_2 \hspace{-2pt}\triangleq\hspace{-2pt} (x_i)_{i={N/2} + 1}^{N}$, and $X_3 \hspace{-2pt}\triangleq\hspace{-2pt} X_1 \hspace{-2pt}\bigoplus\hspace{-2pt} X_2 \hspace{-2pt}=\hspace{-2pt} \big(x_i \hspace{-2pt}\oplus\hspace{-2pt} x_{i + {N/2}}\big)_{i=1}^{N/2}$. Suppose the client wants to privately retrieve two items $x_{j_1}$ and $x_{j_2}$. If $x_{j_1}$ and $x_{j_2}$ belong to different sub-databases then the client can send two PIR queries to these two sub-databases for these items, and another to retrieve a random item in the third sub-database. If $x_{j_1}$ and $x_{j_2}$ belong to the same sub-database, for example, $X_1$, then the client will send three parallel PIR queries to retrieve $x_{j_1}$ from $X_1$, $x_{j_2+N/2}$ from $X_2$, and $x_{j_2}\oplus x_{j_2+N/2}$ from $X_3$. The last two items can be XOR-ed to recover $x_{j_2}$. 

More generally, by recursively applying the above construction $\log_2h$ times with $\ell\geq 2$, the $\ell$-subcube code~\cite{ishai2004} has a total storage of $N h^{\log_2\frac{\ell + 1}{\ell}}$ with $m = h^{\log_2{(\ell + 1)}}$ sub-databases, each of size $s = \frac{N}{h^{\log_2\ell}}$.
Using a larger $\ell$ reduces the storage overhead and the sub-database size (hence reducing the server computation time) but results in more sub-databases (increasing the communication cost).
Another example of batch codes was the one developed in~\cite{rawat2016}, which was based on small regular bipartite graph with no cycles of length less than eight originally introduced by Balbuena~\cite{Balbuena2009}. The indexing of these batch codes, without any clever trick (which is currently missing), would require the client to either download or reconstruct the entire code, both of which require a space or time complexity of at least $\Omega\big(Nh^{\log_2\frac{\ell+1}{\ell}}\big)$ and $\Omega(2N)$, respectively. By contrast, TreePIR indexing complexity is $\O(h^3)$, \textit{exponentially} faster.

\textbf{Combinatorial batch codes} (CBC), introduced by Stinson-Wei-Paterson~\cite{stinson2009}, are special batch codes in which each sub-database/bucket stores \textit{subsets} of items (no encoding is allowed as in the subcube code). 
Our ancestral coloring can be considered as a \textit{relaxed} CBC that allows the retrieval of not every but \textit{some} subsets of $h$ items (corresponding to the Merkle proofs). 
By exploiting this relaxation, our scheme only stores $N$ items across all $h$ sub-databases (no redundancy). By contrast, an optimal CBC, with $h$ sub-databases, requires $\Theta(hN)$ total storage (\hspace{-0.5pt}\cite[Thm. 2.2]{stinson2009}). However, due to its simple construction, the client can compute the required sub-indices in $\Theta(h)$ steps via mathematical formulas. The more general CBC with $m$ sub-databases (\hspace{-0.5pt}\cite[Thm. 2.7]{stinson2009}) incurs very high space complexity of $\Omega(Nm)$ (the entire code) and time complexity $\mathcal{O}(Nh^2)$ to find a maximum bipartite matching. 

\textbf{Probabilistic batch codes} (PBC) have been described in Section~\ref{subsec:batchPIR}. Here we discuss the two indexing strategies of PBC proposed in~\cite[p.~970]{angel2018}, both of which are prohibitively expensive. 
PBC Indexing Strategy $\#1$ requires the client to first download a map (a hash table) that stores $3N$ (key,value) pairs when a minimum number of three hash functions are used in the Cuckoo hashing. The key $k\in [N]$ represents the tree node index, while the value $v = (i, j_k)$ gives the index $i\in [1.5h]$ of the sub-database containing that node together with the node's relative position $j_k\in [2N/h]$ within the sub-database. Hence, each (key,value) pair must be represented by at least $\log_2 N + \log_2 (1.5h) + \log_2 \frac{2N}{h}$ bits, resulting in a map of size at least $4N(\log_2 N + \log_2 (1.5h) + \log_2 \frac{2N}{h})$ bits (assuming a standard load factor of 0.75 for the hash table). The tree itself contains $N$ nodes and can be represented by an array of $N$ elements (tree node indexed $i$ has children at indices $2i+1$ and $2i+2$) of size $256$ bits. Hence, it has size around $256N$ bits. The ratio between the map size (the index) and the Merkle tree size (the entire database) is about 1/3, 2/3, and 1 when $N$ is $2^{10}$, $2^{20}$, and $2^{30}$, respectively. Thus, \textit{downloading the index is almost as expensive as downloading the entire Merkle tree for large trees}. While a Bloom filter can help reduce the map size, the client would need to perform $\O(N)$ hash operations to find the correct indices, which is slow for large trees. 
PBC Indexing Strategy \#2 requires the client to build the indexing map itself while discarding the unused part of the map. The indexing requires $3N$ hashing operations, which alone takes about 100 seconds for a tree of height $h = 24$ (using the implementation from~\cite{mughees2022vectorized} with some optimization) and about \textit{seven hours} when $h = 32$. 

\textbf{Vectorized batch-PIR} (VBPIR) was recently proposed by Mughees-Ren~\cite{mughees2022vectorized} (S\&P'23), which also uses PBC~\cite{angel2018} but with a more batch-friendly vectorized homomorphic encryption. Instead of running independent PIR schemes for the sub-databases/buckets, this scheme \textit{merges} the client queries and the server responses to reduce the communication overhead. Replacing the PBC component by our coloring will reduce the number of sub-databases in their scheme by a factor of $1.5\times$ (from $1.5h$ to $h$) and the sub-database size by a factor of $2\times$ (from $2N/h$ to $N/h$). This will further optimize its storage overhead and reduce the server/client running time for the retrieval of Merkle proofs. Last but not least, the indexing space/time complexity will be improved \textit{exponentially} from $\Omega(2^{h+1})$ to $\O(h^3)$.

\textbf{PIRANA} \cite{Pirana2024} (S\&P'24), the latest batch-PIR, also 
employs PBC for batch retrieval, hence inheriting its costly indexing. 
PBC+PIRANA first uses PBC to generate $w=3$ copies of each data item and then organizes these $3N$ items into $m = 3N/s$ sub-databases of fixed size $s\in \{4096, 8192\}$, which is also the number of slots in a ciphertext.
Thus, when retrieving a Merkle proof, replacing the PBC component in PIRANA with TreePIR would reduce the total storage from $3N$ to $N$ items and the number of  sub-databases from $3N/s$ to $N/s$. Hence, the total server computation time of TreePIR+PIRANA would be $3\times$ lower than PBC+PIRANA.  
Note that PIRANA employs the \textit{constant-weight-code} trick from~\cite{mahdavi2022} to reduce the number of queries from $m$ to $m'<m$, with ${{m'}\choose{k}} \geq 3N$, where $k$ is the Hamming weight of the code. Hence, TreePIR+PIRANA client would send $\sqrt[k]{3}$ fewer queries.
TreePIR+PIRANA would also incur exponentially lower indexing complexity.  

\begin{table}[htb!]
\centering
\setlength{\tabcolsep}{4pt}
    \caption{TreePIR can be combined with Lueks-Goldberg's \textit{multi-client} IT-PIR scheme ($c$ clients retrieve $c$ proofs privately) to reduce its server computation complexity by a factor of $h^{0.80735}\times$.}
\begin{tabular}{p{1.8cm}|c|c|p{1.7cm} }
 \hline
\textbf{ } &  \multicolumn{2}{c|}{\textbf{Lueks-Goldberg (LG)~\cite{LueksGoldberg2015}}} & \textbf{LG \hspace{-2pt}+\hspace{-2pt} TreePIR} \\
 \hline \cline{1-4}
    \textbf{Sub-databases} &  $\sqrt{N}\times\sqrt{N}$ & $\sqrt{\frac{hN}{2}}\times\sqrt{\frac{hN}{2}}$ & $h\hspace{-2pt}\left(\hspace{-2pt}\sqrt{\frac{N}{h}}\hspace{-2pt}\times\hspace{-2pt}\sqrt{\frac{N}{h}}\hspace{-2pt}\right)$\\ 
 
    \textbf{Multiplications} & ${(ch)}^{0.80735}N$ & $c^{0.80735}\frac{hN}{2}$ &  $c^{0.80735}N$\\ 
 
    \textbf{Additions} & $\frac{8}{3}{(ch)}^{0.80735}N$ & $\frac{8}{3}c^{0.80735}\frac{hN}{2}$ &  $\frac{8}{3}{c}^{0.80735}N$\\
 \hline \cline{1-4}
\end{tabular}
    \label{table:mulClientITPIR}
\end{table}


Lueks-Goldberg~\cite{LueksGoldberg2015} combine the Strassen's efficient matrix multiplication algorithm \cite{strassen1969} and Goldberg's IT-PIR scheme~\cite{goldberg2007} to speed up the batch-PIR process. 
In the one-client setting, the database with $N$ nodes is represented as a $\sqrt{N}\times\sqrt{N}$ matrix $D$. 
In the original PIR scheme~\cite{goldberg2007}, each server performs a vector-matrix multiplication of the query vector $\bq$ and the matrix $D$. In the batch PIR version~\cite{LueksGoldberg2015}, the $h$ PIR query vectors $\bq_1,\ldots,\bq_h$ are first grouped together to create an $h\times\sqrt{N}$ query matrix $Q$. Each server then applies Strassen's algorithm to perform the fast matrix multiplication $QD$ to generate the responses, incurring a computational complexity of $\mathcal{O}\big(h^{0.80735}N\big)$ (instead of $\mathcal{O}(hN)$ as in an ordinary matrix multiplication). In the multi-client setting where $c$ clients requests $c$ Merkle proofs, each server performs a multiplication of matrices of size $(ch)\times \sqrt{N}$ and $\sqrt{N}\times \sqrt{N}$ in time $\O((ch)^{0.80735}N)$.  
Our coloring can be applied on top of this scheme to improve its running time. More specifically, an ancestral coloring partitions the tree nodes into $h$ sub-databases, represented by $h$ $\sqrt{N/h}\times\sqrt{N/h}$ matrices. Each server then computes $h$ multiplications on matrices of size $c\times \sqrt{N/h}$ and $\sqrt{N/h}\times\sqrt{N/h}$ in time $\O(c^{0.80735}N)$, hence reducing the (original) computation time by a factor of $h^{0.80735}\times$. 

The work of Kales-Omolola-Ramacher~\cite{kales2019} aimed particularly for the Certificate Transparency infrastructure and improved upon Lueks-Goldberg's work~\cite{LueksGoldberg2015} for the case of a single client with multiple queries over large growing Merkle trees with \textit{billions} leaves. Their idea is to split the original tree into multiple tiers of smaller sub-trees 
with heights $10$-$16$, where the sub-trees at the bottom store all the certificates. The Merkle proofs for the certificates within each bottom sub-tree (static, not changing over time) can be \textit{embedded} inside a Signed Certificate Timestamp, which is included in the certificate. 
These proofs can be used to verify the membership of the certificates within the bottom sub-trees. A batch PIR scheme such as Lueks-Goldberg~\cite{LueksGoldberg2015} can be used for the client to retrieve the Merkle proof of the root of each bottom sub-tree within the top sub-tree. 

We do not include the scheme from~\cite{kales2019} in 
Table~\ref{table:BatchPIRcomparison} as it was designed particularly for Certificate Transparency with a tailored modification, hence requiring extra design features outside of the scope of the batch-PIR problem. 
To combine TreePIR with~\cite{kales2019}, we will need to extend the theory of ancestral coloring to \textit{growing} trees, which remains an intriguing question for future research. Note that our method still works if the database is organized as a growing \textit{forest} of perfect Merkle trees (as in~\cite{BaileySankagiri_FC21,FreglyHarveyKaliskiSheth_CTRSA_2023}). However, queries for nodes belonging to small trees of the forest will have lower privacy. A potential extension of our approach to Sparse Merkle Tree~\cite{DahlbergPP16} is discussed in Appendix~\ref{app:SMT}.

\section{A Divide-And-Conquer Algorithm for Finding Ancestral Colorings of Perfect Binary Trees}
\label{sec:Algorithm}
\vspace{-5pt}

\subsection{The Color-Splitting Algorithm}
\vspace{-5pt}

We develop in this section a divide-and-conquer algorithm that finds an ancestral coloring in almost linear time. All missing proofs can be found in Appendix~\ref{app:main_proofs}.
First, we need the definition of a color sequence. 

\begin{definition} [Color sequence]
\label{def:colorconfiguration}
A color sequence of \textit{dimension} $h$ is a sorted sequence of positive integers $\vec c = (c_1,c_2,\ldots,c_h)$, where $c_1\leq c_2\leq\cdots \leq c_h$. The sum $\sum_{i=1}^h c_i$ is referred to as the sequence's \textit{total size}. 
The element $c_i$ is called the \textit{color size}, which represents the number of nodes in a tree that will be assigned Color $i$. 
The color sequence $\vec c$ is called \textit{balanced} if the color sizes $c_i$ differ from each other by at most one, or equivalently, $c_j-c_i\leq 1$ for all $h\geq j > i \geq 1$. It is assumed that the total size of a color sequence is equal to the total number of nodes in a tree (excluding the root).
\end{definition}

\textbf{A high-level description.} The Color-Splitting Algorithm (CSA) starts from a color sequence $\vec c$ and proceeds to color the tree nodes, \textit{two sibling nodes at a time}, from the top of the tree down to the bottom in a recursive manner while keeping track of the number of remaining nodes that can be assigned each color. 
Note that the elements of a color sequence $\vec c$ are always sorted in \textit{non-decreasing} order, e.g., $\vec c=$ [4 Red, 5 Green, 5 Blue], and CSA always tries to color all the children of the current root $R$ with \textit{either} the same color $1$ if $c_1=2$, \textit{or} with two different colors $1$ and $2$ if $2 < c_1 \leq c_2$. 
This rule stems from the intuition that one must use colors of smaller sizes on the top layers and colors of larger sizes on the lower layers (more nodes).
The remaining colors are carefully \textit{split} 
between the left and the right subtrees of $R$ while ensuring that the color used for each child node will no longer be used for the subtree rooted at that node (to guarantee the Ancestral Property). The key technical challenge is to ensure that the split is done in a way that \textit{prevents the algorithm from getting stuck}, i.e., to make sure that it always has ``enough'' colors to produce  ancestral colorings for \textit{both} subtrees. 
If a \textit{balanced} ancestral coloring is required, CSA starts with the balanced color sequence $\vec{c}^*=[c^*_1,\ldots,c^*_h]$, in which $|c^*_i-c^*_j|\leq 1$ for every $i,j\in [h]$. 

Before introducing rigorous notations and providing a detailed algorithm description, let start with an example of how the coloring algorithm works on $T(3)$.

\vspace{-15pt}
\begin{example}
\label{ex:toy}
\begin{figure}[htb!]
\centering
\includegraphics[scale=0.7]{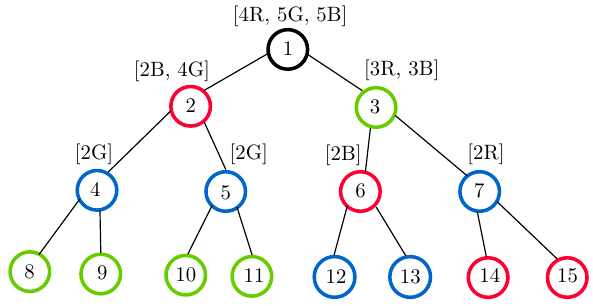}
\caption{An illustration of the Color-Splitting Algorithm being applied to $T(3)$ and the initial color sequence $\vec c=[4,5,5]$. The feasible color sequences used at different node (regarded as root nodes of subtrees) are also given. The splits of color sequences follow the rule in \textbf{FeasibleSplit}(), while the assignment of colors to nodes follow \textbf{ColorSplittingRecursive}().}
\label{fig:split}
\end{figure}

The Color-Splitting Algorithm starts from the root node 1 with the balanced color sequence $\vec{c}^*=[4,5,5]$, which means that it is going to assign Red (Color 1) to four nodes, Green (Color~2) to five nodes, and Blue (Color~3) to five nodes (see~Fig.~\ref{fig:split}). We use $[4\text{R},5\text{G},5\text{B}]$ instead of $[4,5,5]$ to keep track of the colors. Note that the root needs no color. 
According to the algorithm's rule, as Red and Green have the lowest sizes, which are greater than two, the algorithm colors the left child (Node 2) red and the right child (Node 3) green. The dimension-$3$ color sequence $\vec c=[4\text{R},5\text{G},5\text{B}]$ is then split into two dimension-$2$ color sequences $\vec a=[2\text{B}, 4\text{G}]$ and $\vec b=[3\text{R},3\text{B}]$. 
How the split works will be discussed in detail later, however, we can observe that both resulting color sequences have a valid total size $6 = 2+2^2$, which matches the number of nodes in each subtree. Moreover, as $\vec a$ has no Red and $\vec b$ has no Green, the Ancestral Property is guaranteed for Node 2 and Node 3, i.e., these two nodes have different colors from their descendants. The algorithm now repeats what it does to these two subtrees rooted at Node 2 and Node 3 using $\vec a$ and $\vec b$. For the left subtree rooted at 2, the color sequence $\vec a=[2\text{B}, 4\text{G}]$ has two Blues, and so, according to CSA's rule, the two children 4 and 5 of 2 both receive Blue as their colors. The remaining four Greens are split evenly into $[2\text{G}]$ and $[2\text{G}]$, to be used to color $8,9,10$, and~$11$. The remaining steps are carried out in the same manner. 
\end{example}

We observe that not every color sequence $c$ of dimension~$h$, even with a valid total size $\sum_{i=1}^h c_i$, can be used to construct an ancestral coloring of the perfect binary tree $T(h)$. For example, it is easy to verify that there are no ancestral colorings of $T(2)$ using $\vec c=[1,5]$ (1 Red, 5 Greens), and no ancestral colorings of $T(3)$ using $\vec c = [2,3,9]$ (2 Reds, 3 Greens, 9 Blues).
It turns out that there exists a very neat characterization of \textit{all} color sequences of dimension $h$ for that an ancestral coloring of $T(h)$ exists. We refer to them as $h$-\textit{feasible color sequences} (see Definition~\ref{def:feasibility}). 

\begin{definition}[Feasible color sequence]
\label{def:feasibility}

A (sorted) color sequence $\vec c$ of dimension $h$ is called $h$-\textit{feasible} if it satisfies the following two conditions:
\begin{itemize}
    \item (C1) $\sum_{i=1}^\ell c_i \geq \sum_{i=1}^\ell 2^i$, for every $1\leq \ell \leq h$, and 
    \item (C2) $\sum_{i=1}^h c_i = \sum_{i=1}^h 2^i = 2^{h+1}-2$.
\end{itemize}
Condition (C1) means that Colors $1,2,\ldots,\ell$ are sufficient in numbers to color all nodes in Layers $1,2,\ldots,\ell$ of the perfect binary tree $T(h)$ (Layer $i$ has $2^i$ nodes). Condition (C2) states that the total size of $\vec c$ is equal to the number of nodes in $T(h)$. 
\end{definition}

The biggest challenge in designing CSA is to maintain feasible color sequences at \textit{every} step of the algorithm.

\begin{example}
The following color sequences for the trees $T(1), T(2), T(3)$ (see Fig.~\ref{fig:toy}) are feasible: $[2]$, $[3,3]$, and $[4,5,5]$. The sequences $[2,3]$ and $[3,4,6,17]$ are not feasible: $2+3 < 6=2+2^2$, $3+4+6 < 14 = 2+2^2+2^3$. Clearly, color sequences of the forms $[1,\ldots]$, $[2,2,\ldots,]$ or $[2,3,\ldots]$, or $[3,4,6,\ldots]$ violate (C1) and hence are not feasible. 
\end{example}

\begin{definition}
The perfect binary $T(h)$ is said to be \textit{ancestral $\vec c$-colorable}, where $c$ is a color sequence, if there exists an \textit{ancestral coloring} of $T(h)$ in which precisely $c_i$ nodes are assigned Color $i$, for all $i=1,\ldots,h$. Such a coloring is called an \textit{ancestral $\vec c$-coloring} of $T(h)$.
\end{definition}

Lemma~\ref{lem:path} states that every ancestral coloring for $T(h)$ requires at least $h$ colors. Missing proofs are in Appendix~\ref{app:main_proofs}. 


\begin{lemma}
\label{lem:path}
If the perfect binary tree $T(h)$ is ancestral $\vec c$-colorable, where $\vec c=[c_1,\ldots,c_{h'}]$, then $h' \geq h$.
Moreover, if $h'=h$ then all $h$ colors must show up on nodes along any root-to-leaf path (except the root, which is colorless).
Equivalently, nodes having the same color $i \in \{1,2,\ldots,h\}$ must collectively belong to $2^h$ different root-to-leaf paths.
\end{lemma}

Theorem~\ref{thm:main} characterizes \textit{all} color sequences of dimension $h$ that can be used to construct an ancestral coloring of $T(h)$. Note that the balanced color sequence that corresponds to a balanced ancestral coloring is only a special case among all such sequences. On the other hand, even when starting with a balanced color sequence, once the algorithm reaches lower layers, it still has to deal with imbalanced sequences. As far as we know, there is no existing algorithm that can find a balanced ancestral coloring efficiently.

\begin{theorem}[Ancestral-Coloring Theorem for Perfect Binary Trees]
\label{thm:main}
For every $h\geq 1$ and every color sequence $c$ of dimension $h$, the perfect binary tree $T(h)$ is ancestral $\vec c$-colorable if and only if $\vec c$ is $h$-feasible.
\end{theorem}
\begin{proof}
The Color-Splitting Algorithm can be used to show that if $\vec c$ is $h$-feasible then $T(h)$ is ancestral $\vec c$-colorable. For the necessary condition, we show that if $T(h)$ is ancestral $\vec c$-colorable then $\vec c$ must be $h$-feasible. Indeed, for each $1\leq \ell\leq h$, in an ancestral $\vec c$-coloring of $T(h)$, the nodes having the same color~$i$, $1\leq i\leq \ell$, should collectively belong to $2^h$ different root-to-leaf paths in the tree according to Lemma~\ref{lem:path}. Note that each node in Layer $i$, $i = 1,2\ldots,h$, belongs to $2^{h-i}$ different paths. Hence, collectively, nodes in Layers $1,2,\ldots,\ell$ belong to $\sum_{i=1}^\ell 2^i\times 2^{h-i}=\ell2^h$ root-to-leaf paths. We can see that each path is counted $\ell$ times in this calculation. 
Note that each node in Layer $i$ belongs to strictly more paths than each node in Layer $j$ if $i < j$. Thus, if (C1) is violated, i.e., $\sum_{i=1}^\ell c_i < \sum_{i=1}^\ell 2^i$, which implies that the number of nodes having colors $1,2\ldots,\ell$ is smaller than the total number of nodes in Layers $1,2,\ldots,\ell$, then the total number of paths (each can be counted more than once) that nodes having colors $1,2,\ldots,\ell$ belong to is strictly smaller than $\ell2^h$. As a consequence, there exists a color $i \in \{1,2,\ldots,\ell\}$ such that nodes having this color collectively belong to fewer than $2^h$ paths, contradicting Lemma~\ref{lem:path}.
\end{proof}

\begin{corollary}
\label{cr:balanced}
A balanced ancestral coloring exists for the perfect binary tree $T(h)$ for every $h\geq 1$.
\end{corollary}
\begin{proof}
This follows directly from Theorem~\ref{thm:main}, noting that a \textit{balanced} color sequence of dimension $h$ is also $h$-feasible due to Corollary~\ref{cr:balance_conf} (see Appendix~\ref{app:main_proofs}). 
\end{proof}

We now formally describe the Color-Splitting Algorithm. The algorithm starts at the root $R$ of $T(h)$ with an $h$-\textit{feasible} color sequence $c$ (see \textbf{ColorSplitting$(h,\vec c)$}) and then colors the two children $A$ and $B$ of the root as follows: these nodes receive the same Color 1 if $c_1=2$ or Color 1 and Color 2 if $c_1 > 2$. Next, the algorithm splits the remaining colors in $\vec c$ (whose total size has already been reduced by two) into two $(h-1)$-feasible color sequences $\vec a$ and $\vec b$, which are subsequently used for the two subtrees $T(h-1)$ rooted at $A$ and $B$ (see \textbf{ColorSplittingRecursive$(R,h,\vec c)$}). Note that the splitting rule (see \textbf{FeasibleSplit$(h,\vec c)$}) ensures that if Color $i$ is used for a node then it will \textit{not} be used in the subtree rooted at that node, hence guaranteeing the Ancestral Property.  
We prove in Section~\ref{subsec:correctness} and Appendix~\ref{app:main_proofs} that it is always possible to split an $h$-feasible sequence into two new $(h-1)$-feasible sequences, which guarantees a successful termination if the input of the CSA is an $h$-feasible color sequence.
The biggest hurdle in the proof stems from the fact that after being split, the two color sequences are \textit{sorted} before use, which makes the feasibility analysis rather involved.
We overcome this obstacle by introducing a partition technique in which the elements of the color sequence $\vec c$ are partitioned into groups of elements of equal values (called \textit{runs}) and showing that the feasibility conditions hold first for the end-points and then for all middle-points of the runs. 

\begin{example}
We illustrate the Color-Splitting Algorithm when $h = 4$ in Fig.~\ref{fig:split4}. The algorithm starts with a $4$-feasible sequence $\vec c=[3,6,8,13]$, corresponding to 3 Reds, 6 Greens, 8 Blues, and 13 Purples. The root node 1 is colorless. As $2<c_1=3<c_2=8$, CSA colors 2 with Red, 3 with Green, and splits $\vec c$ into $\vec a = [3\text{B}, 5\text{G}, 6\text{P}]$ and $\vec b = [2\text{R}, 5\text{B}, 7\text{P}]$, which are both $3$-feasible and will be used to color the subtrees rooted at 2 and 3, respectively. Note that $\vec a$ has no Red and $\vec b$ has no Green, which enforces the Ancestral Property for Node 2 and Node 3. 
The remaining nodes are colored in a similar manner.
\end{example}

\vspace{-10pt}
\begin{figure}[htb!]
\centering
\includegraphics[scale=0.62]{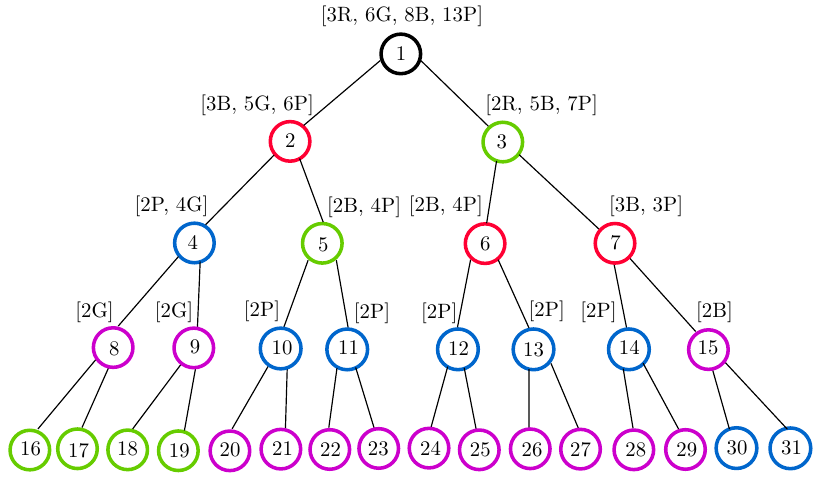}
\caption{An illustration of the Color-Splitting Algorithm being applied to $T(4)$ and the initial (imbalanced) color sequence $\vec c=[3, 6, 8, 13]$. The feasible color sequences used at different nodes (regarded as root nodes of subtrees) are also given. }
\label{fig:split4}
\end{figure}
\vspace{-5pt}

\vspace{-10pt}
\subsection{A Proof of Correctness}
\label{subsec:correctness}
\vspace{-5pt}

We establish the correctness of 
Algorithm~\ref{algo:CSA} in Lemma~\ref{lem:correctness}. The missing proofs can be found in Appendix~\ref{app:main_proofs}. 

\begin{lemma}
\label{lem:if_terminate_then_success}
The Color-Splitting Algorithm, if it terminates successfully, will generate an ancestral coloring. 
\end{lemma}


\begin{lemma}[Correctness of Color-Splitting Algorithm]
\label{lem:correctness}
If the initial input color sequence $\vec c$ is $h$-feasible then the Color-Splitting Algorithm terminates successfully and generates an ancestral $\vec c$-coloring for $T(h)$. Its time complexity is $O\big(2^{h+1}\log h\big)$, almost linear in the number of tree nodes.
\end{lemma}

\floatname{algorithm}{Algorithm}
\begin{algorithm}[htb!]
\caption{\textbf{ColorSplitting}($h,\vec c$)}
\begin{algorithmic}
    \STATE \textcolor{gray}{\scriptsize {// The algorithm finds an ancestral $\vec c$-coloring of $T(h)$, where $\vec c$ is $h$-feasible} }
	\STATE Set $R := 1$;\quad \textcolor{gray}{\scriptsize {// the root of $T(h)$ is 1, which requires no color}}
	\STATE ColorSplittingRecursive$(R, h, \vec c)$;
\end{algorithmic}
\label{algo:CSA}
\end{algorithm}
\vspace{-20pt}

\floatname{algorithm}{Procedure}
\begin{algorithm}[!htb]
\renewcommand{\thealgorithm}{}
\caption{\textbf{ColorSplittingRecursive}$(R, h, \vec c)$}
\begin{algorithmic}[1]
    \STATE \textcolor{gray}{\scriptsize {// The procedure colors the children of the root $R$ of the current subtree $T(h)$ of height $h$ and creates feasible color sequences for its left/right subtrees.}} 
    \IF{$h \geq 1$} 
        \STATE $A := 2R$; $B := 2R+1$; \textcolor{gray}{\scriptsize {// left and right child of $R$}}
        \IF{$c_1=2$}
            \STATE Assign Color 1 to both $A$ and $B$;
            \ELSE 
                \STATE Assign Color 1 to $A$ and Color 2 to $B$;
        \ENDIF
        \IF{$h \geq 2$}
            \STATE Let $\vec a$, $\vec b$ be the output of \textbf{FeasibleSplit}$(h,\vec c)$;
            \STATE \textbf{ColorSplittingRecursive}$(A, h-1, \vec a)$; 
            \STATE \textbf{ColorSplittingRecursive}$(B, h-1, \vec b)$;
        \ENDIF
    \ENDIF
\end{algorithmic}
\end{algorithm}
\vspace{-15pt}

\begin{algorithm}[H]
\renewcommand{\thealgorithm}{}
\caption{\textbf{FeasibleSplit}($h, \vec c$)}
\begin{algorithmic}[1]
    \STATE \textcolor{gray}{\scriptsize {// This algorithm splits a (sorted) $h$-feasible sequence into two (sorted) $(h-1)$-feasible ones, which will be used for coloring the subtrees ($h\geq 2$)}}
    \IF[\textbf{Case 1:} $c_1=2$, note that $c_1 \geq 2$ as $\vec c$ is feasible]{$c_1=2$}
        \STATE Set $a_2\hspace{-2pt}:=\hspace{-2pt}\lfloor\hspace{-1pt} c_2/2\hspace{-1pt} \rfloor$;\hspace{-1pt} $b_2\hspace{-2pt}:=\hspace{-2pt}\lceil\hspace{-1pt} c_2/2 \hspace{-1pt}\rceil$;\hspace{-1pt} $S_2(a) \hspace{-2pt}:=\hspace{-2pt} a_2$;\hspace{-1pt} $S_2(b) \hspace{-2pt}:=\hspace{-2pt} b_2$;
        \FOR{$i=3$ to $h$}
            \IF{$S_{i-1}(a) < S_{i-1}(b)$}
                \STATE Set $a_i:=\lceil c_i/2 \rceil$ and $b_i:=\lfloor c_i/2 \rfloor$;
            \ELSE
                \STATE Set $a_i:=\lfloor c_i/2 \rfloor$ and $b_i:=\lceil c_i/2 \rceil$;
            \ENDIF
            \STATE Update $S_i(a)\hspace{-2pt} := \hspace{-2pt}S_{i-1}(a)\hspace{-2pt}+\hspace{-2pt}a_i$;\ $S_i(b) \hspace{-2pt}:=\hspace{-2pt} S_{i-1}(b)\hspace{-2pt}+\hspace{-2pt}b_i$;
        \ENDFOR
    \ELSE[\textbf{Case 2:} $c_1>2$] 
        \STATE Set $a_2\hspace{-2pt}:=\hspace{-2pt}c_2\hspace{-2pt}-\hspace{-2pt}1$;\hspace{-1pt} $b_2\hspace{-2pt}:=\hspace{-2pt} c_1\hspace{-2pt}-\hspace{-2pt}1$;\hspace{-1pt} \textcolor{gray}{\scriptsize {//$b_2$ now refers to the 1st color in $\vec b$}}
        \IF{$h\geq 3$}
            \STATE Set $a_3\hspace{-2pt}:=\hspace{-2pt}\left\lceil \hspace{-2pt}\frac{c_3+c_1-c_2}{2} \hspace{-2pt}\right\rceil$ and $b_3\hspace{-2pt}:=\hspace{-2pt} c_2\hspace{-2pt}-\hspace{-2pt}c_1 \hspace{-2pt}+\hspace{-2pt} \left\lfloor\hspace{-2pt} \frac{c_3+c_1-c_2}{2} \hspace{-2pt}\right\rfloor$;
            \STATE Set $S_3(a) := a_2+a_3$ and $S_3(b) := b_2+b_3$; 
            \FOR{$i=4$ to $h$}
                \IF{$S_{i-1}(a) < S_{i-1}(b)$}
                    \STATE Set $a_i:=\lceil c_i/2 \rceil$ and $b_i:=\lfloor c_i/2 \rfloor$;
                \ELSE
                    \STATE Set $a_i:=\lfloor c_i/2 \rfloor$ and $b_i:=\lceil c_i/2 \rceil$;
                \ENDIF
                \STATE Update $S_i(a) \hspace{-2pt}:=\hspace{-2pt} S_{i-1}(a)\hspace{-2pt}+\hspace{-2pt}a_i$;\ $S_i(b) \hspace{-2pt}:=\hspace{-2pt} S_{i-1}(b)\hspace{-2pt}+\hspace{-2pt}b_i$;
            \ENDFOR
        \ENDIF
    \ENDIF
    \STATE Sort $\vec a = [a_2,a_3,\ldots,a_h]$ and $\vec b = [b_2,b_3,\ldots,b_h]$ in non-decreasing order;
    \RETURN $\vec a$ and $\vec b$;
\end{algorithmic}
\end{algorithm}
\vspace{-10pt}

\vspace{-5pt}

\subsection{Solving the Sub-Index Problem in Time $\O(h^3)$}
\label{subsec:indexing}
\vspace{-5pt}

As discussed in Section~\ref{sec:OurProposalRelatedWorks}, after partitioning a (swapped) Merkle tree into $h$ color classes, in order for the client to make PIR queries to the sub-databases (corresponding to color classes), it must know the sub-index $\jkl$ of the node $\kl$ in the color class $C_{\il}$, for all $\ell=1,2,\ldots,h$ (see Algorithm~\ref{algo:coloring-based-PIR}). 
Trivial solutions including the client storing all $C_i$'s or regenerating $C_i$ itself by running the CSA on its own all require a space or time complexity in $\Omega(N)$. For trees of height $h=30$ or more, these solutions would demand a prohibitively large computational overhead or else Gigabytes of indexing data being sent to the client, rendering the whole retrieval scheme impractical. Fortunately, the way our divide-and-conquer algorithm (CSA) colors the tree also provides an efficient and neat solution for the sub-index problem. 
We describe our proposed solution below. 



\floatname{algorithm}{Algorithm}
\setcounter{algorithm}{2}
\begin{algorithm}[htb]
    \begin{algorithmic}[1] 
    \STATE{\textbf{Input}: Tree height $h$, leaf index\hspace{-2pt} $j\hspace{-2pt} \in\hspace{-2pt} \{2^h\hspace{-2pt},\ldots,2^{h+1}\hspace{-2pt}-\hspace{-2pt}1\hspace{-2pt}\}$, $h$~node indices $k_1,\ldots,k_h$, and an $h$-feasible color sequence $\vec c = (c_1,\ldots,c_h)$;}
    \vspace{2pt}
    \hrule
    \vspace{2pt}
    \STATE Initialize an array $\texttt{idx}[h]$ to hold the output sub-indices of nodes $k_1\hspace{-1pt},\ldots,k_h$\hspace{-1pt} in their corresponding color classes;
    \STATE Initialize an array $\texttt{count}[h] := [0,0,\ldots,0]$; 
    \STATE $R:=1$;
    \FOR{$\ell=1$ to $h-1$}
        \STATE Let $\sbl$ be the sibling node of $\kl$; \textcolor{gray}{\scriptsize {// $\sbl \hspace{-2pt}:=\hspace{-2pt} \kl\hspace{-2pt}-\hspace{-2pt}1$ or $\kl\hspace{-2pt}+\hspace{-2pt}1$}}
        \STATE Assign Color $\il$ to $\kl$ and Color $\ilp$ to $\sbl$; \textcolor{gray}{\scriptsize {// Following lines 3-7 in \textbf{ColorSplittingRecursive}, note that it is possible that $\il=\ilp$}}
        \STATE $\texttt{is\_left} := \text{True}$ if $\kl$ is the \textit{left} child of $R$, and False otherwise; \textcolor{gray}{\scriptsize {// if $\kl=2R$ then it is the left child of $R$}}
        \STATE Let $\vec a$ and $\vec b$ be the output color sequences of \textbf{FeasibleSplit}$(h-\ell+1,\vec c)$; \textcolor{gray}{\scriptsize {// time complexity $\O(h)$}}
        \STATE $\vec c := \vec a$ if $\texttt{is\_left}=\text{True}$, and $\vec c := \vec b$ otherwise;
        \STATE \textbf{UpdateCount}$(\texttt{count}, \il,\ilp, \texttt{is\_left}, \vec a, \vec b)$;\hspace{-1pt} \textcolor{gray}{\scriptsize {// $\O(h^2\hspace{-1pt})$}}
        \STATE $\texttt{idx}[\ell] := \texttt{count}[\il]+1$; \textcolor{gray}{\scriptsize {// the sub-index of $\kl$ in $C_i$}}
        \STATE $R := \kl$; \textcolor{gray}{\scriptsize {// move down to the child node $\kl$}}    
    \ENDFOR
    \STATE Let $i_h$ be the only color left in $\vec c$; \textcolor{gray}{\scriptsize {// $\vec c = [2]$, e.g. two greens}}
    \STATE Update $\texttt{count}[i_h] := \texttt{count}[i_h]+1$ if $k_h$ is the \textit{right} child of $k_{h-1}$; \textcolor{gray}{\scriptsize {// if it is the left child, do nothing}}
    \STATE $\texttt{idx}[h] := \texttt{count}[i_h]+1$; \textcolor{gray}{\scriptsize {// the sub-index of $k_h$ in $C_{i_h}$}}
    \RETURN $\texttt{idx}$;
    \end{algorithmic}
\caption{\textbf{TreePIR-Indexing}}
\label{algo:indexing}
\end{algorithm}
\vspace{-15pt}

\floatname{algorithm}{Procedure}
\renewcommand{\thealgorithm}{}
\begin{algorithm}[htb]
    \begin{algorithmic}[1] 
    \FOR{$i=1$ to $h$}
        \IF{\texttt{is\_left} = False AND $\vec b[i] \neq 0$}
            \STATE Find the number $l_i$ of Color $i$ given to $\vec a$; \textcolor{gray}{\scriptsize {// $\O(h)$}}
            \STATE $\texttt{count}[i] := \texttt{count}[i]+l_i$; \textcolor{gray}{\scriptsize {// Gaining $l_i$ nodes of Color $i$ on the left side of $\kl$ (excluding the sibling)}}
            \IF{$i=\ilp$} 
                \STATE $\texttt{count}[i] := \texttt{count}[i]+1$; \textcolor{gray}{\scriptsize {// Gaining one node of Color $i$ (due to the sibling receiving colore $i$) on the left side of $\kl$}}
            \ENDIF
        \ENDIF  
    \ENDFOR
    \end{algorithmic}
\caption{\textbf{UpdateCount}$(\texttt{count}, \il,\ilp, \texttt{is\_left}, \vec a, \vec b))$}
\end{algorithm}

The main idea of Algorithm~\ref{algo:indexing} is to apply a modified \textit{non-recursive} version of the Color-Splitting Algorithm (Algorithm~\ref{algo:CSA}) \textit{from the root to the leaf} $j$ only, while using an array of size $h$ to keep track of the number of nodes with color $i$ on the \textit{left} of the current node $R=1,k_1,k_2,\ldots,k_h$. 
Note that due to the Ancestral Property, nodes belonging to the same color class are not ancestor-descendant of each other. Therefore, the \textit{left-right} relationship is well defined: for any two nodes $u$ and $v$ having the same color, let $w$ be their common ancestor, then $w \neq u$, $w \neq v$, and $u$ and $v$ must belong to different subtrees rooted at the left and the right children of $w$; suppose that $u$ belongs to the left and $v$ belongs to the right subtrees, respectively, then $u$ is said to be on the left of $v$, and $v$ is said to be on the right of $u$. The key observation is that if a node $k$ has color $i$ and there are $j_k-1$ nodes of color $i$ on its left in the tree, then $k$ has index $j_k$ in the color class $C_i$, assuming that nodes in $C_i$ are listed in the left-to-right order, i.e., left nodes appear first. This left-right order (instead of the natural top-down, left-right order) is crucial for our indexing to work.

\begin{figure}[htb!]
    \centering
    \includegraphics[scale=0.8]{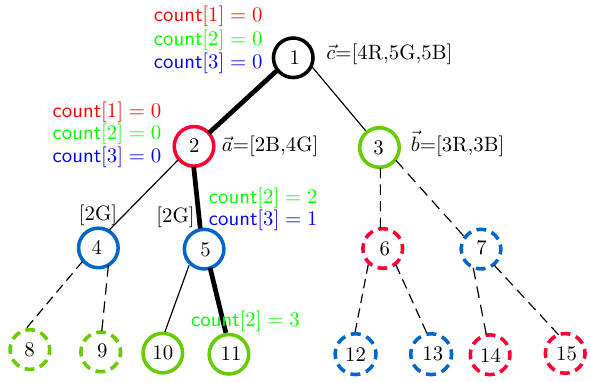}
    \caption{A demonstration of how Algorithm~\ref{algo:indexing} finds the sub-indices of all nodes along the root-to-leaf-11 path in their corresponding color classes. The algorithm performs the color-splitting algorithm (Algorithm~\ref{algo:CSA}) only on the nodes along the path and their siblings. Other nodes (dashed) are ignored. When the current node $k$ receives color $i$, it has sub-index $j_k=\texttt{count}[i]+1$ in the color class $C_i$. Thus, node 2 is the first node among the reds, node 5 the second among the blues, and node 11 the forth among the greens.} 
    \label{fig:indexing_example}
\end{figure}
\vspace{-10pt}


\begin{example}
\label{ex:indexing}
Consider the tree of height $h=3$ in Fig.~\ref{fig:indexing_example}, in which the path 1-2-5-11 is considered. Let colors $i=1,2,3$ denote Red, Green, and Blue, respectively. As the algorithm colors node 2 red and node 3 green, and gives three reds and three blues to the right branch, at node 2 (treated as the current node), the $\texttt{count}$ array is updated as follows: 
\begin{itemize}
    \item $\texttt{count}[1]\hspace{-2pt}=\hspace{-2pt}0$ (unchanged).
    As node 2 receives color $1$ (Red), it has sub-index $j_2 \hspace{-2pt}=\hspace{-2pt} \texttt{count}[1]\hspace{-2pt}+\hspace{-2pt}1\hspace{-2pt}=\hspace{-2pt}1$ in $C_1$.
    \item $\texttt{count}[2]=0$ (unchanged).
    \item $\texttt{count}[3]=0$ (unchanged).
\end{itemize}
The algorithm continues in a similar manner. Once it reaches the leaf 11, it has found all the sub-indices of nodes 2, 5, 11 in the corresponding color classes: $j_2\hspace{-2pt}=\hspace{-2pt}1$, $j_5\hspace{-2pt}=\hspace{-2pt}2$, and $j_{11}\hspace{-2pt}=\hspace{-2pt}4$.  
This is correct given that the color classes are arranged according to the \textit{left-right order}, with $C_1=\{2,6,14,15\}$, $C_2=\{8,9,10,11,3\}$, and $C_3=\{4,5,12,13,7\}$.
\end{example}

\begin{theorem}
\label{thm:indexing}
Algorithm~\ref{algo:indexing} returns the correct sub-indices of nodes along a root-to-leaf path in their corresponding color classes. Moreover, the algorithm has worst-case time complexity $\O(h^3)$ for a perfect binary tree of height $h$.  
\end{theorem}
\begin{proof}
At the beginning, at the root of the tree, $\texttt{count}[i]=0$ for all $i=1,2,\ldots,h$. 
As the algorithm proceeds to color the left and right children of the root and the colors are distributed to the left and right subtrees, the array $\texttt{count}$ is updated accordingly, reflecting the number of nodes of color $i$ on the left of the current node in consideration. Therefore, when a node $k$ receives its color $i$,  
its sub-index in $C_i$ is $j_k=\texttt{count}[i]+1$.  
The statement on the complexity follows from the description of the algorithm. The key reason for this low complexity is that the algorithm only colors $h$ nodes along the root-to-leaf path together with their siblings.
\end{proof}

\section{Experiments and Evaluations}
\label{sec:ExpandEval}
\vspace{-5pt}

We now describe the implementation of TreePIR and compare its performance with the state-of-the-art Probabilistic Batch Code (PBC)~\cite{angel2018}. 
Both PBC and TreePIR, which are different types of batch code, can be combined with an underlying PIR scheme to form a batch-PIR. Thus, we evaluate their performance when combined with two well-known PIR schemes, namely SealPIR, which was originally used for PBC in~\cite{angel2018}, and 
Spiral~\cite{menon2022}. Additionally, we also compare the efficiency of TreePIR versus PBC when embedded in Vectorized Batch-PIR (VBPIR)~\cite{mughees2022vectorized}. 
Note that PIRANA~\cite{Pirana2024} code is not available. Hence, we can only evaluate TreePIR+PIRANA theoretically (see Section~\ref{subsec:comparisons}).

As discussed in Section~\ref{subsec:ProblemDescription}, the performance of a batch-code-based batch-PIR depends on the number of sub-databases and their sizes, and on the performance of the underlying PIR scheme. Moreover, it also depends on the dimension $d$ that PIR sub-databases are represented, e.g., $d=2$ for SealPIR, $d=2, 3$ for VBPIR, and $d=4$ for Spiral. Theoretically, as TreePIR uses $1.5\times$ fewer sub-databases of size $2\times$ smaller compared to PBC, theoretically, TreePIR uses $3\times$ less storage, with $\sqrt[d]{2}\times$ faster max server computation, $1.5\sqrt[d]{2}\times$ faster total server computation, and $1.5 \times$ faster client query-generation time (except for VBPIR, the speedup for client query-generation time is $\sqrt[d]{2}\times$). The client answer-extraction times are similar for both as dummy responses are ignored in PBC. These theoretical gains were also reflected in the experiments.

TreePIR has significantly faster setup and indexing thanks to its efficient coloring and indexing algorithms on trees. In particular, for trees with $2^{10}$-$2^{24}$ leaves, TreePIR's setup and indexing are $8$-$60\times$ and $19$-$160\times$ faster than PBC's, respectively. TreePIR still works well beyond that range, requiring $180$ seconds to setup a tree of $2^{30}$ leaves, and only $0.7$ milliseconds to index in a tree of $2^{36}$ leaves.




\subsection{Experiment Setup}
We ran our experiments on a laptop (Intel® Core™ i9-13900H and 32GiB of system memory) in the Ubuntu 22.04 LTS environment. 
For each tree size, we ran the batch-PIR protocols for 10 random Merkle proofs and recorded the averages. 
To build Merkle trees of $2^{10}$-$2^{20}$ leaves, we fetched $2^{20}$ entries from Google's Xenon2024~\cite{MerkleTown}, 
each of which comprises of an entry number, a timestamp, and a certificate, and applied SHA-256 on the certificates to produce tree leaves. 
Consequently, each tree node has size 32 bytes. 
To evaluate TreePIR's performance for larger trees ($n = 2^{22},\ldots,2^{36})$, we use random hashes to avoid excessive hashing overheads. We did not consider PBC beyond $h = 24$ as its index became too large (see Table~\ref{table:mapsize}).
We used the existing C++ implementations of PBC~\cite{VBPIR_repository} by Mughees-Ren~\cite{mughees2022vectorized} (after fixing some minor errors), Spiral~\cite{Spiral_repository}, and VBPIR~\cite{VBPIR_repository}. 
All the codes for our experiments are available online at https://github.com/newPIR/TreePIR. 

        


\subsection{Evaluations}
\vspace{-5pt}

\textbf{Setup Computation Time.} 
PBC's setup employs $w = 3$ independent hash functions to allocate each tree node to three among $1.5h$ sub-databases. It also generates the index $\mathrm{map}_{\mathrm{PBC}}$, which will be sent to the client. In its setup phase, TreePIR colors the swapped Merkle tree and creates $h$ balanced sub-databases. 
While both PBC and TreePIR have setup times asymptotically (almost) linear in $2^h$, PBC's hashings and index generation are much slower than TreePIR's tree coloring.
As shown in Table~\ref{table:partitiontime}, TreePIR's setup is $8$-$60\times$ faster than PBC's for $10\leq h\leq 24$. 

\begin{table}[htb!]
    \centering
    \setlength{\tabcolsep}{2.8pt}
    \caption{A comparison of the \textit{setup} time between TreePIR and PBC. TreePIR's setup is $8$-$60\times$ faster for trees of $2^{10}$-$2^{24}$ leaves.}
    \begin{tabular}{p{2cm} c c c c c c c c }
        \hline \cline{1-9}
        $h$ & $10$ & $12$ & $14$ & $16$ & $18$ & $20$ & $22$ & $24$\\
        \hline \cline{1-9}
        PBC (ms) & $3.4$ & $8.9$ & $53.1$ & $406$ & $2132$ & $9259$ & $37591$ & $159814$\\
        
        \textbf{TreePIR} (ms) & $0.4$ & $2.4$ & $7.6$ & $39.1$ & $56.1$ & $179$ & $600$ & $2700$\\
        \hline \cline{1-9}
        $h$ & $26$ &  &  & $28$ &  & $29$ &  & $30$\\
        \hline \cline{1-9}
        \textbf{TreePIR} (sec)  & $9.6$ &  &  & $37.3$ & & $77.1$ & & $179.4$ \\
        \hline \cline{1-9}
    \end{tabular}
    \vspace{5pt}
    
    \label{table:partitiontime}
\end{table}

\textbf{Setup Communication Cost.} To perform PIR queries on sub-databases, the client must know the positions of the (Merkle proof) nodes within the sub-databases. To~that~end, PBC server generates an index/map ($\mathrm{map}_{\mathrm{PBC}}$), which must be downloaded  by the client.  
As shown in Table~\ref{table:mapsize}, the PBC index reaches $294$ MB for a tree of $2^{20}$ leaves and $4.7$ GB for the tree of $2^{24}$ leaves, which is impractical. By contrast, TreePIR indexing is on-the-fly and requires no maps. 

\begin{table}[h]
    \centering
    \setlength{\tabcolsep}{4pt}
    \caption{PBC client must download/store a large index that grows linearly with the tree size, whereas TreePIR requires no index.}
    \begin{tabular}{p{2.6cm} c c c c c c c c }
         \hline \cline{1-9}       
        $h$ & $10$ & $12$ & $14$ & $16$ & $18$ & $20$ & $22$ & $24$ \\
        \hline \cline{1-9}
        $\mathrm{map}_{\mathrm{PBC}}$ (MB) & $0.3$ & $1.2$ & $4.6$ & $18$ & $73$ & $294$ & $1175$ & $4703$\\
        \hline \cline{1-9}
    \end{tabular}
    \vspace{5pt}
    
    \label{table:mapsize}
\end{table}

\textbf{Client Indexing Time} is the total time the client finds the sub-indices of all $h$ nodes from a Merkle proof within the $h$ sub-databases.
Even when excluding the time spent on downloading the large index and loading it into RAM, PBC client's indexing time (using Cuckoo hashing) is still around $19$-$160\times$ slower than TreePIR's efficient indexing. 
TreePIR's indexing still works for even larger trees ($h \geq 26$) where PBC's index already becomes to large to handle.

\begin{table}[h]
    \centering
    \setlength{\tabcolsep}{3.4pt}
    \caption{A comparison of the \textit{indexing} times of TreePIR and PBC. Despite ignoring the download time of its (large) index, PBC's indexing is still $19$-$160\times$ slower than TreePIR's indexing.} 
    \begin{tabular}{p{2.2cm} c c c c c c c c }
        \hline \cline{1-9}
        $h$ & $10$ & $12$ & $14$ & $16$ & $18$ & $20$ & $22$ & $24$\\
        \hline \cline{1-9}
        Indexing PBC (ms) & $4$ & $4$ & $4$ & $4$ & $5$ & $11$ & $22$ & $74$\\
        
        \textbf{TreePIR} (ms) & $0.21$ & $0.24$ & $0.25$ & $0.32$ & $0.34$ & $0.38$ & $0.41$ & $0.46$\\
        \hline
        \cline{1-9}
        $h$ & $26$ & $28$ & $30$ & $32$ & $34$ & $36$ \\
        \cline{1-7} \cline{1-7}
        \textbf{TreePIR} (ms) & $0.47$ & $0.48$ & $0.51$ & $0.52$ & $0.61$ & $0.69$ \\
        \cline{1-7} \cline{1-7}
    \end{tabular}
    \vspace{5pt}
    
    \label{table:indexingtime}
\end{table}

\textbf{Client PIR Computation Time} includes \textit{query-generation} and \textit{answer-extraction} times. TreePIR uses $h$ sub-databases, hence requiring only $h$ queries and responses, $1.5\times$ fewer than PBC.
As seen in Table~\ref{table:PIRClientQueryTime}, SealPIR+TreePIR and Spiral+TreePIR have query-generation times about $1.5\times$ lower than those of SealPIR+PBC and Spiral+PBC, reflecting the theoretical gap.
The client query-generation time in VBPIR+TreePIR was around $\sqrt[d]{2} \times$ lower than VBPIR+PBC, with $d=2$ or $3$. 

\begin{table}[h]
    \centering
    \setlength{\tabcolsep}{6.8pt}
    \caption{The client \textit{query-generation} time of TreePIR and PBC when combined with SealPIR, Spiral, and VBPIR. }
    \begin{tabular}{p{2.9cm} c c c c c c}
        \hline \cline{1-7}
        $h$ & $10$ & $12$ & $14$ & $16$ & $18$ & $20$ \\
        \hline \cline{1-7}
        SealPIR+PBC (ms) & $21$ & $24$ & $28$ & $31$ & $37$ & $40$ \\
        
        \textbf{SealPIR+TreePIR} (ms) & $13$ & $16$ & $18$ & $21$ & $24$ & $27$ \\
        \hline \cline{1-7}
        Spiral+PBC (ms) & $33$ & $39$ & $46$ & $54$ & $62$ & $66$ \\
        
        \textbf{Spiral+TreePIR} (ms) & $20$ & $24$ & $29$ & $33$ & $37$ & $41$ \\
        \hline \cline{1-7}
        VBPIR+PBC (ms) & $5.1$ & $5.1$ & $7.0$ & $7.0$ & $7.3$ & $7.5$ \\
        
        \textbf{VBPIR+TreePIR} (ms) & $4.1$  & $4.1$ & $4.3$ & $6.2$ & $6.4$ & $6.6$ \\
        \hline \cline{1-7}
    \end{tabular}
    \label{table:PIRClientQueryTime}
\end{table}

The client extraction time (Table~\ref{table:PIRClientExtractTime}) are the same for SealPIR+PBC and Spiral+PBC, and for Spiral+PBC and Spiral+TreePIR, as clients process only $h$ responses for both (PBC client ignores the $0.5h$ dummy responses). The same goes for VBPIR+PBC and VBPIR+TreePIR, except when $h=14$ where one uses $d=2$ and the other uses $d=3$. 

\begin{table}[h]
    \centering
    \setlength{\tabcolsep}{5pt}
    \caption{The client \textit{answer-extraction} time of TreePIR and PBC when combined with SealPIR, Spiral, and VBPIR are similar.}
    \begin{tabular}{p{2.9cm} c c c c c c}
        \hline \cline{1-7}
        $h$ & $10$ & $12$ & $14$ & $16$ & $18$ & $20$ \\
        \hline \cline{1-7}
        SealPIR+PBC (ms) & $12.4$ & $14.9$ & $17.3$ & $19.4$ & $22.4$ & $24.6$ \\
        
        \textbf{SealPIR+TreePIR} (ms) & $12.6$ & $15.0$ & $17.2$ & $19.4$ & $22.2$ & $24.3$ \\
        \hline \cline{1-7}
        Spiral+PBC (ms) & $6.7$ & $7.9$ & $9.4$ & $10.4$ & $10.9$ & $12.2$ \\
        
        \textbf{Spiral+TreePIR} (ms) & $5.8$ & $7.2$ & $8.4$ & $9.6$ & $10.5$ & $12.2$ \\
        \hline \cline{1-7}
        VBPIR+PBC (ms) & $1.3$ & $1.3$ & $0.7$ & $0.7$ & $0.7$ & $0.7$ \\
        
        \textbf{VBPIR+TreePIR} (ms) & $1.3$ & $1.3$ & $1.3$ & $0.7$ & $0.7$ & $0.7$ \\
        \hline \cline{1-7}
    \end{tabular}
    \vspace{5pt}
    \label{table:PIRClientExtractTime}
    \vspace{-10pt}
\end{table}

\textbf{Server(s) Storage Cost.} PBC uses $1.5\times$ more sub-databases of size $2\times$ larger than TreePIR's. As the result, the total storage of PBC is $3\times$ larger than that of TreePIR, as also reflected in Table~\ref{table:dbsize}. 

\begin{table}[h]
    \centering
    \setlength{\tabcolsep}{3pt}
    \caption{A comparison of individual sub-database sizes, $s_{\mathrm{PBC}}$ and $s_{\mathrm{TreePIR}}$, and the total storage (all sub-databases), $S_{\mathrm{PBC}}$ and $S_{\mathrm{TreePIR}}$, for PBC and TreePIR, respectively. TreePIR requires $3\times$ smaller total storage and $2\times$ smaller sub-databases, as expected.}
    \begin{tabular}{p{2cm} c c c c c c c c}
        \hline \cline{1-9}
        $h$ & $10$ & $12$ & $14$ & $16$ & $18$ & $20$ & $22$ & $24$\\
        \hline \cline{1-9}
        $s_{\mathrm{PBC}}$ (MB) & $0.014$ & $0.05$ & $0.15$ & $0.53$ & $1.9$ & $6.7$ & $24.5$ & $89.6$\\
        
        $s_{\mathbf{TreePIR}}$ (MB) & $0.007$ & $0.02$ & $0.08$ & $0.26$ & $0.9$ & $3.4$ & $12.2$ & $44.7$\\
        \hline \cline{1-9}
        $S_{\mathrm{PBC}}$ (MB) & $0.21$ & $0.83$ & $3.21$ & $12.8$ & $51$ & $202$ & $807$ & $3227$\\
        
        $S_{\mathbf{TreePIR}}$ (MB) & $0.07$ & $0.26$ & $1.05$ & $4.2$ & $17$ & $67$ & $268$ & $1074$\\
        \hline \cline{1-9}
    \end{tabular}
    \vspace{5pt}
    \label{table:dbsize}
\end{table}
\vspace{-5pt}

\textbf{PIR Server Computation Time.} 
We record in Table~\ref{table:parallelPIRServersMax} the \textit{maximum} time the server took 
to generate a response for a sub-database. 
On the other hand, Table~\ref{table:PIRServerTotal} shows the time the server took to produce \textit{all} responses. 
The maximum server computation time is a relevant metric when the parallel mode is considered (one server/thread per sub-database), while the total server computation time is relevant when the sequential mode is considered (a single server/thread handles all sub-databases sequentially). 
The running times reported in Tables~\ref{table:parallelPIRServersMax} and~\ref{table:PIRServerTotal} reflect the theoretical speedups for TreePIR predicted by the theory: $\sqrt[d]{2}\times$ for maximum and $1.5\sqrt[d]{2}\times$ for total server computation times, respectively. 

\begin{table}[h]
    \centering
    \setlength{\tabcolsep}{6pt}
    \caption{Theoretically, TreePIR's max server computation time is $\sqrt[d]{2} \times$ faster than PBC (this metric is irrelevant to VBPIR, which runs in sequential mode)]. This is reflected correctly in the table with $d = 2$ for SealPIR and $d=4$ for Spiral.}
    \begin{tabular}{p{2.8cm} c c c c c c}
        \hline \cline{1-7}
        $h$ & $10$ & $12$ & $14$ & $16$ & $18$ & $20$ \\
        \hline \cline{1-7}
        SealPIR+PBC (ms) & $6.0$ & $6.2$ & $12.4$ & $21.3$ & $60$ & $107$ \\
        
        \textbf{SealPIR+TreePIR} (ms) & $5.8$ & $5.8$ & $9.1$ & $18.9$ & $36$ & $76$ \\
        \hline \cline{1-7}
        Spriral+PBC (ms) & $33$ & $34$ & $34$ & $34$ & $35$ & $39$ \\
        
        \textbf{Spiral+TreePIR} (ms) & $30$ & $31$ & $31$ & $31$ & $32$ & $33$ \\
        \hline \cline{1-7}
    \end{tabular}
    \vspace{5pt}
    \label{table:parallelPIRServersMax}
\end{table}

\begin{table}[htb!]
    \centering
    \setlength{\tabcolsep}{5pt}
    \caption{TreePIR's total server computation time is $1.5$-$2\times$ faster than PBC for larger trees. Theoretically, it is $1.5 \sqrt[d]{2} \times$ faster.}
    \begin{tabular}{p{2.8cm} c c c c c c}
        \hline \cline{1-7}
        $h$ & $10$ & $12$ & $14$ & $16$ & $18$ & $20$ \\
        \hline \cline{1-7}
        SealPIR+PBC (ms) & $69$ & $99$ & $235$ & $486$ & $1185$ & $2916$ \\
        
        \textbf{SealPIR+TreePIR} (ms) & $47$ & $55$ & $104$ & $233$ & $531$ & $1250$ \\
        \hline \cline{1-7}
        Spiral+PBC (ms) & $501$ & $613$ & $700$ & $805$ & $916$ & $1151$ \\
        
        \textbf{Spiral+TreePIR} (ms) & $309$ & $373$ & $429$ & $507$ & $561$ & $663$ \\
        \hline \cline{1-7}
        VBPIR+PBC (ms) & $399$ & $405$ & $586$ & $969$ & $2357$ & $7481$ \\
        
        \textbf{VBPIR+TreePIR} (ms) & $396$ & $397$ & $415$ & $588$ & $1372$ & $3871$ \\
        \hline \cline{1-7}
    \end{tabular}
    \vspace{5pt}
    \label{table:PIRServerTotal}
    \vspace{-10pt}
\end{table}

\textbf{PIR Communication Cost} measures the total amount of data transmitted over the network between the client and PIR server(s), that is, the total size of PIR queries and responses. Note that there are $h$ sub-databases in TreePIR and $1.5h$ sub-databases in PBC. This is reflected in Figure~\ref{fig:PIRCom}: SealPIR+PBC and Spiral+PBC communication costs are roughly $1.5\times$ higher than that of SealPIR+TreePIR and Spiral+PBC, respectively. 
VBPIR packs several queries and responses into single ciphertexts, hence the communication costs are similar when combined with PBC and TreePIR.
\begin{figure}[htb!]
\centering
\includegraphics[scale=0.37]{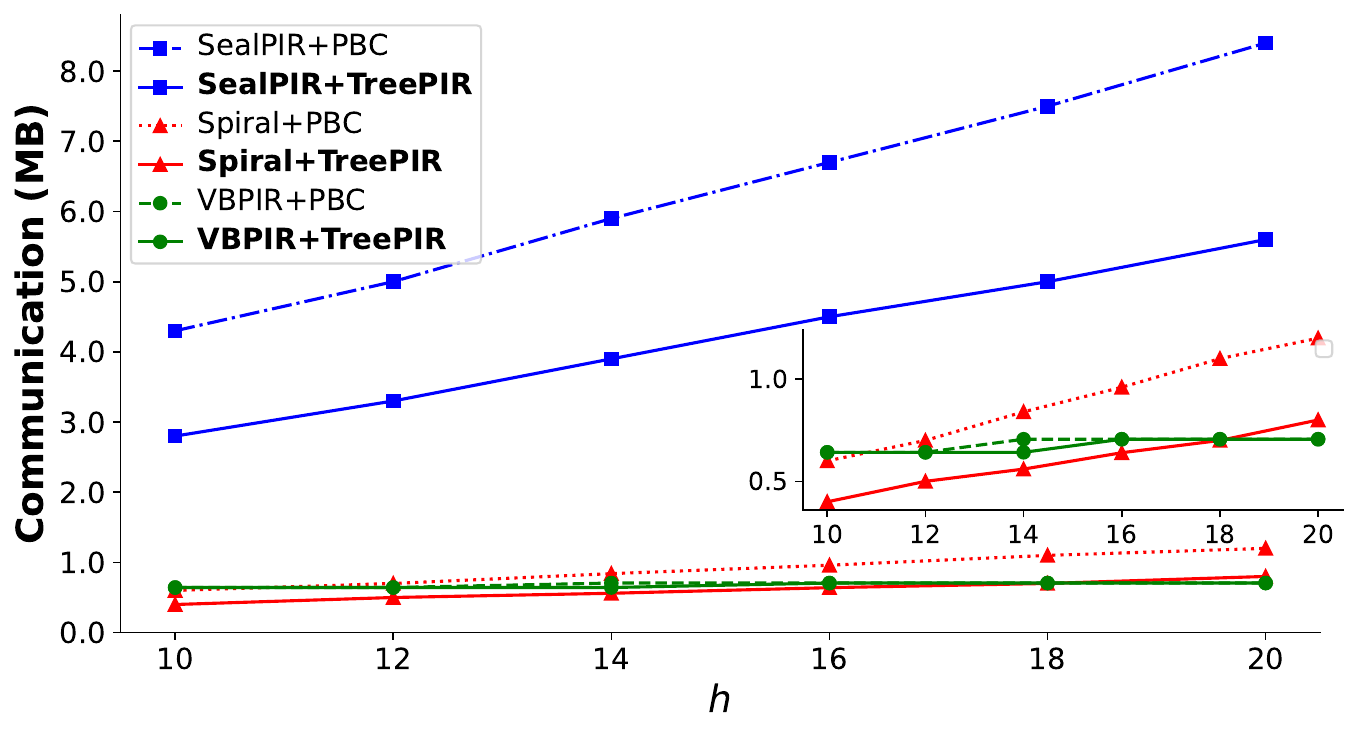}
\vspace{-10pt}
\caption{TreePIR's communication cost is about $1.5\times$ lower than PBC's as expected for most combinations (except VBPIR).}
\label{fig:PIRCom}
\end{figure}
\vspace{-10pt}


\section{Conclusions}
\label{sec:Conclusion}
\vspace{-5pt}

We consider in this work the problem of private retrieval of Merkle proofs in a Merkle tree, which has direct applications in various systems that provide data verifiability feature such as Google's Certificate Transparency, Amazon DynamoDB, and blockchains. By exploiting a unique feature of Merkle proofs, we propose an efficient retrieval scheme based on the novel concept of ancestral coloring of trees, achieving an \textit{optimal} storage overhead and much lower computational and communication complexities compared to existing schemes. 
In particular, we develop a very fast indexing algorithm that only incurs \textit{polylog} space and time complexities, which can output the required indices for a perfect Merkle tree of 64 \textit{billion} leaves in one \textit{microsecond}. By contrast, most prior works require \textit{linear} space or time for indexing.
Significant open problems include an extension of our results to $q$-ary trees (e.g. Verkle trees), which are becoming more and more popular in blockchains, and tackling sparse Merkle trees or growing Merkle trees for dynamic databases. 


\bibliographystyle{IEEEtran}
\bibliography{ParallelPrivateRetrievalMerkleProofs}

\appendices

\section{Graphs and Trees}
\label{app:graph_tree}
\vspace{-5pt}

An (undirected) \textit{graph} $G\hspace{-2pt}=\hspace{-2pt}(V,E)$ consists of a set of vertices $V$ and a set of undirected edges $E$. A \textit{path} from a vertex $v_0$ to a vertex $v_m$ in a graph $G$ is a sequence of alternating vertices and edges $\big(v_0,(v_0,v_1),(v_1,v_2),\ldots,(v_{m-1},v_m),v_m\big)$, so that no vertex appears twice. Such a path is said to have \textit{length} $m$ (there are $m$ edges in it). A graph is \textit{connected} if for every pair of vertices $u \neq v$, there is a path from $u$ to $v$. A \textit{cycle} is defined the same as a path except that the first and the last vertices are identical.
A \textit{tree} is a connected graph without any cycle.
We also refer to the vertices in a tree as its \textit{nodes}. 

A \textit{rooted} tree is a tree with a designated node referred to as its \textit{root}. 
Every node along the (unique) path from a node $v$ to the root is an \textit{ancestor} of $v$. The \textit{parent} of $v$ is the first vertex after $v$ encountered along the path from $v$ to the root. If $u$ is an ancestor of $v$ then $v$ is a \textit{descendant} of $u$. If $u$ is the parent of $v$ then $v$ is a \textit{child} of $u$. A \textit{leaf} of a tree is a node with no children. A \textit{binary tree} is a rooted tree in which every node has at most two children.
The \textit{depth} of a node $v$ in a rooted tree is 
the length of the path from the root to $v$. The \textit{height} of a rooted tree is defined as the maximum depth of a leaf. A \textit{perfect} binary tree is a binary tree in which every non-leaf node has two children and all the leaves are at the same depth. A perfect binary tree of height $h$ has $n=2^h$ leaves and $2^{h+1}-1$ nodes.

Let $[h]$ denote the set $\{1,2,\ldots,h\}$. A (node) \textit{coloring} of a tree $T=(V,E)$ with $h$ colors is a map $\phi \colon V \to [h]$ that assigns nodes to colors. The set of all tree nodes having color $i$ is called a \textit{color class}, denoted $C_i$, for $i\in [h]$.



A Merkle tree~\cite{Merkle1988} is a well-known data structure represented by a binary tree whose nodes store the cryptographic hashes (e.g. SHA-256~\cite{NISTHashFunction}) of the concatenation of the contents of their child nodes. The leaf nodes of the tree store the hashes of the data items of a database.  
In the example given in Fig.~\ref{fig:merkle_swapped}, the leaves $x_8,x_9,\ldots,x_{15}$ are hashes of eight data items $(T_i)_{i=1}^8$, i.e. $x_{i+7}=\H(T_i)$, $i\in[8]$, where $\H(\cdot)$ denotes a cryptographic hash function.
As a cryptographic hash function is collision-resistant, i.e., given $H_X = \H(X)$, it is computationally hard to find $Y \neq X$ satisfying $\H(Y)=H_X$, no change in the transactions can be made without changing the Merkle root. Thus, once the Merkle root is published, no one can modify any transaction while keeping the same root hash. 

The binary tree structure of the Merkle tree allows an efficient inclusion test: a client with a transaction, e.g., $T_3$, can verify that this transaction is indeed included in the Merkle tree with the \textit{published} root, e.g. $x_1$, by downloading the corresponding \textit{Merkle proof} $(x_{11}, x_4, x_3)$. 
A Merkle proof consists of $h=\lceil\log(n)\rceil$ hashes (excluding the root), where $h$ is the height and $n$ is the number of leaves. In this example, upon receiving a proof $\pi = (x'_{11}, x'_4, x'_3)$, the client, who has $T_3$, computes $x_{10}\gets \H(T_3)$, $x'_5 \gets \H(x_{10}||x'_{11})$, then $x'_2 \gets \H(x'_4||x'_5)$, then $x'_1\gets \H(x'_2||x'_3)$, and verifies $x'_1 \overset{?}{=} x_1$.  
The test is successful if the last equality holds.

\vspace{-5pt}
\section{Proof of Theorem~\ref{thm:main}}
\label{app:proof_thm_main}
\vspace{-5pt}

The last statement is straightforward by the definition of a balanced ancestral coloring and the description of Algorithm~\ref{algo:coloring-based-PIR}. 
We prove the first statement below.

In Algorithm~\ref{algo:coloring-based-PIR}, the nodes in the input Merkle tree are first swapped with their siblings to generate the swapped Merkle tree. The Merkle proofs in the original tree correspond precisely to the root-to-leaf paths in the swapped tree. 
An ancestral coloring ensures that $h$ nodes along every root-to-leaf path (excluding the colorless root) have different colors. Equivalently, every root-to-leaf path has identical color pattern $[1,1,\ldots,1]$, i.e. each color appears exactly once. 
This means that to retrieve a Merkle proof, the client sends exactly one PIR query to each sub-database. Hence, no information regarding the topology of the nodes in the Merkle proof is leaked as the result of the tree partition to any server.
Thus, the privacy of the retrieval of the Merkle proof reduces to the privacy of each individual PIR scheme.

More formally, in the language of batch-PIR, as the client sends $h$ independent PIR queries $\bq_1,\ldots,\bq_h$ to $h$ disjoint sub-databases indexed by $C_1,\ldots,C_h$ (as the color classes do not overlap), a probabilistic polynomial time (PPT) adversary can distinguish between queries $\bq(B)$ for $B = \{k_1,\ldots,k_h\}$ and queries $\bq(B')$ for $B'=\{k'_1,\ldots,k'_h\}$ only with negligible probability. More specifically, for every collection of $h+1$ PPT algorithms $\cA,\cA_1,\ldots,\cA_h$,
\vspace{-5pt}
\[
\begin{split}
&\big|\Pr[\bq(B) \rand \Q(1^\lambda,n,B) \colon \cA(1^\lambda,n,h,\bq(B))=1]\\ &-\Pr[\bq(B')\hspace{-2pt} \rand \hspace{-2pt} \Q(1^\lambda,n,B') \colon \cA(1^\lambda,n,h,\bq(B'))\hspace{-2pt}=\hspace{-2pt}1]\big| \\
&\leq\hspace{-2pt} \sum_{i=1}^h\hspace{-2pt}\bigg(\hspace{-2pt} \big|\hspace{-2pt}\Pr[\bq(\{k_i\})\hspace{-2pt} \rand\hspace{-2pt} \Q(\hspace{-2pt}1^\lambda\hspace{-2pt},n,\hspace{-2pt}\{k_i\}) \colon\hspace{-2pt} \cA_i(1^\lambda,n,\hspace{-2pt}h,\hspace{-2pt}\bq(\hspace{-2pt}\{k_i\}\hspace{-2pt}))\hspace{-2pt}=\hspace{-2pt}1]\\
&-\Pr[\bq(\{k'_i\})\hspace{-2pt} \rand \hspace{-2pt} \Q(1^\lambda\hspace{-2pt},n,\{k'_i\}) \colon \cA_i(1^\lambda,n,h,\bq(\{k'_i\}))\hspace{-2pt}=\hspace{-2pt}1]\big|\hspace{-2pt}\bigg)\\
&\in\negl(\lambda),
\end{split}
\]
as each term in the last sum belongs to $\negl(\lambda)$ due to the privacy of the underlying PIR (retrieving a single element), and there are $h\in \poly(\lambda)$ terms only. \qed

\section{Proofs for Section~\ref{sec:Algorithm}}
\label{app:main_proofs}

\begin{proof}[Proof of Lemma~\ref{lem:path}]
A root-to-leaf path contains exactly $h$ nodes except the root. Since these nodes are all ancestors and descendants of each other, they should have different colors. Thus, $h' \geq h$. As~$T(h)$ has precisely $2^h$ root-to-leaf paths, other conclusions follow trivially.
\end{proof}

\begin{proof}[Proof of Corollary~\ref{cr:balanced}]
This follows directly from Theorem~\ref{thm:main}, noting that a \textit{balanced} color sequence of dimension $h$ is also $h$-feasible due to Corollary~\ref{cr:balance_conf}. 
\end{proof}

\begin{proof}[Proof of Lemma~\ref{lem:if_terminate_then_success}]
As the algorithm proceeds recursively, it suffices to show that at each step, after the algorithm colors the left and right children $A$ and $B$ of the root node $R$, it will never use the color assigned to $A$ for any descendant of $A$ nor use the color assigned to $B$ for any descendant of $B$. Indeed, according to the coloring rule in \textbf{ColorSplittingRecursive$(R,h,\vec c)$} and \textbf{FeasibleSplit$(h,\vec c)$}, if $c$ is the color sequence available at the root $R$ and $c_1=2$, then after allocating Color 1 to both $A$ and $B$, there will be no more Color $1$ to assign to any node in both subtrees rooted at $A$ and $B$, and hence, our conclusion holds. Otherwise, if $2 < c_1 \leq c_2$, then  $A$ is assigned Color $1$ and according to \textbf{FeasibleSplit$(h,\vec c)$}, the color sequence $\vec a = [a_2=c_2-1,a_3,\ldots,a_h]$ used to color the descendants of $A$ no longer has Color 1. The same argument works for $B$ and its descendants. In other words, in this case, Color 1 is used exclusively for $A$ and descendants of $B$ while Color 2 is used exclusively for $B$ and descendants of $A$, which will ensure the Ancestral Property for both $A$ and $B$.  
\end{proof}

\begin{proof}[Proof of Lemma~\ref{lem:correctness}]
Note that because the Color-Splitting Algorithm starts off with the desirable color sequence $\vec c$, if it terminates successfully then the output coloring, according to Lemma~\ref{lem:if_terminate_then_success}, will be the desirable ancestral $\vec c$-coloring. Therefore, our remaining task is to show that the Color-Splitting Algorithm always terminates successfully. Both  \textbf{ColorSplittingRecursive$(R,h,\vec c)$} and \textbf{FeasibleSplit$(h,\vec c)$} work if $c_1 \geq 2$ when $h\geq 1$, i.e., when the current root node still has children below. This condition holds trivially in the beginning because the original color sequence is feasible. To guarantee that $c_1 \geq 2$ in all subsequent algorithm calls, we need to prove that starting from an $h$-feasible color sequence $\vec c$ with $h\geq 2$, \textbf{FeasibleSplit$(h,\vec c)$} always produces two $(h-1)$-feasible color sequences $\vec a$ and $\vec b$ for the two subtrees. This is the most technical and lengthy part of the proof of correctness of CSA and will be settled separately in Lemma~\ref{lem:c1=2} and Lemma~\ref{lem:c1>2}. 

In the remainder of the proof of this lemma, we analyze the time complexity of CSA. Let $C(n)$, $n = 2^h$, be the number of basic operations (e.g., assignments) required by the recursive procedure \textbf{ColorSplittingRecursive$(R,h,\vec c)$}. From its description, the following recurrence relation holds
\[
C(n) = 
\begin{cases}
2C(n/2) + \alpha h\log(h), & \text{ if } n \geq 4,\\
\beta, & \text{ if } n = 2,
\end{cases}
\]
where $\alpha$ and $\beta$ are positive integer constants and $\alpha h\log h$ is the running time of \textbf{FeasibleSplit$(h,\vec c)$} (dominated by the time required for sorting $\vec a$ and $\vec b$), noting that $h = \log_2 n$. Note that the master theorem is not applicable for this form of $C(n)$. We can apply the backward substitution method (or an induction proof) to determine $C(n)$ as follows.
\[
\begin{split}
C(n) &= 2C(n/2) + \alpha h\log h\\
&= 2\big(2C(n/2^2) + \alpha (h-1)\log(h-1)\big) + \alpha h\log h\\
&= 2^2 C(n/2^2) + \alpha\big(2(h-1)\log(h-1) + h\log h\big)\\
&=\cdots \\
&=2^k C(n/2^k) + \alpha \sum_{i=0}^{k-1}2^i(h-i)\log(h-i),
\end{split}
\]
for every $1\leq k\leq h-1$. Substituting $k = h-1$ in the above equality and noting that $C(2) = \beta$, we obtain

\begin{multline*}
C(n) = \beta 2^{h-1} + \alpha \sum_{i=0}^{h-1}2^i(h-i)\log(h-i) \\
\leq \beta 2^{h-1} + \alpha\bigg(h\sum_{i=0}^{h-1}2^i-\sum_{i=0}^{h-1}i2^i\bigg)\log h \\
= \beta 2^{h-1} + \alpha\big(h(2^h-1) - ((h-2)2^h+2)\big)\log h \\
= \beta 2^{h-1} + \alpha(2^{h+1}-h+2)\log h \in \mathcal{O}(2^{h+1}\log h) \\ 
= \mathcal{O}(n\log\log n).
\end{multline*}

Therefore, as claimed, the Color-Splitting Algorithm has a running time $O\big(2^{h+1}\log h\big)$.
\end{proof}

To complete the proof of correctness for CSA, it remains to settle Lemmas~\ref{lem:c1=2} and~\ref{lem:c1>2}, which state that 
the procedure \textbf{FeasibleSplit$(h,\vec c)$} produces two $(h-1)$-feasible color sequences $\vec a$ and $\vec b$ from a $h$-feasible color sequence $\vec c$. To this end, we first establish
Lemma~\ref{lem:aux}, which provides a crucial step in the proofs of Lemmas~\ref{lem:c1=2} and~\ref{lem:c1>2}. 
In essence, it establishes that if Condition (C1) in the feasibility definition (see~Definition~\ref{def:feasibility}) is satisfied at the two indices $m$ and $\ell$ of a non-decreasing sequence $a_1,\ldots,a_m,\ldots, a_{\ell}$, where $m < \ell$, and moreover, the elements $a_{m+1}, a_{m+2},\ldots, a_{\ell}$ differ from each other by at most one, then (C1) also holds at every \textit{middle} index $p$ $(m < p < \ell)$. This simplifies significantly the proof that the two color sequences $\vec a$ and $\vec b$ generated from an $h$-feasible color sequence $\vec c$ in the Color-Splitting Algorithm also satisfy (C1) (replacing $h$ by $h-1$), and hence, are $(h-1)$-feasible.   

\begin{lemma}
\label{lem:aux}
Suppose $0 \leq m < \ell$ and the sorted sequence of integers $(a_1,\ldots,a_\ell)$,
\[2 \leq a_1 \leq \ldots\leq a_m \leq a_{m+1}\leq \cdots \leq a_\ell,\] 
satisfies the following properties, in which $S_x \triangleq \sum_{i=1}^x a_i$,
\begin{itemize}
    \item (P1) $S_m \geq \sum_{i = 1}^m 2^i$ (trivially holds if $m=0$ since both sides of the inequality will be zero), 
    \item (P2) $S_\ell \geq \sum_{i = 1}^\ell 2^i$,
    \item (P3) $a_{m+1}=\cdots=a_{m+k}=\lfloor \frac{c}{2} \rfloor$, $a_{m+k+1}=\cdots=a_\ell=\lceil \frac{c}{2} \rceil$, for some $0\leq k \leq \ell-m$ and $c \geq 4$.
\end{itemize}
Then $S_p \geq \sum_{i=1}^p 2^i$ for every $p$ satisfying $m < p < \ell$.
\end{lemma}
\begin{proof}
See Appendix~\ref{app:aux}.
\end{proof}

As a corollary of Lemma~\ref{lem:aux}, a balanced color sequence is feasible.
Note that even if only balanced colorings are needed, CSA still needs to handle unbalanced sequences when coloring the subtrees.

\begin{corollary}[Balanced color sequence]
\label{cr:balance_conf}
For $h \geq 1$, set $u = (2^{h+1}-2) \pmod h$ and $\vec c^* = [c^*_1,c^*_2,\ldots,c^*_h]$, where $c^*_i=\left\lfloor \frac{2^{h+1}-2}{h} \right\rfloor$ if $1 \leq i \leq h-u$ and $c^*_i=\left\lceil \frac{2^{h+1}-2}{h} \right\rceil$ for $h-u+1 \leq i \leq h$. Then $\vec c^*$, referred to as the (sorted) balanced sequence, is $h$-feasible. 	
\end{corollary}
\begin{proof}
The proof follows by setting $m = 0$ in  Lemma~\ref{lem:aux}. 
\end{proof}

Lemma~\ref{lem:c1=2} and Lemma~\ref{lem:c1>2} establish that as long as $\vec c$ is an $h$-feasible color sequence, \textbf{FeasibleSplit}$(h,\vec c)$ will produce two $(h-1)$-feasible color sequences that can be used in the subsequent calls of \textbf{ColorSplittingRecursive()}.
The two lemmas settle the case $c_1=2$ and $c_1>2$, respectively, both of which heavily rely on Lemma~\ref{lem:aux}.
It is almost obvious that (C2) holds for the two sequences $\vec a$~and~$\vec b$. The condition (C1) is harder to tackle.
The main observation that helps simplify the proof is that although the two new color sequences $\vec a$ and $\vec b$ get their elements sorted, most elements $a_i$ and $b_i$ are moved around not arbitrarily but locally within a group of indices where $c_i$'s are the same - called a \textit{run}, and by establishing the condition (C1) of the feasibility for $\vec a$ and $\vec b$ at the end-points of the runs (the largest indices), one can use Lemma~\ref{lem:aux} to also establish (C1) for $\vec a$ and $\vec b$ at the middle-points of the runs, thus proving (C1) at every index. Further details can be found in the proofs of the lemmas.

\begin{lemma}
\label{lem:c1=2}
Suppose $c$ is an $h$-feasible color sequence, where $2 = c_1 \leq c_2 \leq \cdots \leq c_h$. Let $\vec a$ and $\vec b$ be two sequences of dimension $h-1$ obtained from $\vec c$ as in \textbf{FeasibleSplit$(h,\vec c)$} Case~1, before sorted, i.e., $a_2:=\lfloor c_2/2\rfloor$ and $b_2:=\lceil c_2/2\rceil$, and for $i = 3,4,\ldots,h$,
\[
\begin{cases}
a_i:=\lceil c_i/2 \rceil \text{ and } b_i:=\lfloor c_i/2 \rfloor, &\text{ if } \sum_{j=2}^{i-1}a_j < \sum_{j=2}^{i-1}b_j,\\
a_i:=\lfloor c_i/2 \rfloor \text{ and } b_i:=\lceil c_i/2 \rceil, &\text{ otherwise}.
\end{cases}
\]
Then $\vec a$ and $\vec b$, after being sorted, are $(h-1)$-feasible color sequences. 
We assume that $h\geq 2$.
\end{lemma}
\begin{proof}
See Appendix~\ref{app:c1=2}.
\end{proof}

\begin{lemma}
\label{lem:c1>2}
Suppose $c$ is an $h$-feasible color sequence, where $2 < c_1 \leq c_2 \leq \cdots \leq c_h$. Let $\vec a$ and $\vec b$ be two sequences of dimension $h-1$ obtained from $\vec c$ as in \textbf{FeasibleSplit$(h,\vec c)$} Case~2, before sorted, i.e., $a_2:=c_2-1$ and $b_2:= c_1-1$, and if $h\geq 3$ then
\[
a_3:=\left\lceil \frac{c_3+c_1-c_2}{2} \right\rceil, \quad b_3:= c_2-c_1 + \left\lfloor \frac{c_3+c_1-c_2}{2} \right\rfloor,
\]
and for $i=4,5,\ldots,h$, 
\[
\begin{cases}
a_i:=\lceil c_i/2 \rceil \text{ and } b_i:=\lfloor c_i/2 \rfloor, &\text{ if } \sum_{j=2}^{i-1}a_j < \sum_{j=2}^{i-1}b_j,\\
a_i:=\lfloor c_i/2 \rfloor \text{ and } b_i:=\lceil c_i/2 \rceil, &\text{ otherwise}.
\end{cases}
\]
Then $\vec a$ and $\vec b$, after being sorted, are $(h-1)$-feasible color sequences. Note that while $a_i$ $(2\leq i\leq h)$ and $b_i$ $(3\leq i\leq h)$ correspond to Color $i$, $b_2$ actually corresponds to Color~$1$. We assume that $h\geq 2$.
\end{lemma}
\begin{proof}
See Appendix~\ref{app:c1>2}.
\end{proof}

\section{Proof of Lemma~\ref{lem:aux}}
\label{app:aux}

\begin{proof}
We consider two cases based on the parity of $c$. 

\textbf{Case 1: $c$ is even}. In this case, from (P3) we have $a_{m+1}=\cdots=a_\ell=a$, where $a \triangleq c/2$. 
As $p > m$, we have $S_p = S_m + \sum_{i=m+1}^p a_i$.
If $S_m \geq \sum_{i=1}^p 2^i$ then the conclusion trivial holds. Otherwise, let $S_m = \sum_{i=1}^p 2^i - \delta$, for some $\delta > 0$. Then (P1) implies that 
\begin{equation}
\label{eq:delta}
0<\delta = \sum_{i=1}^p 2^i - S_m \leq \sum_{i=1}^p 2^i - \sum_{i = 1}^m 2^i = \sum_{i=m+1}^p 2^i.
\end{equation}
Moreover, from (P2) we have
\[
S_\ell = S_m + (\ell-m)a = \bigg(\sum_{i=1}^p 2^i-\delta\bigg) + (\ell-m)a \geq \sum_{i = 1}^\ell 2^i,
\]
which implies that 
\[
a \geq \frac{1}{\ell-m} \bigg( \sum_{i=p+1}^\ell 2^i + \delta \bigg).
\]
Therefore, 

\begin{multline*}
S_p = S_m + (p-m)a = \bigg(\sum_{i=1}^p 2^i - \delta\bigg) + (p-m)a \\ 
\geq \bigg(\sum_{i=1}^p 2^i - \delta\bigg) + \frac{p-m}{\ell-m}\bigg( \sum_{i=p+1}^\ell 2^i + \delta \bigg).
\end{multline*}

Hence, in order to show that $S_p \geq \sum_{i=1}^p 2^i$, it suffices to demonstrate that 
\[
(p-m)\sum_{i=p+1}^\ell 2^i \geq (\ell-p)\delta.
\]
Making use of \eqref{eq:delta}, we need to show that the following inequality holds
\[
(p-m)\sum_{i=p+1}^\ell 2^i \geq (\ell-p)\sum_{i=m+1}^p 2^i,
\]
which is equivalent to
\[
\frac{\sum_{i=p+1}^\ell 2^i}{\sum_{i=m+1}^p 2^i} 
\geq \frac{\ell-p}{p-m}
\Longleftrightarrow \frac{2^{p-m}(2^{\ell-p}-1)}{2^{p-m}-1}
\geq \frac{\ell-p}{p-m}.
\]
The last inequality holds because 
$\frac{2^{p-m}}{2^{p-m}-1} > 1 \geq \frac{1}{p-m}$
and $2^{\ell-p}-1 \geq \ell-p$ for all $0 \leq m < p < \ell$. Here we use the fact that $2^x-1-x\geq 0$ for all $x \geq 1$.
This establishes Case 1.

\textbf{Case 2: $c$ is odd}. In this case, from (P3) we have $a_{m+1}=\cdots=a_{m+k}=a$ and $a_{m+k+1}=\cdots=a_\ell=a+1$, where $a \triangleq \lfloor c/2 \rfloor$. We claim that if the inequality
\begin{equation}
\label{eq:middle}
S_{m+k} \geq \sum_{i=1}^{m+k}2^i
\end{equation}
holds then we can reduce Case 2 to Case 1. Indeed, suppose that \eqref{eq:middle} holds. Then by replacing $\ell$ by $m+k$ and applying Case 1, we deduce that $S_p \geq \sum_{i=1}^p 2^i$ for every $p$ satisfying $m < p < m+k$. Similarly, by replacing $m$ by $m+k$ and applying Case 1, the inequality $S_p \geq \sum_{i=1}^p 2^i$ holds for every $p$ satisfying $m+k < p < \ell$. Thus, in the remainder of the proof, we aim to show that \eqref{eq:middle} is correct.

To simplify the notation, set $p \triangleq m+k$. If $p=m$ or $p=\ell$ then (P1) and (P2) imply \eqref{eq:middle} trivially. Thus, we assume that $m < p < \ell$. Since $S_p = S_m + \sum_{i=m+1}^p a_i$, if $S_m \geq \sum_{i=1}^p 2^i$ then the conclusion trivial holds. Therefore, let $S_m = \sum_{i=1}^p 2^i - \delta$, for some $\delta > 0$. Similar to Case~1, (P1) implies that 
\begin{equation}
\label{eq:delta2}
0<\delta = \sum_{i=1}^p 2^i - S_m \leq \sum_{i=1}^p 2^i - \sum_{i = 1}^m 2^i = \sum_{i=m+1}^p 2^i.
\end{equation}
Moreover, from (P2) we have
\begin{multline*}
S_\ell = S_m + (\ell-m)a + (\ell-p) \\
= \bigg(\sum_{i=1}^p 2^i-\delta\bigg) + (\ell-m)a + (\ell-p) \geq \sum_{i = 1}^\ell 2^i,
\end{multline*}
which implies that 
\[
a \geq \frac{1}{\ell-m} \bigg( \sum_{i=p+1}^\ell 2^i + \delta - (\ell-p) \bigg).
\]
Therefore, 

\begin{multline*}
S_p = S_m + (p-m)a = \bigg(\sum_{i=1}^p 2^i - \delta\bigg) + (p-m)a \\
\geq \bigg(\sum_{i=1}^p 2^i - \delta\bigg) + \frac{p-m}{\ell-m}\bigg( \sum_{i=p+1}^\ell 2^i + \delta - (\ell-p)\bigg).
\end{multline*}

Hence, in order to show that $S_p \geq \sum_{i=1}^p 2^i$, it suffices to demonstrate that 
\[
(p-m)\bigg(\sum_{i=p+1}^\ell 2^i - (\ell-p)\bigg) \geq (\ell-p)\delta.
\]
Making use of \eqref{eq:delta2}, we need to show that the following inequality holds
\[
(p-m)\bigg(\sum_{i=p+1}^\ell 2^i - (\ell-p)\bigg) \geq (\ell-p)\sum_{i=m+1}^p 2^i,
\]
which is equivalent to
\begin{equation}
\label{eq:last}
\begin{split}
&\frac{\sum_{i=p+1}^\ell 2^i - (\ell-p)}{\sum_{i=m+1}^p 2^i}
\geq \frac{\ell-p}{p-m} \\
&\Longleftrightarrow \frac{2^{p-m}(2^{\ell-p}-1)- (\ell-p)/2^{m+1}}{2^{p-m}-1} \geq \frac{\ell-p}{p-m}.
\end{split}
\end{equation}

If $\ell-p=1$ then the last inequality of \eqref{eq:last} becomes 
\[
\frac{2^{p-m}-1/2^{m+1}}{2^{p-m}-1} \geq \frac{1}{p-m},
\]
which is correct because $1/2^{m+1} < 1$ and $p - m \geq 1$. 
If $\ell-p\geq 2$ then the last inequality of \eqref{eq:last} can be rewritten as
\[
\frac{2^{p-m}}{2^{p-m}-1} \bigg(2^{\ell-p}-1-\frac{\ell-p}{2^{p+1}}\bigg) 
\geq \frac{\ell-p}{p-m},
\]
which is correct because $\frac{2^{p-m}}{2^{p-m}-1} > 1 \geq \frac{1}{p-m}$ for $p > m$ and 
\[
2^{\ell-p}-1-\frac{\ell-p}{2^{p+1}} \geq 2^{\ell-p}-1-\frac{\ell-p}{4} \geq \ell-p,
\]
noting that $p \geq 1$ and that $2^x-1-\frac{5}{4}x > 0$ for every $x \geq 2$. 
This establishes Case~2.
\end{proof}

\section{Proof of Lemma~\ref{lem:c1=2}}
\label{app:c1=2}

\begin{proof}
To make the proof more readable, let $\vec a'=[a'_2,\ldots,a'_h]$ and $\vec b'=[b'_2,\ldots,b'_h]$ be obtained from $\vec a$ and $\vec b$ after their elements are sorted in non-decreasing order. According to Definition~\ref{def:feasibility}, the goal is to show that (C1) and (C2) hold for $\vec a'$ and $\vec b'$, while replacing $h$ by $h-1$. Because the sums $S_i(a)\triangleq\sum_{j=2}^i a_i$ and $S_i(b)\triangleq\sum_{j=2}^i b_i$ always differ from each other by at most one for every $2\leq i\leq h$, it is clear that at the end when $i = h$, they should be the same and equal to half of $S_h(c) \triangleq \sum_{j=2}^h c_j = \sum_{j=2}^h 2^j = 2\big(\sum_{j=1}^{h-1}2^j\big)$. Therefore, (C2) holds for $\vec a$ and $\vec b$ and hence, for $\vec a'$ and $\vec b'$ as well. It remains to show that (C1) holds for these two color sequences, that is, to prove that $\sum_{i=2}^{\ell} a'_i \geq \sum_{i=1}^{\ell-1} 2^i$ and $\sum_{i=2}^{\ell} b'_i \geq \sum_{i=1}^{\ell-1} 2^i$ for every $2 \leq \ell \leq h$ (note that $\vec a'$ and $\vec b'$ both start from index 2). 

Since $a_i$ and $b_i$ are assigned either $\lfloor c_i/2 \rfloor$ or $\lceil c_i/2 \rceil$ in somewhat an alternating manner, which keeps the sums $S_i(a)\triangleq\sum_{j=2}^i a_i$ and $S_i(b)\triangleq\sum_{j=2}^i b_i$ differ from each other by at most one, it is obvious that each sum will be approximately half of $\sum_{j=2}^i c_j$ (rounded up or down), which is greater than or equal to $\sum_{j=2}^i 2^j = 2\big( \sum_{j=1}^{i-1} 2^j \big)$. Therefore, (C1) holds for $\vec a$ and $\vec b$. However, the trouble is that this may no longer be true after sorting, in which smaller values are shifted to the front. We show below that (C1) still holds for $\vec a'$ and $\vec b'$ using Lemma~\ref{lem:aux}.

As $\vec c$ is a sorted sequence, we can partition $c_2,c_3,\ldots, c_h$ into $k$ different \textit{runs} where within each run all $c_i$'s are equal,
\begin{multline}
    \label{eq:run}
c_2 = \cdots = c_{i_1} < c_{i_1+1} = \cdots = \\
c_{i_2} < \cdots < c_{i_{k-1}+1}=\cdots =c_{i_k} \equiv c_h.
\end{multline}
For $r = 1,2,\ldots,k$, let $R_r \triangleq [i_{r-1}+1,i_r]$, where $i_0\triangleq 1$. Then \eqref{eq:run} means that for each $r\in [k]$, $c_i$'s are the same for all $i \in R_r$ and moreover, $c_i < c_{i'}$ if $i \in R_r$, $i' \in R_{r'}$, and $r < r'$. In order to show that (C1) holds for $\vec a'$ and $\vec b'$, our strategy is to first prove that (C1) holds for $\vec a'$ and $\vec b'$ at the \textit{end-points} $\ell=i_r$ of the runs $R_r$, $r \in [k]$, and then employ Lemma~\ref{lem:aux} to conclude that (C1) also holds for these color sequences at all the \textit{middle-points} of the runs. 

Since $\lfloor c_i/2 \rfloor\leq a_i \leq \lceil c_i/2 \rceil$ for every $2 \leq i \leq h$, it is clear that $a_i \leq a_{i'}$ if $i \in R_r$, $i' \in R_{r'}$, and $r < r'$. Therefore, $\vec a'$ can be obtained from $\vec a$ by sorting its elements \textit{locally} within each run. As a consequence, $\sum_{i \in R_r}a'_i = \sum_{i\in R_r} a_i$ for every $r \in [k]$, which implies that
\begin{equation}
\label{eq:Sira}
S_{i_r}(a')\triangleq \sum_{i = 2}^{i_r} a'_i = \sum_{i = 2}^{i_r} a_i \geq \sum_{i=1}^{i_r-1}2^i,
\end{equation}
where the last inequality comes from the fact that (C1) holds for $\vec a$ and $\vec b$ (as shown earlier).
The inequality \eqref{eq:Sira} implies that (C1) holds for $\vec a'$ at the end-points $i_r$, $r \in [k]$, of the runs. By Lemma~\ref{lem:aux}, we deduce that (C1) also holds for $\vec a'$ at every middle index $p \in R_r$ for all $r \in [k]$. Hence, (C1) holds for $\vec a'$, and similarly, for $\vec b'$. This completes the proof.     
\end{proof}

\section{Proof of Lemma~\ref{lem:c1>2}}
\label{app:c1>2}

Thanks to the trick of arranging $c_i$'s into \textit{runs} of elements of the same values and Lemma~\ref{lem:aux}, a proof for Lemma~\ref{lem:c1>2}, although still very lengthy, is manageable.
Before proving Lemma~\ref{lem:c1>2}, we need another auxiliary result.

\begin{lemma}
\label{lem:AB}
Let $h \geq 4$ and $\vec c$, $\vec a$, and $\vec b$ be defined as in Lemma~\ref{lem:c1>2}. 
As $\vec c$ is a sorted sequence, we can partition its elements from $c_4$ to $c_h$ into $k$ different \textit{runs} where within each run all $c_i$'s are equal,
\begin{multline*}
c_4 = \cdots = c_{i_1} < c_{i_1+1} = \cdots \\
= c_{i_2} < \cdots < c_{i_{k-1}+1}=\cdots =c_{i_k} \equiv c_h.
\end{multline*}
For any $1\leq r \leq k$, let $A_r\triangleq \sum_{j=4}^{i_r}a_j$, $B_r\triangleq \sum_{j=4}^{i_r}b_j$, and $C_r\triangleq \sum_{j=4}^{i_r}c_j$. Then $\min\{A_r,B_r\} \geq \left\lfloor \frac{C_r}{2} \right\rfloor - 1$.
\end{lemma}
\begin{proof}
First, we have $|(a_2+a_3)-(b_2+b_3)| = |\lceil\frac{c_3+c_1-c_2}{2} \rceil - \lfloor \frac{c_3+c_1-c_2}{2} \rfloor| \leq 1$. Moreover, the way that $a_i$ and $b_i$ are defined for $i\geq 4$ in Lemma~\ref{lem:c1>2} guarantees that 
\[
|(a_2+a_3+A_r)-(b_2+b_3+B_r)| \leq 1,
\]
which implies that $|A_r-B_r| \leq 2$.
Furthermore, as $a_i+b_i = c_i$ for all $i \geq 4$, we have 
\[
A_r+B_r = \sum_{j=4}^{i_r}(a_j+b_j) = \sum_{j=4}^{i_r}c_j = C_r.
\]
Thus, $\min\{A_r,B_r\} \geq \left\lfloor \frac{C_r}{2} \right\rfloor - 1$. 
\end{proof}

\begin{proof}[Proof of Lemma~\ref{lem:c1>2}]
We use a similar approach to that of Lemma~\ref{lem:c1=2}. However, the proof is more involved because we now must take into account the relative positions of $a_2$, $a_3$, $b_2$, and $b_3$ within each sequence after being sorted, and must treat $\vec a$ and $\vec b$ separately. The case with $h=2$ or $3$ is easy to verify. We assume $h\geq 4$ for the rest of the proof.

Let $\vec a'=[a'_2,\ldots,a'_h]$ and $\vec b'=[b'_2,\ldots,b'_h]$ be obtained from $\vec a$ and $\vec b$ after sorting. According to Definition~\ref{def:feasibility}, the goal is to show that (C1) and (C2) hold for these two sequences while replacing $h$ by $h-1$. The proof that (C2) holds for $\vec a'$ and $\vec b'$ is identical to that of Lemma~\ref{lem:c1=2}, implied by the following facts: first, $\sum_{i=2}^h a'_i=\sum_{i=2}^h a_i$ and $\sum_{i=2}^h b'_i=\sum_{i=2}^h b_i$, and second, due to the definitions of $a_i$ and $b_i$, the two sums $\sum_{i=2}^h a_i$ and $\sum_{i=2}^h b_i$ are exactly the same and equal to $\frac{1}{2}\big(\sum_{j=1}^h c_i - 2\big) = \frac{1}{2}\sum_{i=2}^h 2^i = \sum_{i=1}^{h-1}2^i$.  
It remains to show that (C1) holds for $\vec a'$ and $\vec b'$. 

As $\vec c$ is sorted, we partition its elements from $c_4$ to $c_h$ into $k$ different \textit{runs} where within each run all $c_i$'s are equal,
\begin{multline}
    \label{eq:run2}
c_4 = \cdots = c_{i_1} < c_{i_1+1} = \cdots = c_{i_2} < \cdots \\
< c_{i_{k-1}+1}=\cdots =c_{i_k} \equiv c_h.
\end{multline}
For $r = 1,2,\ldots,k$, define the $r$-th run as $R_r \triangleq [i_{r-1}+1,i_r]$, where $i_0\triangleq 3$. Then \eqref{eq:run2} means that for each $r\in [k]$, $c_j$'s are the same for all $j \in R_r$ and moreover, $c_j < c_{j'}$ if $j \in R_r$, $j' \in R_{r'}$, and $r < r'$. As $\vec c$ is $h$-feasible, it satisfies (C1), 
\[
    c_1+c_2+c_3+\sum_{j=4}^{i_1}c_j+\cdots+\sum_{j=i_{r-1}+1}^{i_r}c_j = \sum_{j=1}^{i_r}c_j \geq \sum_{j=1}^{i_r}2^j.
\]
Equivalently, setting $C_r \triangleq \sum_{j=4}^{i_1}c_j+\cdots+\sum_{j=i_{r-1}+1}^{i_r}c_j$, we have 
\begin{equation}
    \label{eq:Ca}
    c_1+c_2+c_3+C_r\geq \sum_{j=1}^{i_r}2^j.
\end{equation}
We will make extensive use of this inequality later.

In order to show that (C1) holds for $\vec a'$ and $\vec b'$, our strategy is to first prove that (C1) holds for these sequences at the \textit{end-points} $\ell=i_r$ of the runs $R_r$, $r \in [k]$, and then employ Lemma~\ref{lem:aux} to conclude that (C1) also holds for these sequences at all the \textit{middle-points} of the runs. 
We also need to demonstrate that (C1) holds for $\vec a'$ and $\vec b'$ at the indices of $a_2$, $a_3$, $b_2$, and $b_3$ within these sorted sequences.
Notice the differences with the proof of Lemma~\ref{lem:c1=2}: first, we consider the runs from $c_4$ instead of $c_2$, and second, we must take into account the positions of $a_2,a_3,b_2,b_3$ relative to the runs within the sorted sequence $\vec a'$ and $\vec b'$.

Since $\lfloor c_i/2 \rfloor\leq a_i \leq \lceil c_i/2 \rceil$ for every $2 \leq i \leq \ell$, it is clear that $a_i \leq a_{i'}$ if $i \in R_r$, $i' \in R_{r'}$, and $r < r'$. Therefore, $\vec a'$ can be obtained from $\vec a$ by sorting its elements $a_4,\ldots,a_h$ \textit{locally} within each run and then inserting $a_2$ and $a_3$ into their correct positions (to make $\vec a'$ non-decreasing). The same conclusion holds for $\vec b'$. 

Note also that $a_2$ and $a_3$ are inserted between runs (unless they are the first or the last element in $\vec a'$) and do not belong to any run. The same statement holds for $b_2$ and $b_3$.
 
\textbf{First, we show that $\vec a'$ satisfies (C1)}. We divide the proof into two cases depending on whether $a_2 \leq a_3$ or not.

    \textbf{(Case a1) $a_2 \leq a_3$}. The sorted sequence $\vec a'$ has the following format in which within each run $R_r$ are the elements $a_j$, $j \in R_r$, ordered so that those $a_j = \left\lfloor \frac{c_j}{2} \right\rfloor$ precede those $a_j = \left\lceil \frac{c_j}{2} \right\rceil$. Note that $c_j$'s are equal for all $j\in R_r$. 
    \[
    \underbracket{\hspace{20pt}\text{runs}\hspace{20pt}}_{}\ a_2\ \underbracket{\hspace{20pt}\text{runs}\hspace{20pt}}_{}\ a_3\ \underbracket{\hspace{20pt}\text{runs}\hspace{20pt}}_{}
    \]
    Note that it is possible that there are no runs before $a_2$, or between $a_2$ and $a_3$, or after $a_3$.
    
    \textit{In the first sub-case}, the index of interest $\ell=i_r$ is smaller than the index of $a_2$ in~$\vec a'$. 
    In order to show that (C1) holds for $\vec a'$ at $\ell=i_r$, we prove that
    \begin{equation}
    \label{eq:A}
    A_r \triangleq \sum_{j=4}^{i_r}a_j \geq \sum_{j=1}^{i_r-3} 2^j,
    \end{equation}
    noting that in this sub-case the set $\{a_j \colon 4 \leq j \leq i_r\}$ corresponds precisely to the set of the first $i_r-3$ elements in $\vec a'$.
    We now demonstrate that \eqref{eq:A} can be implied from \eqref{eq:Ca}. Indeed, since $\vec c$ is non-decreasing, from \eqref{eq:Ca} we deduce that 
    \[
    4C_r \geq c_1+c_2+c_3+C_r\geq \sum_{j=1}^{i_r}2^j.
    \]
    Combining this with Lemma~\ref{lem:AB}, we obtain the desired inequality \eqref{eq:A} as follows.
    \begin{multline*}
    A_r \geq \left\lfloor \frac{C_r}{2} \right\rfloor - 1
    \geq \left\lfloor\frac{1}{8}\sum_{j=1}^{i_r}2^j\right\rfloor-1 \\
    =\left\lfloor\sum_{j=1}^{i_r-3}2^j + \frac{14}{8}\right\rfloor-1 = \sum_{j=1}^{i_r-3}2^j.
    \end{multline*}
    
    \textit{In the second sub-case}, the index of interest $\ell$ is greater than or equal to the index of $a_2$ but smaller than that of $a_3$ in~$\vec a'$. In other words, either $\ell=i_r$ is greater than the index of $a_2$ or $\ell$ is precisely the index of~$a_2$ in $\vec a'$. If the latter occurs, let $i_r$ be the end-point of the run preceding $a_2$ in $\vec a'$. To prove that (C1) holds for $\vec a'$ at $\ell$, in both cases we aim to show that
    \begin{equation}
    \label{eq:Aa2}
    a_2+A_r \triangleq a_2+\sum_{j=4}^{i_r}a_j \geq \sum_{j=1}^{i_r-2} 2^j,
    \end{equation}
    noting that $\{a_2\}\cup \{a_j\colon 4 \leq j \leq i_r\}$ forms the set of the first $i_r-2$ elements in $\vec a'$.
    We demonstrate below that \eqref{eq:Aa2} is implied by \eqref{eq:Ca}. \textit{First}, since $\vec c$ is non-decreasing, from \eqref{eq:Ca} we deduce that 
    \[
    (c_1+c_2)+2C_r \geq (c_1+c_2)+(c_3+C_r) \geq \sum_{j=1}^{i_r}2^j,
    \]
    which implies that 
    \begin{equation}
    \label{eq:case_a1_2}
    \frac{c_1+c_2}{4} + \frac{C_r}{2} \geq \frac{1}{4}\sum_{j=1}^{i_r}2^j = \sum_{j=1}^{i_r-2}2^j + \frac{3}{2}.
    \end{equation}
    \textit{Next}, since $c_2\geq c_1 \geq 3$, we have
    \[
    a_2 = c_2 - 1 > \frac{c_2}{2} = \frac{2c_2}{4} \geq \frac{c_1+c_2}{4}.
    \]
    \textit{Then}, by combining this with Lemma~\ref{lem:AB} and \eqref{eq:case_a1_2}, we obtain the following inequality
    \begin{multline*}
    a_2+A_r > \frac{c_1+c_2}{4} + \bigg(\left\lfloor \frac{C_r}{2} \right\rfloor - 1\bigg) 
    \geq \frac{c_1+c_2}{4} + \frac{C_r-1}{2} - 1 \\
    = \frac{c_1+c_2}{4} + \frac{C_r}{2} - \frac{3}{2}
    \geq \sum_{j=1}^{i_r-2}2^j.
    \end{multline*}
    Thus, \eqref{eq:Aa2} follows.
    
    \textit{In the third sub-case}, the index of interest $\ell$ is greater than or equal to the index of~$a_3$. In other words, either $\ell=i_r$ is greater than the index of $a_3$ or $\ell$ is precisely the index of~$a_3$ in $\vec a'$. If the latter occurs, let $i_r$ be the end-point of the run preceding $a_3$ in $\vec a'$. To prove that (C1) holds for $\vec a'$ at $\ell$, in both cases we aim to show that
    \begin{equation}
    \label{eq:Aa2a3}
    a_2+a_3+\sum_{j=4}^{i_r}a_j \geq \sum_{j=1}^{i_r-1} 2^j,
    \end{equation}
    noting that in this sub-case $\{a_2,a_3\}\cup \{a_j\colon 4 \leq j \leq i_r\}$ forms the set of the first $i_r-1$ elements in $\vec a'$.
    This turns out to be the easiest sub-case. With the presence of both $a_2$ and $a_3$ in the sum on the left-hand side of \eqref{eq:Aa2a3}, according to the way $a_i$ and $b_i$ are selected in Lemma~\ref{lem:c1>2}, we have
    \[
    \left|\left(a_2+a_3+\sum_{j=4}^{i_r}a_j\right)-\left(b_2+b_3+\sum_{j=4}^{i_r}b_j\right)\right| \leq 1,
    \]
    and since $a_2+a_3+b_2+b_3=c_1+c_2+c_3-2$ and $a_j+b_j=c_j$ for $j\geq 4$, we also have
    \begin{multline*}
    \left(a_2+a_3+\sum_{j=4}^{i_r}a_j\right)+\left(b_2+b_3+\sum_{j=4}^{i_r}b_j\right) \\
    = \sum_{j=1}^{i_r}c_j-2
    \geq \sum_{j=2}^{i_r}2^j
    = 2\sum_{j=1}^{i_r-1} 2^j,
    \end{multline*}
    which, together, imply \eqref{eq:Aa2a3}.
    \textbf{(Case a2) $a_2 > a_3$}. The sorted sequence $\vec a'$ has the following format.
    \[
    \underbracket{\hspace{20pt}\text{runs}\hspace{20pt}}_{}\ a_3\ \underbracket{\hspace{20pt}\text{runs}\hspace{20pt}}_{}\ a_2\ \underbracket{\hspace{20pt}\text{runs}\hspace{20pt}}_{}
    \]
    Note that it is possible that there are no runs before $a_3$, or between $a_3$ and $a_2$, or after $a_2$.
    
    Again, due to Lemma~\ref{lem:aux}, we only need to demonstrate that (C1) holds for $\vec a'$ at the end-point~$i_r$ of each run $R_r$, $1\leq r \leq k$, \textit{and} at $a_2$ and $a_3$. 
    
    \textit{In the first sub-case}, $\ell=i_r$ is smaller than the index of $a_3$ in~$\vec a'$. As $a_2$ and $a_3$ are not involved, the same proof as in the first-subcase of Case a1 applies. 
    
    \textit{In the second sub-case}, the index of interest $\ell$ is greater than or equal to the index of $a_3$ but smaller than that of $a_2$ in~$\vec a'$. In other words, either $\ell=i_r$ is greater than the index of~$a_3$ and smaller than that of $a_2$, or $\ell$ is precisely the index of $a_3$ in $\vec a'$. If the latter occurs, let~$i_r$ be the end-point of the run preceding $a_3$ in $\vec a'$.  To prove that (C1) holds for $\vec a'$ at $\ell$, we aim to show
    \begin{equation}
    \label{eq:Aa3}
    a_3+A_r \triangleq a_3+\sum_{j=4}^{i_r}a_j \geq \sum_{j=1}^{i_r-2} 2^j,
    \end{equation}
    noting that in this sub-case $\{a_3\}\cup \{a_j\colon 4 \leq j \leq i_r\}$ forms the set of the first $i_r-2$ elements in $\vec a'$.
    We demonstrate below that \eqref{eq:Aa3} can be implied from \eqref{eq:Ca} in both cases when $i_r = 4$ and $i_r > 4$. First, assume that $i_r = 4$, i.e, $r=1$ and the run $R_1$ consists of only one index $4$. Now, \eqref{eq:Ca} can be written as
    \[
    c_1+c_2+c_3+c_4 \geq \sum_{j=1}^4 2^j = 30,
    \]
    and what we need to prove is
    \[
    a_3 + \left\lfloor \frac{c_4}{2} \right\rfloor \geq \sum_{j=1}^{2} 2^j = 6,
    \]
    noting that the element in $\vec a'$ corresponding to $c_4$ is either $\lfloor c_4/2 \rfloor$ or $\lceil c_4/2 \rceil$.
    Equivalently, plugging in the formula for $a_3$, what we aim to show is
    \begin{equation}
    \label{eq:c1234}
    \left\lceil \frac{c_3+c_1-c_2}{2} \right\rceil + \left\lfloor \frac{c_4}{2} \right\rfloor \geq 6.
    \end{equation}
    This inequality is correct because 
    \[
    \left\lceil \frac{c_3+c_1-c_2}{2} \right\rceil + \left\lfloor \frac{c_4}{2} \right\rfloor
    \geq \left\lceil \frac{c_1}{2} \right\rceil + \left\lfloor \frac{c_4}{2} \right\rfloor
    \geq 6,
    \]
    where the first inequality holds because $c_3 \geq c_2$ and the second inequality holds because $c_1\geq 3$ and $c_4\geq 8$, given that $c_4\geq c_3 \geq c_2 \geq c_1$ and $c_1+c_2+c_3+c_4\geq 30$. 
    To complete this sub-case, we assume that $i_r > 4$. In this scenario, $C_r$ in \eqref{eq:Ca} has at least two terms $c_j$'s, which are all greater than or equal to $c_2$ and $c_3$, and hence, $C_r \geq c_2+c_3$. Therefore, using Lemma~\ref{lem:AB}, we have 
    \begin{multline*}
    a_3+A_r = \left\lceil \frac{c_3+c_1-c_2}{2} \right\rceil + A_r 
    \geq \left\lceil \frac{c_1}{2} \right\rceil + \left(\left\lfloor \frac{C_r}{2} \right\rfloor - 1\right) \\
    > \frac{c_1}{4}+\left( \frac{C_r-1}{2}-1 \right)
    = \frac{1}{4}(c_1+2C_r)-\frac{3}{2} \\
    \geq \frac{1}{4}(c_1+c_2+c_3+C_r)-\frac{3}{2}
    \geq \frac{1}{4}\sum_{j=1}^{i_r}2^j-\frac{3}{2} =\sum_{j=1}^{i_r-2}2^j.
    \end{multline*}
    
    Thus, \eqref{eq:Aa3} follows. 
    
    \textit{In the third sub-case}, assume that the index of interest $\ell$ is greater than or equal to the index of $a_2$. As both $a_2$ and $a_3$ are involved, the proof goes exactly the same way as in the third sub-case of Case 1a. Thus, we have shown that $\vec a'$ also satisfies (C1) and is $(h-1)$-feasible.


\textbf{Next, we show that $\vec b'$ satisfies (C1)}. The proof is very similar to that for $\vec a'$. We divide the proof into two cases depending on whether $b_2 \leq b_3$ or not.

    \textbf{(Case b1) $b_2 \leq b_3$}. The sorted sequence $\vec b'$ has the following format in which within each run $R_r$ are the elements $b_j$, $j \in R_r$, ordered so that those $b_j = \left\lfloor \frac{c_j}{2} \right\rfloor$ precede those $b_j = \left\lceil \frac{c_j}{2} \right\rceil$. Note that $c_j$'s are equal for all $j\in R_r$. 
    \[
    \underbracket{\hspace{20pt}\text{runs}\hspace{20pt}}_{}\ b_2\ \underbracket{\hspace{20pt}\text{runs}\hspace{20pt}}_{}\ b_3\ \underbracket{\hspace{20pt}\text{runs}\hspace{20pt}}_{}
    \]
    It is possible that there are no runs before $b_2$, or between $b_2$ and $b_3$, or after $b_3$.
    
    The proof for the first and the third sub-cases are exactly the same as for $\vec a'$. We only need to consider the second sub-case, in which the index $\ell$ is greater than or equal to the index of $b_2$ but smaller than that of $b_3$ in~$\vec b'$. In other words, either $\ell=i_r$ is greater than the index of $b_2$ or $\ell$ is precisely the index of $b_2$ in $\vec a'$. If the latter occurs, let $i_r$ be the end-point of the run preceding $b_2$ in $\vec b'$. To prove that (C1) holds for $\vec b'$ at $\ell$, we aim to show that
    \begin{equation}
    \label{eq:Bb2}
    b_2+B_r \triangleq b_2+\sum_{j=4}^{i_r}b_j \geq \sum_{j=1}^{i_r-2} 2^j,
    \end{equation}
    noting that in this sub-case $\{b_2\}\cup \{b_j\colon 4 \leq j \leq i_r\}$ forms the set of the first $i_r-2$ elements in $\vec b'$.
    We demonstrate below that \eqref{eq:Bb2} can be implied from \eqref{eq:Ca} in both cases when $i_r = 4$ and $i_r > 4$. First, assume that $i_r = 4$, i.e, $r=1$ and the run $R_1$ consists of only one index $4$. Now, \eqref{eq:Ca} can be written as\vspace{-5pt}
    \[
    c_1+c_2+c_3+c_4 \geq \sum_{j=1}^4 2^j = 30,
    \]
    and what we need to prove is\vspace{-10pt}
    \[
    b_1 + \left\lfloor \frac{c_4}{2} \right\rfloor \geq \sum_{j=1}^{2} 2^j = 6,
    \]
    noting that the element in $\vec b'$ corresponding to $c_4$ is either $\lfloor c_4/2 \rfloor$ or $\lceil c_4/2 \rceil$.
    Equivalently, plugging in the formula for $b_2$, what we aim to show is
    \[
    (c_1-1) + \left\lfloor \frac{c_4}{2} \right\rfloor \geq 6,
    \]
    which is correct because $c_1\geq 3$ and $c_4 \geq 8$, given that $c_4\geq c_3 \geq c_2 \geq c_1$ and $c_4+c_3+c_2+c_1~\geq~30$. 
    To complete this sub-case, we assume that $i_r > 4$. In this scenario, $C_r$ in \eqref{eq:Ca} has at least two terms $c_j$'s, which are all greater than or equal to $c_2$ and $c_3$, and hence, $C_r \geq c_2+c_3$. Therefore, using Lemma~\ref{lem:AB}, we have 
    \begin{multline*}
    b_2+B_r = (c_1-1) + B_r 
    > \frac{c_1}{4} + \left(\left\lfloor \frac{C_r}{2} \right\rfloor - 1\right) \\
    \geq \frac{c_1}{4}+\left( \frac{C_r-1}{2}-1 \right)
    = \frac{c_1}{4}+\left( \frac{C_r}{2}-\frac{3}{2} \right) \\
    = \frac{1}{4}(c_1+2C_r)-\frac{3}{2} \geq \frac{1}{4}(c_1+c_2+c_3+C_r)-\frac{3}{2} \\
    \geq \frac{1}{4}\sum_{j=1}^{i_r}2^j-\frac{3}{2}
    = \sum_{j=1}^{i_r-2}2^j.
    \end{multline*}
    Thus, \eqref{eq:Bb2} follows. 
    \textbf{(Case b2) $b_2 > b_3$}. The sorted sequence $\vec b'$ has the following format.
    \[
    \underbracket{\hspace{20pt}\text{runs}\hspace{20pt}}_{}\ b_3\ \underbracket{\hspace{20pt}\text{runs}\hspace{20pt}}_{}\ b_2\ \underbracket{\hspace{20pt}\text{runs}\hspace{20pt}}_{}
    \]
    Note that it is possible that there are no runs before $b_3$, or between $b_3$ and $b_2$, or after $b_2$.
    
    Due to Lemma~\ref{lem:aux}, we only need to demonstrate that (C1) holds for $\vec b'$ at the end-point~$i_r$ of each run $R_r$, $1\leq r \leq k$, \textit{and} at $b_2$ and $b_3$. Similar to Case b1, we only need to investigate the second sub-case when the index of interest $\ell$ is greater than or equal to the index of $b_3$ but smaller than that of $b_2$ in~$\vec b'$. In other words, either $\ell=i_r$ is greater than the index of $b_3$ and smaller than that of $b_2$, or $\ell$ is precisely the index of $b_3$ in $\vec b'$. If the latter occurs, let $i_r$ be the end-point of the run preceding $b_3$ in $\vec b'$. To prove that (C1) holds for $\vec b'$ at $\ell$, we aim to show that
    \begin{equation}
    \label{eq:Bb3}
    b_3+B_r \triangleq b_3+\sum_{j=4}^{i_r}b_j \geq \sum_{j=1}^{i_r-2} 2^j,
    \end{equation}
    noting that in this sub-case $\{b_3\}\cup \{b_j\colon 4 \leq j \leq i_r\}$ forms the set of the first $i_r-2$ elements in $\vec b'$.
    We demonstrate below that \eqref{eq:Bb3} can be implied from \eqref{eq:Ca} in both cases when $i_r = 4$ and $i_r > 4$. First, assume that $i_r = 4$, i.e, $r=1$ and the run $R_1$ consists of only one index $4$. Now, \eqref{eq:Ca} can be written as
    \[
    c_1+c_2+c_3+c_4 \geq \sum_{j=1}^4 2^j = 30,
    \]
    and what we need to prove is
    \[
    b_3 + \left\lfloor \frac{c_4}{2} \right\rfloor \geq \sum_{j=1}^{2} 2^j = 6,
    \]
    noting that the element in $\vec b'$ corresponding to $c_4$ is either $\lfloor c_4/2 \rfloor$ or $\lceil c_4/2 \rceil$.
    Equivalently, plugging in the formula for $b_3$, what we aim to show is
    \[
    c_2-c_1+\left\lfloor \frac{c_3+c_1-c_2}{2} \right\rfloor + \left\lfloor \frac{c_4}{2} \right\rfloor \geq 6.
    \]
    This inequality is correct because 
    \begin{multline*}
    \quad c_2-c_1+\left\lfloor \frac{c_3+c_1-c_2}{2} \right\rfloor + \left\lfloor \frac{c_4}{2} \right\rfloor \\
    \geq c_2-c_1 + \frac{c_3+c_1-c_2-1}{2} + \frac{c_4-1}{2} \\
    \geq \frac{c_3+c_4}{2}-1 > 6,
    \end{multline*}
    where the last inequality holds because $c_3+c_4 \geq 16$, given that $c_4\geq c_3\geq c_2\geq c_1$ and $c_1+c_2+c_3+c_4\geq 30$. 
    To complete this sub-case, we assume that $i_r > 4$. In this scenario, $C_r$ in \eqref{eq:Ca} has at least two terms $c_j$'s, which are all greater than or equal to $c_1$ and $c_2$, and hence, $C_r \geq c_1+c_2$. Therefore, using Lemma~\ref{lem:AB} and \eqref{eq:Ca}, we have 
    \begin{multline*}
    b_3+B_r = \left(c_2-c_1+\left\lfloor \frac{c_3+c_1-c_2}{2} \right\rfloor\right) + B_r \\
    \geq \frac{c_2+c_3-c_1-1}{2} + \left(\left\lfloor \frac{C_r}{2} \right\rfloor - 1\right) \\
    \geq \frac{c_3-1}{2}+\left( \frac{C_r-1}{2}-1 \right)
    > \frac{c_3}{4} + \frac{C_r}{2}-\frac{3}{2} \\
    \geq \frac{1}{4}(c_3+2C_r)-\frac{3}{2} \geq \frac{1}{4}(c_1+c_2+c_3+C_r)-\frac{3}{2} \\
    \geq \frac{1}{4}\sum_{j=1}^{i_r}2^j-\frac{3}{2} 
    =\sum_{j=1}^{i_r-2}2^j.
    \end{multline*}
    Hence, \eqref{eq:Bb3} follows. Thus, we have shown that $\vec b'$ also satisfies (C1) and is $(h-1)$-feasible.
\end{proof}

\section{Connection to Combinatorial Batch Codes}
\label{app:batch}

Balanced ancestral colorings of perfect binary trees brings in two new dimensions to batch codes~\cite{ishai2004}: \textit{patterned} batch retrieval (instead of arbitrary batch retrieval as often considered in the literature) and \textit{balanced} storage capacity across servers (and generalized to \textit{heterogeneous} storage capacity).

\begin{figure}[htb!]
\centering
\includegraphics[scale=0.25]{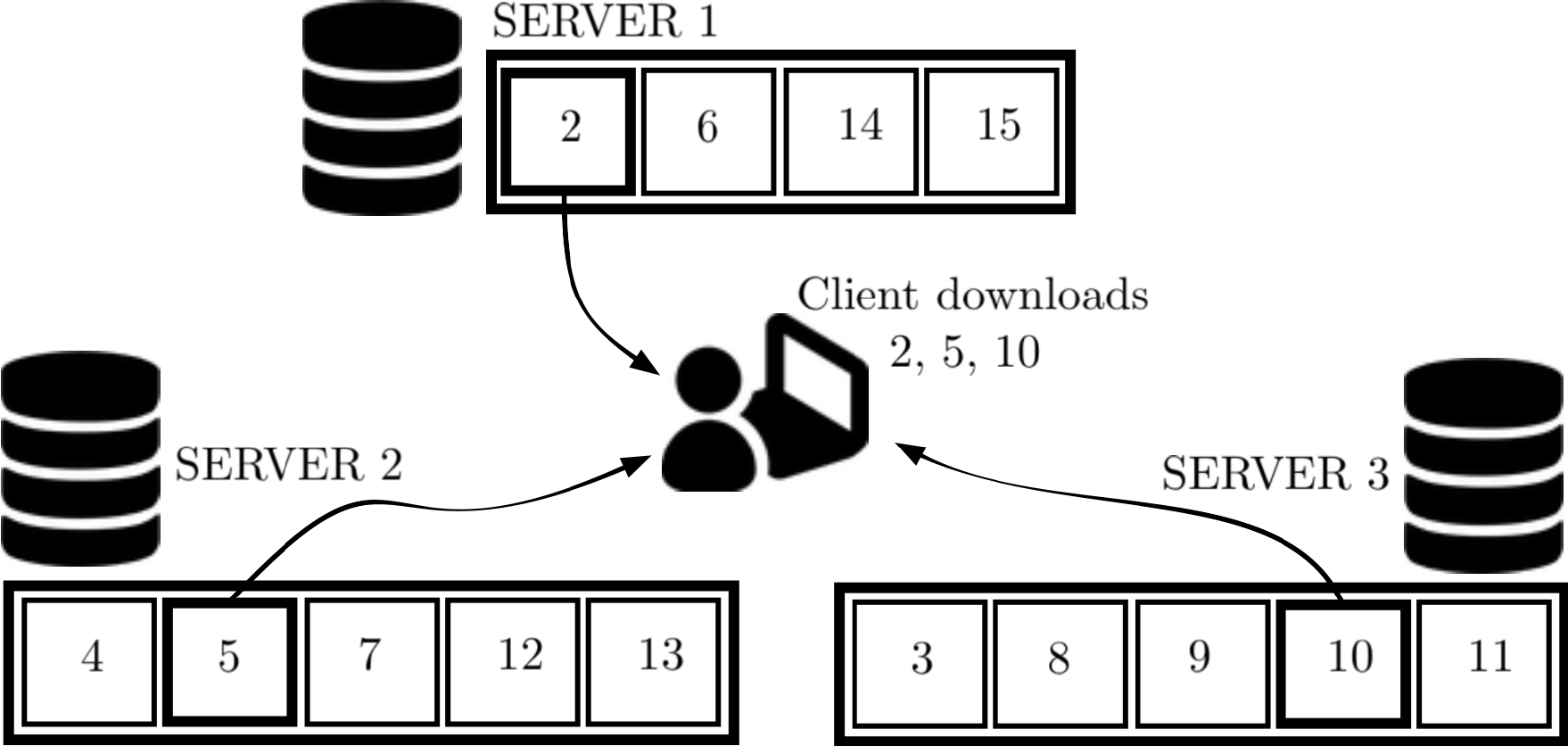}
\caption{An illustration of a \textit{balanced} combinatorial patterned-batch code with a \textit{minimum} total storage overhead (no redundancy), which is spread out evenly among three servers. The set of nodes at each server corresponds to tree nodes in a color class of a balanced ancestral coloring of $T(3)$ as given in Fig.~\ref{fig:toy}. A client only needs to download one item from each server for any batch request following a root-to-leaf-path pattern in $\mathcal{P}=\{\{2,4,8\}$, $\{2,4,9\}$, $\{2,5,10\}$, $\{2,5,11\}$, $\{3,6,12\}$, $\{3,6,13\}$, $\{3,7,14\}$,  $\{3,7,15\}\}$.}
\label{fig:batchcode}
\end{figure}

Given a set of $N$ distinct nodes, an $(N, S, h, m)$ \textit{combinatorial batch code} (CBC) provides a way to assign $S$ copies of nodes to $m$ different servers so that to retrieve a set of $h$ distinct nodes, a client can download at most one item from each server~\cite{ishai2004}. The goal is to minimize the total storage capacity required $S$. For instance, labelling $N = 7$ nodes from 1 to 7, and use $m=5$ servers, an $(N=7, S=15, h = 5, m = 5)$ CBC allocates to five servers the sets $\{1,6,7\}, \{2,6,7\},\{3,6,7\},\{4,6,7\},\{5,6,7\}$ (Paterson \textit{et al.,}~\cite{Paterson2009}). One can verify that an arbitrary set of $h=5$ nodes can be retrieved by collecting one item from each set. Here, the storage requirement is $S = 15 > N = 7$. It has been proved (see, e.g., \cite{Paterson2009}) that when $h = m$, the minimum possible $N$ is in the order of $\Theta(Nh)$, which implies that the average replication factor across $N$ nodes is $\Theta(h)$.

It turns out that $S=N$ (equivalently, replication factor one or \textit{no replication}) is achievable if the $h$ retrieved nodes are \textit{not} arbitrary and follow a special \textit{pattern}. More specifically, if the $N$ nodes can be organised as nodes in a perfect binary tree $T(h)$ except the root and only sets of nodes along root-to-leaf paths (ignoring the root) are to be retrieved, then each item only needs to be stored in one server and requires no replication across different servers, which implies that $S = N$. A trivial solution is to use layer-based sets, i.e., Server $i$ stores nodes in Layer $i$ of the tree, $i = 1,2,\ldots,h$. If we also want to balance the storage across the servers, that is, to make sure that the number of nodes assigned to every server differs from each other by at most one, then the layer-based solution no longer works and a balanced ancestral coloring of the tree nodes (except the root) will be required (see~Fig.~\ref{fig:batchcode} for an example). We refer to this kind of code as \textit{balanced combinatorial patterned-batch codes}. The balance requirement can also be generalized to the \textit{heterogeneous} setting in which different servers may have different storage capacities. It is an intriguing open problem to discover other patterns (apart from the root-to-leaf paths considered in this paper) so that (balanced or heterogeneous) combinatorial patterned-batch codes with optimal total storage capacity $S = N$ exist.

\vspace{-5pt}
\section{Connection to Equitable Colorings.}
\label{app:equitable}

By finding balanced ancestral colorings of perfect binary trees, we are also able to determine the \textit{equitable chromatic number} of a new family of graphs. 

\begin{figure}[htb!]
\centering
\includegraphics[scale=0.4]{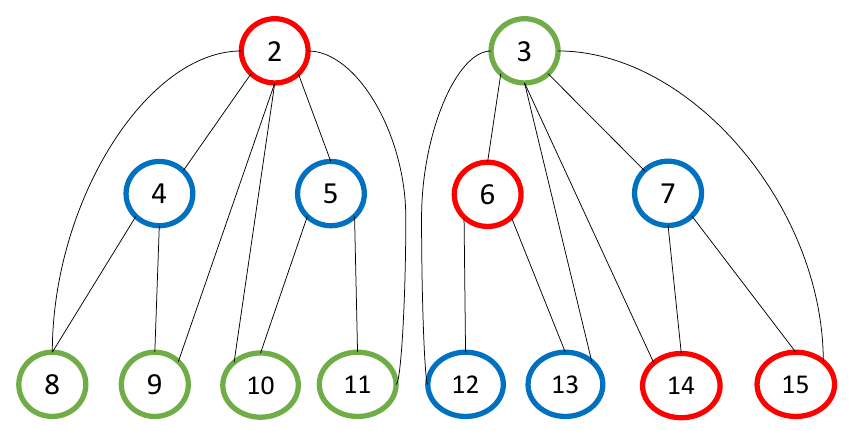}
\caption{An equitable coloring of $T'(3)$ using three colors (Red, Green, and Blue). The graph $T'(3)$ is a \textit{disconnected} graph obtained from $T(3)$ by first removing its root and then by adding edges between every node and all of its descendants.}
\label{fig:equitable}
\end{figure}

An undirected graph $G$ (connected or not) is said to be equitable $k$-colorable if its vertices can be colored with $k$ colors such that no adjacent vertices have the same color and moreover, the sizes of the color classes differ by at most one~\cite{Meyer1973}. The \textit{equitable chromatic number} $\chi_{=}(G)$ is the smallest integer $k$ such that $G$ is equitable $k$-colorable. There has been extensive research in the literature on the equitable coloring of graphs (see, for instance, \cite{Lih2013}, for a survey), noticeably, on the Equitable Coloring Conjecture, which states that for a connected graph $G$ that is neither a complete graph nor an odd cycle, then $\chi_=(G) \leq \Delta(G)$, where $\Delta(G)$ is the maximum degree of a vertex in $G$. Special families of graphs that support the Equitable Coloring Conjecture include trees, bipartite graphs, planar and outerplanar graphs, among others.

The existence of balanced ancestral colorings of the perfect binary tree $T(h)$ established in this work immediately implies that $\chi_=(T'(h)) \leq h$, where $T'(h)$ is a disconnected graph obtained from $T(h)$ by first removing its root and then by adding edges between every node $u$ and all of its descendants (see Fig.~\ref{fig:equitable} for a toy example of $T'(3)$). Note that $T'(h)$ is neither a tree, a bipartite graph, a planar graph, nor an outerplanar graph for $h \geq 5$ (because it contains the complete graph $K_h$ as a subgraph). 
On top of that, it is not even connected and so the Equitable Coloring Conjecture doesn't apply (and even if it applied, the maximum vertex degree is $\Delta(T'(h)) = 2^h-2$, which is very far from $h$ and therefore would give a very loose bound).
Furthermore, as $h$ nodes along a path from node 2 (or 3) to a leaf are ancestors or descendants of each other (hence forming a complete subgraph $K_h$), they should all belong to different color classes, which implies that $\chi_=(T'(h)) \geq h$. Thus, $\chi_=(T'(h)) = h$. Moreover, a balanced ancestral coloring of $T(h)$ provides an equitable $h$-coloring of $T'(h)$. To the best of our knowledge, the equitable chromatic number of this family of graphs  has never been discovered before.

We continue the discussion of the connections with the theory of majorization, e.g., on vertex coloring of claw-free graphs~\cite{deWerra1985} and edge coloring~\cite{FolkFulkerson1969} as below.

\section{Connection to the Theory of Majorization}
\label{app:majorization}

We discuss below the connection of some of our results to the theory of majorization~\cite{Marshall2011}.

\begin{definition}[Majorization]
    For two vectors $\vec{x}=[x_1,\ldots,x_h]$ and $\vec{y}=[y_1,\ldots,y_h]$ in $\mathbb{R}^h$ with $x_1\geq\ldots\geq x_h$ and $y_1\geq\ldots\geq y_h$, $\vec{y}$ majorizes $\vec{x}$, denoted by $\vec{x}\prec \vec{y}$, if the following conditions hold.
    \begin{itemize}
      \item $\sum_{i=1}^\ell y_i \geq \sum_{i=1}^\ell x_i$, for every $1\leq \ell \leq h$, and
      \item $\sum_{i=1}^h y_i = \sum_{i=1}^h x_i$.
    \end{itemize}
    Equivalently, when $\vec{x}$ and $\vec{y}$ are sorted in non-decreasing order, i.e., $x_1\leq\ldots\leq x_h$ and $y_1\leq\ldots\leq y_h$, $\vec{x}\prec \vec{y}$ if the following conditions hold.
    \begin{itemize}
      \item $\sum_{i=1}^\ell y_i \leq \sum_{i=1}^\ell x_i$, for every $1\leq \ell \leq h$, and
      \item $\sum_{i=1}^h y_i = \sum_{i=1}^h x_i$.
    \end{itemize}
\end{definition}

Now, let $\vec c_h =[2,2^2,\ldots, 2^h]$. Clearly, $\vec{c}$ is $h$-feasible (according to Definition~\ref{def:feasibility}) if and only if all elements of $\vec c$ are integers and $\vec c\prec \vec c_h$. From the perspective of majorization, we now reexamine some of our results and make connections with existing related works. 

\textit{First}, Theorem~\ref{thm:main} states that $T(h)$ is ancestral $\vec c$-colorable if and only if $\vec c \prec \vec c_h$. 
Note that it is trivial that $T(h)$ is ancestral $\vec c_h$-colorable (one color for each layer). Equivalently, $T'(h)$, which is obtained from $T(h)$ by adding edges between each node and all of its descendants, is $\vec c_h$-colorable (i.e., it can be colored using $c_i$ Color $i$ so that adjacent vertices have different colors).
A result similar to Theorem~\ref{thm:main} was proved by Folkman and Fulkerson~\cite{FolkFulkerson1969} but for the \textit{edge coloring problem} for a general graph, which states that if a graph $G$ is $\vec c$-colorable and $\vec{c'} \prec \vec c$ then $G$ is also $\vec{c'}$-colorable. However, the technique used in~\cite{FolkFulkerson1969} for edge coloring doesn't work in the context of vertex coloring, which turns out to be much more complicated\footnote{The key idea in~\cite[Thm.~4.2]{FolkFulkerson1969} is to generate an edge colorings with $\vec{c'} \prec \vec c$ by starting from a $\vec c$-coloring and repeatedly applying minor changes to the coloring: an odd-length path consisting of edges in alternating colors (e.g., Red-Blue-Red-Blue-Red) is first identified, and then the colors of these edges are flipped (Blue-Red-Blue-Red-Blue) to increase the number of edges of one color (e.g., Blue) while decreasing that of the other color (e.g., Red) by one. A similar approach for vertex color is to first identify a ``family'' in the tree in which the two children nodes have the same color (e.g., Blue) that is different from the parent node (e.g., Red), and then flip the color between the parent and the children. However,  the existence of such a family is not guaranteed in general. By contrast, in the edge coloring problem, an odd-length path of alternating colors always exists (by examining the connected components of the graph with edges having these two colors).}. 
The same conclusion for vertex coloring, as far as we know, only holds for \textit{claw-free} graphs (see de Werra~\cite{deWerra1985}, and also~\cite[Theorem 4]{Lih2013}). Recall that a claw-free graph is a graph that doesn't have $K_{1,3}$ as an induced subgraph.
However, our graph $T'(h)$ contains a claw when $h\geq 3$, e.g., the subgraph induced by Nodes 2, 4, 10, and 11 (see Fig.~\ref{fig:equitable}). Therefore, this result doesn't cover our case.

\textit{Second}, in the language of majorization, the Color-Splitting Algorithm starts from an integer vector $\vec c$ of dimension $h$ that is majorized by $\vec c_h$ and creates two new integer vectors, both of dimension $h-1$ and majorized by $\vec{c}_{h-1}$. Applying the algorithm recursively produces $2^{h-1}$ sequences of vectors (of decreasing dimensions) where the $i$th element in each sequence is majorized by $\vec{c}_{h-i+1}$.

\begin{figure}[htb!]
    \centering
    \includegraphics[scale=0.39]{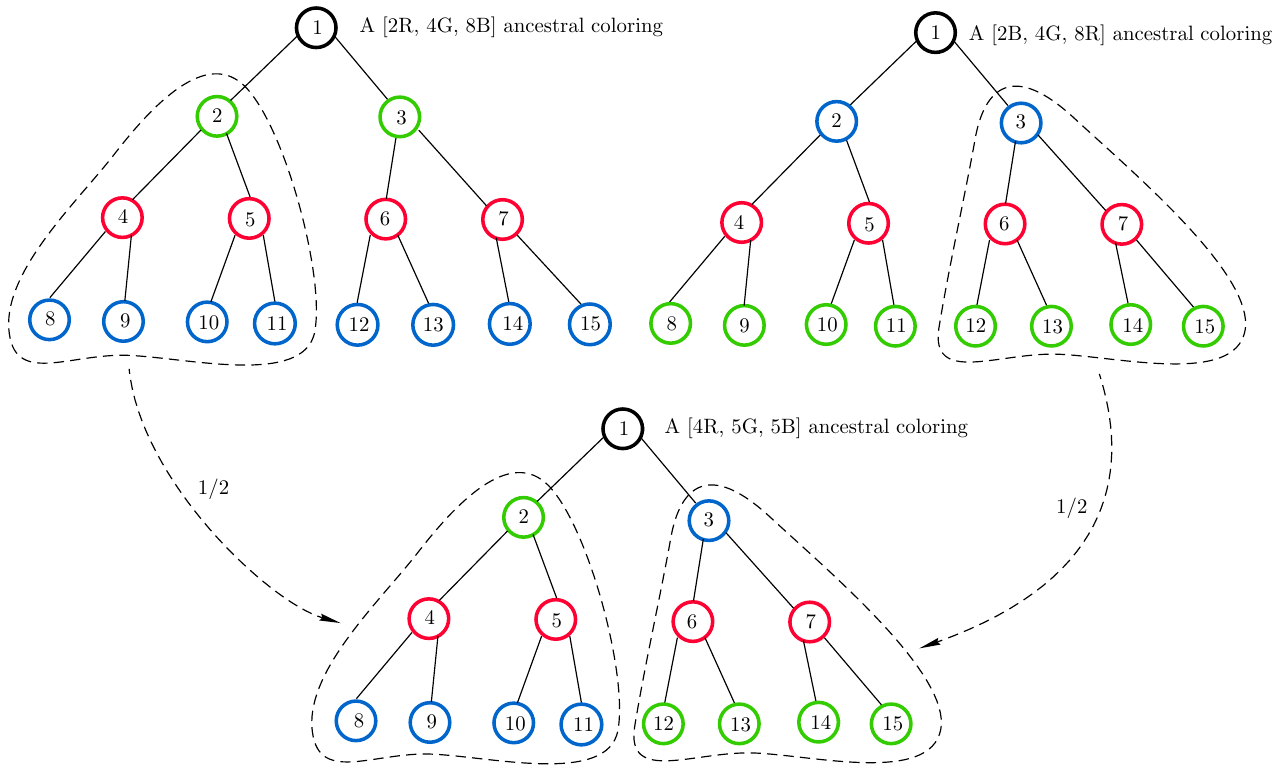}
    \caption{A [4R, 5G, 5B] ancestral coloring for $T_3$ can be obtained by ``stitching'' one half of the [4R, 2G, 8B] (trivial) coloring and one half of the [4R, 8G, 2B] (trivial) coloring. Note that $[4,5,5]=\frac{1}{2}[4,2,8]+\frac{1}{2}[4,8,2]$.}
    \label{fig:stitching1}
\end{figure}

In the remainder of this section, we discuss a potential way to solve the ancestral tree coloring problem using a geometric approach. It was proved by Hardy, Littlewood, and P\'{o}lya in 1929 that $\vec{x}\prec \vec{y}$ if and only if $\vec{x}=\vec{y}P$ for some doubly stochastic matrix $P$. Moreover, the Birkhoff theorem ensures that every doubly stochastic matrix can be expressed as a convex combination of permutation matrices. With these, a geometric interpretation of majorization was established that $\vec{x}\prec \vec{y}$ if and only if $\vec{x}$ lies in the convex hull of all the permutations of $\vec{y}$ \cite{Marshall2011}. Therefore, it is natural to consider a geometric approach to the ancestral coloring problem. The following example is one such attempt.

\begin{example}\label{ex:geometry_h3}
    Let $h=3$. Recall that $\vec{c}_h=[2,4,8]$. Also, let $\vec{c'}=[4,5,5]$ and $\vec{c''}=[3,4,7]$. Clearly, $\vec{c'}\prec \vec{c}_h$ and $\vec{c''}\prec \vec{c}_h$, and hence the CSA would work for both cases. Alternatively, we first write $\vec{c'}$ and $\vec{c''}$ as convex combinations of permutations of $\vec{c}_h$ as
    \begin{align}
        \vec{c'} &= \frac{1}{2}[4,2,8] + \frac{1}{2}[4,8,2], \label{eq:c_1_convex}\\
        \vec{c''} &= \frac{1}{4}[2,4,8] + \frac{1}{4}[2,8,4]+\frac{1}{2}[4,2,8]. \label{eq:c_2_convex}
    \end{align}
\end{example}

\begin{figure}[htb!]
    \centering
    \includegraphics[scale=0.38]{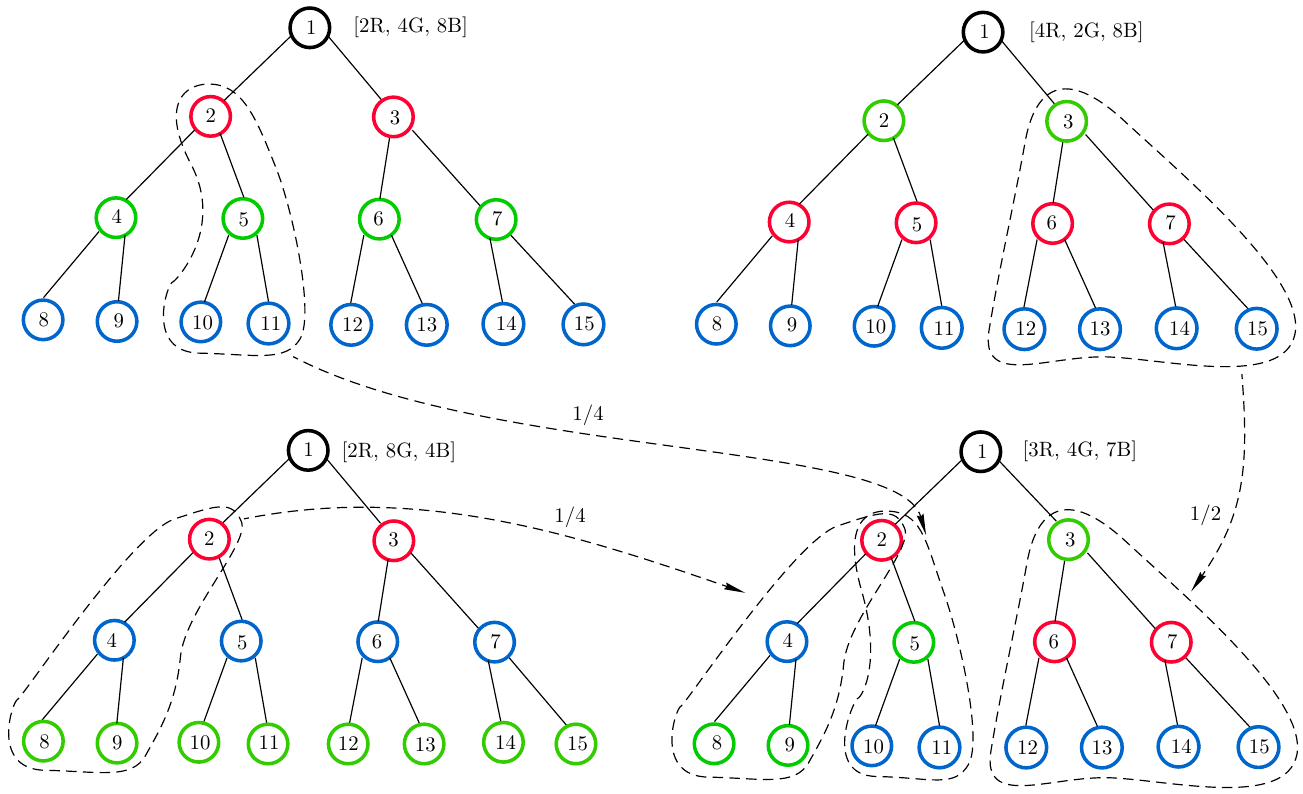}
    \caption{A [3R, 4G, 7B] ancestral coloring for $T_3$ can be obtained by ``stitching'' a quarter of the coloring [2R, 4G, 8B], a quarter of the coloring [2R, 8G, 4B], and a half of the coloring [4R, 2G, 8B]. Note that $[3,4,7]=\frac{1}{4}[2,4,8]+\frac{1}{4}[2,8,4]+\frac{1}{2}[4,2,8]$.}
    \label{fig:stitching2}
\end{figure}

Note that every permutation of $\vec{c}_h$ admits a trivial layer-based ancestral coloring as shown in the two trees at the top in Fig.~\ref{fig:stitching1}. Now, \eqref{eq:c_1_convex} demonstrates that it is possible to obtain an ancestral coloring for $\vec{c'}$ by taking a half of the tree colored by $[4,2,8]$ and a half of that colored by $[4,8,2]$ and stitching them together, which is precisely what we illustrate in Fig.~\ref{fig:stitching1}. Similarly, Fig.~\ref{fig:stitching2} shows how we obtain an ancestral coloring for $\vec{c''}=[3,4,7]$ by taking a quarter of the tree colored by $[2,4,8]$, a quarter of the tree colored by $[2,8,4]$, and a half of the tree colored by $[4,2,8]$ and stitching them together.

With the success of this approach for the above examples ($h=3$) and  other examples with $h=4$ (not included for the sake of conciseness), it is tempting to believe that for every $\vec{c}$ majorized by $\vec{c}_h$, an ancestral coloring can be obtained by stitching sub-trees of trivial layer-based ancestral-colored trees. However, our initial (failed) attempt for $h=5$ indicates that it could be quite challenging to obtain ancestral colorings by such a method, even if a convex and dyadic combination can be identified. Further investigation on this direction is left for future work.

\section{Coloring Sparse Merkle Tree}
\label{app:SMT}

A Sparse Merkle Tree (SMT) (see, e.g.~\cite{DahlbergPP16}) is 
a perfect Merkle tree of enormous size, e.g. $2^{256}$ leaves. It has a unique leaf for every possible output of a cryptographic hash function, hence perfect for storing large data structures such as the state tree of a blockchain (e.g. Ethereum) or any indefinitely growing database. Despite its vast size, it can be efficiently simulated due to its sparsity (most leaves are empty). We provide some sketch of how the problem of coloring Sparse Merkle Tree can be transformed into the problem of coloring \textit{full binary tree} (which is still an open problem) via an example to avoid cumbersome formalization.
We consider a sparse Merkle tree of height 4 in Fig.~\ref{fig:SMT1}, in which the majority of nodes (dashed) are empty. Note that data stored at root nodes of empty sub-trees can be precomputed.  

\begin{figure}[h!]
\centering
\includegraphics[scale=0.35]{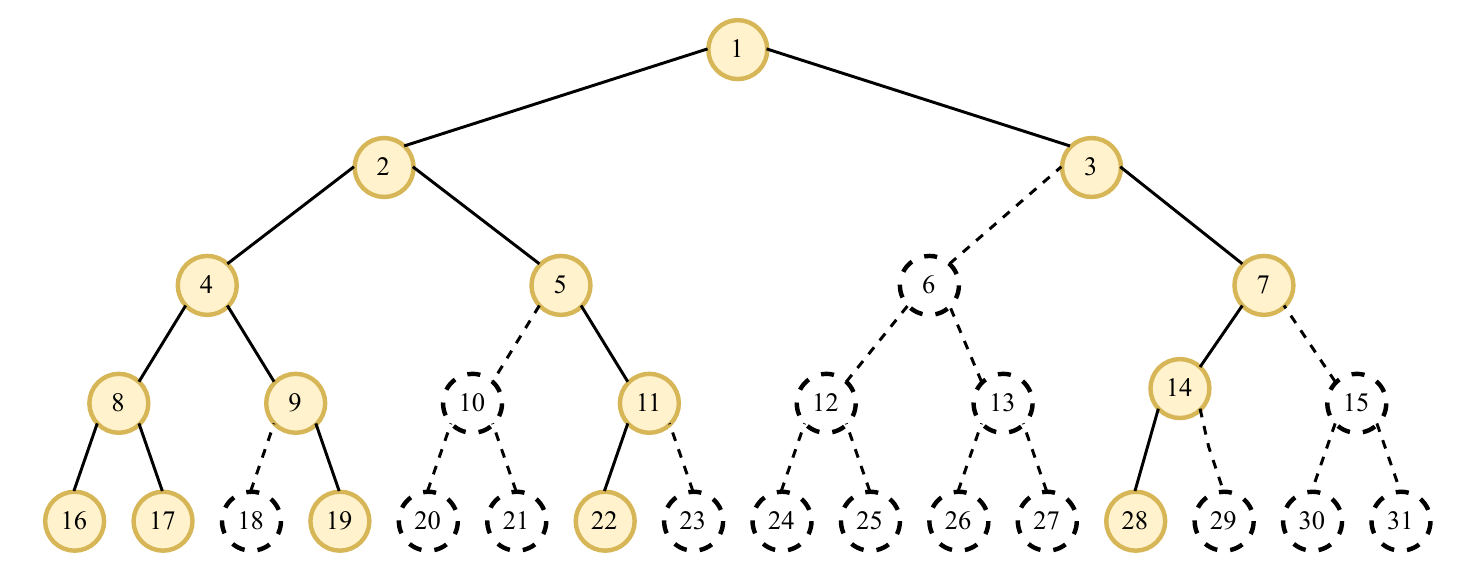}
\caption{An illustration of a Sparse Merkle Tree with five non-empty leaves ($16, 17, 19, 22, 28$). The rest stores a trivial value.}
\label{fig:SMT1}
\end{figure}

As all the dashed nodes can be precomputed, we can ``remove'' them to obtain the tree depicted in Fig.~\ref{fig:SMT2}. 

\begin{figure}[h!]
\centering
\includegraphics[scale=0.35]{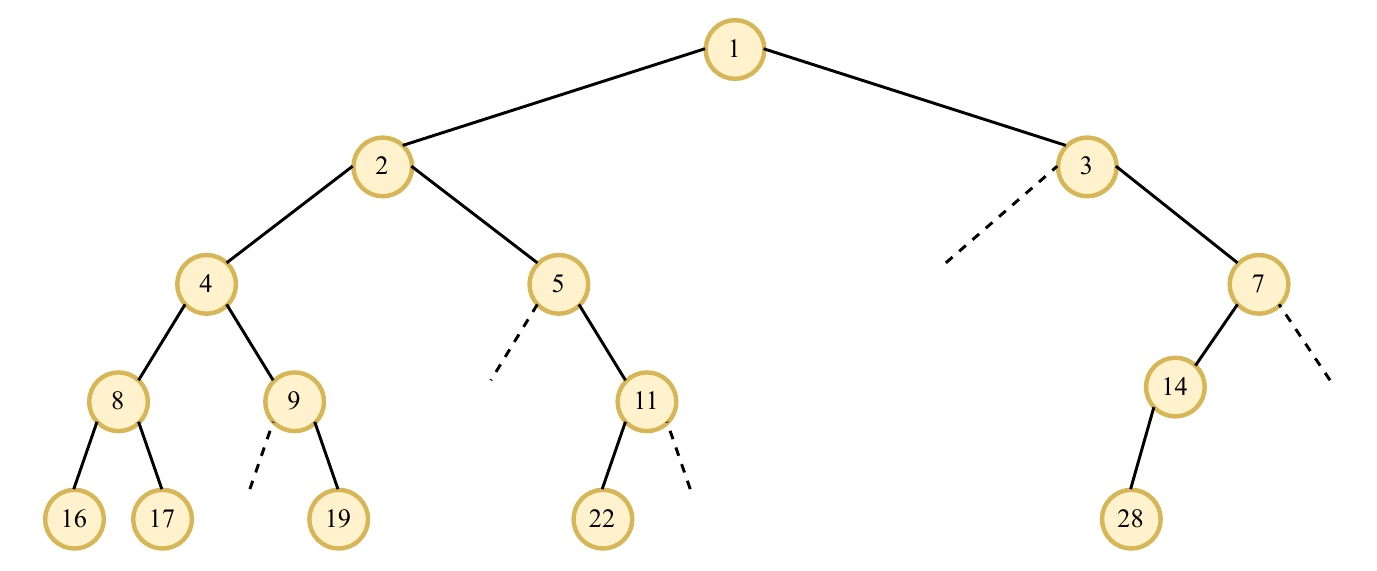}
\caption{The Sparse Merkle Tree after discarding trivial nodes, which can be precomputed.}
\label{fig:SMT2}
\end{figure}

We can compress the tree to turn it into a full binary tree by recursively applying the following rule: if a leaf node does not have a sibling, then it determines the value of its parent node, and hence, can replace the position of the parent (because the value of the parent node depends solely on the value of the leaf and is redundant). The process is repeated until every leaf node has a sibling. The outcome is a full binary tree as in Fig.~\ref{fig:comSMT}. 

\begin{figure}[h!]
\centering
\includegraphics[scale=0.4]{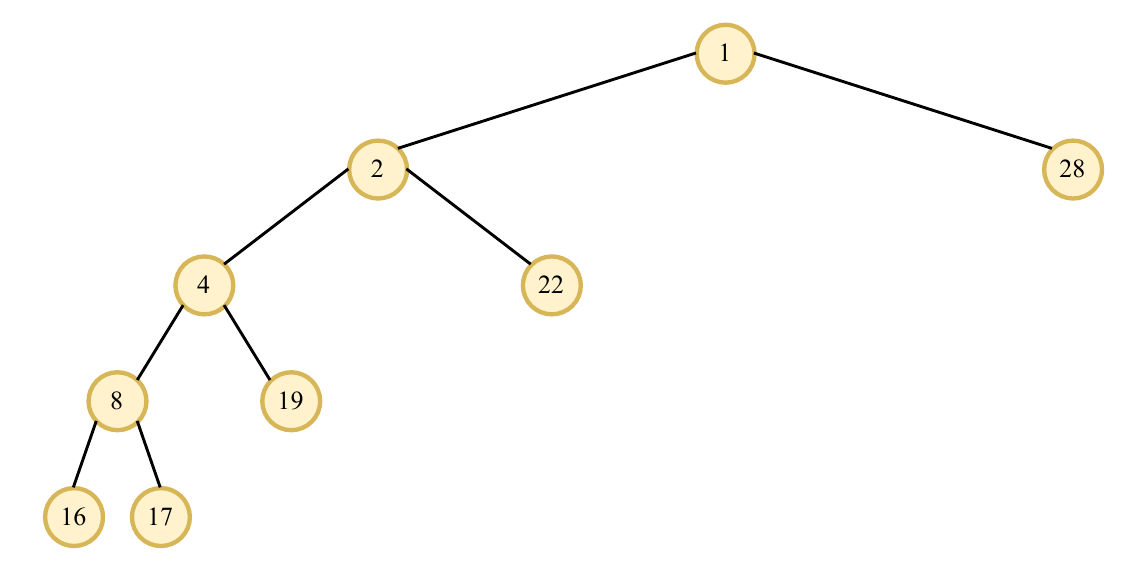}
\caption{Compressing an SMT into a full binary tree.}
\label{fig:comSMT}
\end{figure}

After compressing, we can swap the sibling nodes to make each  Merkle proof into a root-to-leaf path. In this example, a balanced Ancestral Coloring [2R, 2G, 2B, 2P] still exists (see Fig.~\ref{fig:colorSMT}. 
In general, developing a necessary and sufficient condition, if any, for a color sequence to be feasible on a full tree, seems to be a very challenging task. Also, dealing with a growing Sparse Merkle Tree is tricky and requires a new theoretical ground on tree coloring. 

\begin{figure}[h!]
\centering
\includegraphics[scale=0.4]{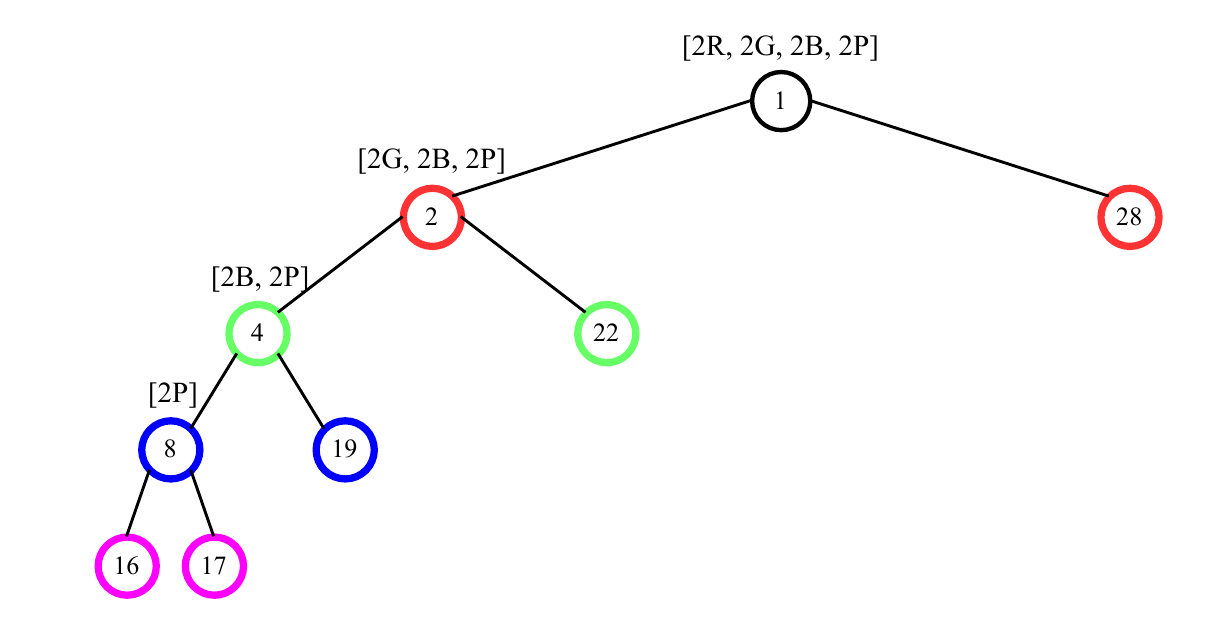}
\caption{A balanced ancestral coloring of the (compressed) full binary tree using the color sequence [2R, 2G, 2B, 2P].}
\label{fig:colorSMT}
\end{figure}

\end{document}